 \documentclass[11pt]{article}
\usepackage{graphicx}%
\usepackage{multirow}%
\usepackage{amsmath,amssymb,amsfonts}\usepackage{amsthm}%
\usepackage{mathrsfs}%
\usepackage[title]{appendix}%
\usepackage{xcolor}%
\usepackage{textcomp}%
\usepackage{manyfoot}%
\usepackage{booktabs}%
\usepackage{float}
\usepackage{authblk}
\usepackage{algorithm}%
\usepackage{algorithmicx}%
\usepackage{algpseudocode}%
\usepackage{listings}%
\usepackage{setspace,bm}
\usepackage{color}
\usepackage{booktabs}
\usepackage{subcaption}
\usepackage{mathtools}
\usepackage[numbers,sort&compress]{natbib}
\usepackage{url}
\usepackage[a4paper,margin=2.5cm]{geometry}
\usepackage{array}
\newcolumntype{L}[1]{>{\raggedright\arraybackslash}p{#1}}
\newcommand{\keywords}[1]{\par\medskip\noindent\textbf{Keywords:} #1\par}
\newcommand{\bmhead}[1]{\subsection*{#1}}
\newcommand{\norm}[1]{\left\lVert#1\right\rVert}
\newcommand{\Om}{\Omega}
\newcommand{\RD}{\mathcal R_D}
\newcommand{\Ezero}{E_0}

\newtheorem{theorem}{Theorem}
\newtheorem{proposition}[theorem]{Proposition}%
\newtheorem{lemma}[theorem]{Lemma}
\newtheorem{corollary}[theorem]{Corollary}
\newtheorem{remark}{Remark}%
\newtheorem{assumption}{Assumption}
\newtheorem{definition}{Definition}%

\title{A Partially Parabolic Damage-Driven
Inflammation Model: Threshold Dynamics and
Linear Spreading}
\author{G. Gambino$^{1}$, F. Gargano$^{2}$, A. Rizzo$^{2}$, V. Sciacca$^{1}$}
\date{}
\begin{document}

	\maketitle
	
\begin{center}
	\small
$^{1}$Department of Mathematics, University of Palermo, Via Archirafi 34, Palermo, 90123, Italy\\
$^{2}$Department of Engineering, University of Palermo, Viale delle Scienze, Ed. 8, Palermo, 90128, Italy
\end{center}


\begin{abstract}
We study a partially parabolic reaction--diffusion--chemotaxis model of damage--inflammation feedback motivated by Duchenne muscular dystrophy. The model couples functional and damaged tissue with effective macrophage activity and a damage-induced inflammatory signal. The invariant tissue simplex bounds the damage source and, together with parabolic estimates, yields global well-posedness and dissipativity in every spatial dimension covered by the regularity assumptions, for arbitrary chemotactic sensitivity. A dimensionless feedback threshold separates resolution to a damage-free dystrophin-deficient state from the existence of a positive damage equilibrium in the spatially homogeneous system. On bounded connected domains, we prove resolution at and below the threshold and, above it, uniform componentwise persistence with spatially uniform lower bounds for initial data with at least one nonzero active component. We derive a candidate linear spreading speed for localized damage and establish a nonlinear upper spreading bound in the absence of chemotaxis. Numerical front speeds agree closely with the linear prediction in the simulations.
\end{abstract}




\section{Introduction}

Duchenne muscular dystrophy (DMD) is a severe X-linked disorder caused by
the absence of functional dystrophin. It is characterized by recurrent muscle
fibre injury, progressive degeneration and loss of functional skeletal muscle
\cite{Duan2021}. Dystrophin contributes to the mechanical stability of the
sarcolemma, and its absence increases the susceptibility of muscle fibres to
contraction-induced damage \cite{Petrof93}. Recurrent injury is accompanied
by a secondary inflammatory response involving inflammatory signalling,
chemokine production and immune-cell recruitment
\cite{Porter02,Porter03,Spencer01}. In particular, CCR2-dependent recruitment of inflammatory monocytes
contributes to macrophage accumulation in dystrophic muscle, while
interference with CCR2 signalling can reduce inflammatory-cell accumulation
and alleviate muscle pathology in experimental DMD models \cite{Mojumdar14}. Inflammatory macrophages can also contribute to
muscle damage and disease severity \cite{Villalta09}. Thus, inflammation is
not the primary cause of DMD, but it can act as an important secondary
mechanism of damage amplification \cite{Dowling23,macrofagi2004}.

Inflammatory cells also participate in debris clearance and muscle repair.
Following muscle injury, recruited monocytes can change their functional
state and contribute to the progression from inflammation to regeneration
\cite{Arnold07,macrofagi2004}. In DMD, however, regeneration does not by
itself correct the dystrophin defect or eliminate the underlying
susceptibility to renewed contraction-induced damage
\cite{Duan2021,Elasbali24}. Dystrophic muscle may consequently undergo
recurrent and incompletely resolved cycles of damage, inflammation and repair.
These processes are also spatially heterogeneous: asynchronous degeneration
and regeneration have been associated with failed regeneration and pathological
remodelling, while imaging studies reveal regional differences in muscle
involvement \cite{Dadgar14,Se20,So26}.

The present work studies the fate of a localized episode of active damage
within the degenerative--regenerative dynamics of DMD, while inflammatory
amplification and tissue repair are both operative. We focus on the
damage--inflammation feedback and do not address the complete clinical
progression of the disease.

Mathematical models of muscular dystrophy have often used ordinary-
differential-equation or compartmental descriptions of interactions among
muscle tissue and immune cells \cite{Hous18}. Predator--prey-type models
provide qualitative descriptions of temporal interactions between immune cells
and muscle tissue \cite{DellAcqua09}. Other compartmental models, including
the model of Jarrah et al.\ \cite{Jarrah14} and the FRiND model
\cite{Hous19}, adopt a damage-mediated viewpoint and couple damaged tissue,
macrophages and regeneration; the FRiND model also distinguishes macrophage
phenotypes. These models describe temporal disease mechanisms, but do not
address the spatial evolution of a localized damage episode.
\citet{lor26} developed kinetic and mean-field models for the statistical
distributions of muscle-fibre and immune-cell densities across a patient
cohort. Spatial approaches have also been developed:
\citet{Virgilio2018} used an agent-based model with histology-derived geometry
to investigate regeneration in healthy and mdx muscle. In a recent preprint,
\citet{Ron26} extended the model of Jarrah et al.\ by incorporating
immune-cell diffusion and numerically investigated its response to localized
damage in one and two spatial dimensions.

Against this background, we study a spatial model of damage--inflammation
feedback. We ask whether a threshold separates resolution from persistence on
bounded domains, and what can be rigorously inferred about the spreading of
localized damage on the whole line.

To address this question, we use a partially parabolic formulation. Functional
and damaged tissue are represented by non-diffusive variables evolving locally,
whereas effective macrophage activity and a damage-induced inflammatory signal
are governed by parabolic transport equations. These describe random motility
and chemotactic transport of the effective macrophage activity in response to
the signal gradient, together with diffusion of the signal. This mixed
structure permits local tissue conversion while allowing inflammatory
activity and signal to move through the tissue.

The model does not introduce separate compartments for macrophage populations
or phenotypes. Instead, the macrophage variable represents the effective
damage-amplifying activity associated with resident macrophages, recruited
monocytes and infiltrating macrophages over the time scale of one active
damage episode, rather than a direct count of all macrophages. Repair-promoting
and resolution-associated functions are not represented by a separate
macrophage population; their aggregate effects are represented indirectly
through effective repair and clearance terms. Likewise, the inflammatory variable is an
aggregate damage-induced signal rather than a single identified chemokine.

The damage-induced signal promotes macrophage recruitment through a
saturating response. The chemotactic flux has a standard
Keller--Segel-type form \cite{HillenPainter,Perthame2006}, with a
phenomenological receptor-type sensitivity that decreases at high signal
levels \cite{LapidusSchiller76}.
Models coupling macrophage chemotaxis, inflammatory signalling and tissue
damage have also been studied in the context of multiple sclerosis
\cite{LBBGPS17}. For a related parabolic--parabolic--ODE system,
\citet{HuFuAi2020} investigated global boundedness and the stability of a
positive equilibrium. In the DMD-motivated model considered here, we study
how a threshold identified in the homogeneous system governs resolution
and uniform componentwise persistence on bounded domains.

The invariant tissue simplex bounds the damage source in the signal
equation; together with parabolic estimates for the signal, this controls the
chemotactic drift and closes the global estimates. Thus the analysis does not
require a smallness assumption on the chemotactic sensitivity.

The analysis of spatial spreading is related to the classical Fisher--KPP
theory of propagation from an unstable state
\cite{Fisher37,KPP,AronsonWeinberger78}. For cooperative systems, equality
between nonlinear spreading speeds and predictions from the linearization
can be established under additional assumptions on the nonlinear dynamics
\cite{WeinbergerLewisLi02}.

Within the reduced model, resolution means disappearance of active damage,
inflammatory signal and effective macrophage activity, together with
convergence of the tissue variables to a quiescent, locally functional
dystrophin-deficient state. Persistence means that the damage--inflammation
loop remains active despite repair and clearance.

Dystrophin does not enter through an explicit molecular variable. The
increased susceptibility associated with dystrophin deficiency is represented
phenomenologically through an effective damage-susceptibility parameter, so
changes in that parameter describe changes in the reduced tissue response
rather than direct molecular effects of dystrophin loss. The initial data
prescribe the spatial distribution of tissue damage and inflammatory activity
at the start of the episode.

\subsection{Main contributions}
\label{subsec:main_contributions}

The main contributions of this work are organized around four complementary
aspects of the model: threshold dynamics, global analysis, persistence and
spatial spreading.

\begin{enumerate}
	\item \emph{Threshold dynamics.} We identify a dimensionless
	damage-amplification threshold for the local damage--inflammation feedback.
	At and below the threshold, every admissible spatially homogeneous solution
	converges to the damage-free equilibrium. Above it, a unique positive damage
	equilibrium exists and is locally asymptotically stable with respect to
	homogeneous perturbations.
	
	\item \emph{Global analysis.} For the partially parabolic system on bounded
	connected domains,
	we establish positivity, global well-posedness and dissipative estimates
	under suitable regularity assumptions. The global theory holds in every
	spatial dimension covered by the functional setting and for every
	$\chi_0\ge0$.
	
	\item \emph{Long-time spatial dynamics.} On bounded connected domains, we prove
	resolution of active damage at and below the feedback threshold and
	persistence of the active damage--inflammation loop above it for nontrivial
	active initial data, with persistence constants uniform over all admissible
	active initial states.
	We obtain exponential decay estimates below threshold and an algebraic
	estimate at the threshold.
	The persistence result is componentwise and yields spatially uniform lower
	bounds, in the essential-infimum sense for the nondiffusing damage variable;
	thus, the threshold identified from the homogeneous system retains
	its dynamical meaning in the spatial model, including the chemotactic
	coupling.
	
	\item \emph{Leading-edge spreading.} Linearization at the damage-free state
	excludes diffusion-driven instability of a stable homogeneous state and
	oscillatory onset. We derive the leading-edge spectral problem for localized
	damage and obtain a candidate linear spreading speed.
	In the absence of chemotaxis, for compactly supported active data on the
	whole line, we also prove a nonlinear upper bound for the spreading of
	active damage. Numerical simulations are consistent with the
	linear prediction in the tested regimes, but the identification of the
	linear speed with the speed selected by the nonlinear system remains an
	open problem. In the numerical experiments with nonzero chemotaxis, the measured front speed changes little across the tested sensitivities, even when the spatial
	maximum of the chemotactic flux magnitude is comparable to or larger than
	that of the diffusive flux.
\end{enumerate}

\subsection{Organization of the paper}
\label{subsec:organization}

The dimensional and dimensionless model is formulated in the remainder of
the Introduction. Section~2 analyzes the homogeneous equilibria and
their local stability. Section~3 establishes positivity, global
well-posedness and dissipativity for the spatial system. Section~4 proves
extinction at and below the damage-amplification threshold, while
Section~5 establishes persistence above threshold. Section~6 studies the
leading-edge linearization, the absence of diffusion-driven instability at
the damage-free equilibrium, and the candidate linear spreading speed; the nonlinear
upper bound and the pulled-front speed conjecture are distinguished there.
Section~7 presents numerical experiments on threshold dynamics, front
propagation and chemotaxis. Section~8 summarizes the main
conclusions and limitations. The
appendices contain the technical parabolic estimates,
continuous-dependence arguments and decay lemmas used in the proofs.

\subsection{Model formulation and nondimensionalization}
\label{subsec:dimensional}

We describe the local interaction among functional muscle tissue, active
tissue damage, macrophage-mediated inflammatory activity and a soluble
damage-induced signal. Let $\Omega_{\rm phys}\subset\mathbb R^n$ denote the
physical tissue domain, with spatial coordinate $\xi$, and let $\tau$ denote
dimensional time. We use volume fractions for the tissue variables. The state
variables are:
\begin{center}
\begin{tabular}{@{}l@{\,:\,}p{0.82\linewidth}@{}}
$H(\tau,\xi)$ & $\ $ functional muscle-tissue fraction;\\
$D(\tau,\xi)$ & $\ $ actively damaged or degenerating tissue fraction;\\
$M(\tau,\xi)$ & $\ $ effective macrophage-activity density in a damage-amplifying state;\\
$C(\tau,\xi)$ & $\ $ aggregate damage-induced inflammatory and chemotactic signal.
\end{tabular}
\end{center}
The variable $D$ represents biologically active damage rather than a cumulative measure of irreversible tissue loss.

Let $K_{\rm cap}\in(0,1]$ denote the constant effective tissue-capacity fraction within a reference local volume.
We assume that initially
\[
H\ge0,\qquad D\ge0,\qquad H+D\le K_{\rm cap},\qquad M\ge0,\qquad C\ge0.
\]
The dimensional model is:
\begin{equation}\label{eq:model}
	\begin{cases}
		\displaystyle
		\partial_\tau H = s_{\rm rec} \,(K_{\rm cap}-H-D) + r_H\, D - a_0\, M H,\\[1em]
		\displaystyle
		\partial_\tau D = a_0\, M H - (\mu_D +r_H) \, D,\\[1em]
		\displaystyle
		\partial_\tau M = D_M \Delta M
		- \nabla\cdot\left(\chi(C) M\,\nabla C\right)
		+ r_M\,\frac{C}{k_M+C} - \mu_M M,\\[1em]
		\displaystyle
		\partial_\tau C = D_C \Delta C + r_C\, D - \mu_C\, C\,
	\end{cases}
\end{equation}
where:
\begin{equation}\label{eq:dimensional_sensitivity}
	\chi(C)= \bar\chi\,\frac{k_\chi^2}{(k_\chi+C)^2}.
\end{equation}
All kinetic, removal, diffusion and scale parameters other than $\bar\chi$
are positive. We allow
$\bar\chi\ge0$, with $\bar\chi=0$ representing the limiting case without
chemotaxis. In that limiting case $k_\chi$ is inoperative, but we retain
$k_\chi>0$ for notational convenience and to keep the sensitivity law and
its estimates in a uniform form. We impose homogeneous Neumann conditions on $M$ and $C$. The condition
$\partial_{\mathbf n}M=0$ eliminates the normal diffusive flux, while
$\partial_{\mathbf n}C=0$ eliminates the normal chemotactic flux. Hence
the total macrophage flux has zero normal component at the boundary. The
tissue variables do not diffuse and therefore require no boundary conditions.
The resulting system is partially parabolic: it couples pointwise ODEs for
the tissue variables with parabolic equations for macrophage activity and
signal, including chemotactic transport of the former.

\paragraph{Tissue damage and recovery.}
The bilinear term $a_0MH$ transfers tissue from the functional compartment to
the actively damaged compartment. It represents secondary, macrophage-mediated
injury and appears with opposite signs in the first two equations, so it does
not by itself change the total accounted tissue $H+D$. This is the simplest
monotone, leading-order description of exposure of susceptible functional
tissue to an effective cytotoxic macrophage-activity density. The model does not
attribute the primary dystrophin defect or the initiating mechanical lesion to
macrophages: an initiating mechanical lesion can be represented by the initial
damage distribution $D(0,\cdot)$, while $a_0MH$ describes subsequent inflammatory amplification
\cite{Petrof93,Villalta09}. The model includes no further source of
mechanically induced damage during the episode.

The two terms that remove $D$ have different destinations and therefore should
not be merged. The term $r_HD$ is an effective transfer from $D$ back to $H$;
it represents recovery of repairable damage and local regeneration while the
structural scaffold is retained \cite{VrackoBenditt72}. At the compartment
level, this does not assert that an individual necrotic fibre literally
returns to its previous state. By contrast, the term $\mu_DD$ accounts for the clearance or loss of
actively damaged tissue without transferring it directly to $H$. This
process increases the available capacity $K_{\rm cap}-H-D$, which supports
the reconstruction of functional tissue at rate
$s_{\rm rec}(K_{\rm cap}-H-D)$. Indeed, summing the first two equations gives
\begin{equation}\label{eq:dimensional_total_tissue}
	\partial_\tau(H+D)=s_{\rm rec}(K_{\rm cap}-H-D)-\mu_DD.
\end{equation}
The term $s_{\rm rec}(K_{\rm cap}-H-D)$ is a coarse-grained restoration law
rather than a mechanistic description of satellite-cell proliferation. It
assumes that sufficient regenerative capacity remains available to support
tissue reconstruction throughout the episode, including during sustained
damage. Macrophages are
known to participate both in tissue injury and in the coordination of debris
clearance and regeneration \cite{Arnold07,macrofagi2004}. In the model, macrophage-mediated damage is represented by $a_0MH$,
whereas direct repair and tissue reconstruction are represented by $r_HD$
and $s_{\rm rec}(K_{\rm cap}-H-D)$, respectively.

\paragraph{Damage-induced recruitment and inflammatory amplification.}
Actively damaged tissue produces signal at rate $r_CD$. A linear production
law is the first-order approximation for an aggregate signal generated by the
current damage burden; the loss term $\mu_CC$ represents degradation, uptake
and local clearance. This construction enforces the damage-driven character of
the model: in the absence of $D$, there is no source of $C$. Chemokine
upregulation and CCR2-dependent recruitment of inflammatory monocytes and
macrophages are well documented in dystrophic and injured muscle
\cite{Porter02,Porter03,Mojumdar14,Martinez10}.

The source
\[
r_M\frac{C}{k_M+C}
\]
represents recruitment or activation from a reservoir of recruitable
macrophage activity rather than autonomous proliferation. It vanishes when
$C=0$, has initial slope $r_M/k_M$,
and saturates at the maximal rate $r_M$; $k_M$ is its half-saturation signal
level. Saturation accounts phenomenologically for finite recruitment capacity,
limited receptor occupancy and vascular or cellular supply. Macrophage
activity density is lost at rate $\mu_MM$. Consequently, inflammatory activity is not
self-sustained in the absence of damage-generated signal.

The recruitment term is independent of $M$; a logistic correction in $M$ is
not considered here.

Together, the local reactions retain the positive feedback loop
\[
D\longrightarrow C\longrightarrow M\longrightarrow D.
\]
Direct repair and clearance remove active damage, while reconstruction
replenishes functional tissue that can subsequently be damaged. The interplay
of these processes underlies the threshold analysed below.

\paragraph{Spatial transport.}
Functional and damaged tissue are treated as locally fixed on the time scale
of interest. The effective macrophage-activity density is transported by
random motility, represented by $D_M\Delta M$, and by directed migration up
gradients of the damage-induced signal. In flux form,
\[
J_M=-D_M\nabla M+\chi(C)M\nabla C,
\]
so that $-\nabla\cdot J_M$ gives the two transport terms in the equation for
the effective macrophage-activity density. The receptor-type sensitivity
\eqref{eq:dimensional_sensitivity} is nonincreasing for nonnegative signal;
when $\bar\chi>0$, it is largest at $C=0$ and decreases strictly as $C$
increases
\cite{LapidusSchiller76}. The parameter $k_\chi$ sets the
concentration scale over which this attenuation occurs; when $\bar\chi>0$,
$\chi(k_\chi)=\bar\chi/4$, so $k_\chi$ is not a half-sensitivity level.
The exact
functional form is a phenomenological modelling choice rather than a
DMD-specific measured law. Finally, $D_C\Delta C$ represents local spreading
of the soluble signal.

Table~\ref{tab:dimensional_parameters} lists the dimensional parameters
and their units. Here $\mathsf M$, $\mathsf S$, $\mathsf T$ and $\mathsf L$
denote units of macrophage-activity density, signal, time and length,
respectively; $H$, $D$ and $K_{\rm cap}$ are dimensionless fractions.

\begin{table}[!htbp]
\centering
\caption{Dimensional parameters of the model. The status column describes
their role in the reduced formulation;
none of the effective coefficients is directly estimated from a DMD cohort.}
\label{tab:dimensional_parameters}
\setlength{\tabcolsep}{4pt}
\begin{tabular}{@{}L{0.11\linewidth}L{0.17\linewidth}L{0.40\linewidth}L{0.22\linewidth}@{}}
	\toprule
	Parameter & Unit & Meaning in the model & Status \\
	\midrule
	$K_{\rm cap}$ & dimensionless & Local tissue-capacity fraction: the upper bound for $H+D$, also used to normalize the tissue variables. & Assumed capacity \\
	$s_{\rm rec}$ & $\mathsf T^{-1}$ & Rate constant for reconstruction of functional tissue from available capacity. & Effective rate \\
	$r_H$ & $\mathsf T^{-1}$ & Rate constant for direct transfer from active damage to functional tissue. & Effective rate \\
	$\mu_D$ & $\mathsf T^{-1}$ & Clearance or loss rate constant for actively damaged tissue. & Effective rate \\
	$a_0$ & $(\mathsf M\mathsf T)^{-1}$ & Susceptibility coefficient in the macrophage-mediated damage term $a_0MH$. & Phenomenological \\
	$D_M$ & $\mathsf L^2\mathsf T^{-1}$ & Random motility coefficient of the effective macrophage-activity density. & Effective transport \\
	$\bar\chi$ & $\mathsf L^2(\mathsf S\mathsf T)^{-1}$ & Chemotactic sensitivity at zero signal. & Phenomenological \\
	$k_\chi$ & $\mathsf S$ & Signal scale controlling the decrease in chemotactic sensitivity. & Phenomenological \\
	$r_M$ & $\mathsf M\mathsf T^{-1}$ & Limiting maximal rate of signal-induced recruitment or activation of effective macrophage activity. & Effective rate \\
	$k_M$ & $\mathsf S$ & Half-saturation signal level for recruitment or activation. & Effective scale \\
	$\mu_M$ & $\mathsf T^{-1}$ & Turnover or loss rate constant of the effective macrophage-activity density. & Effective rate \\
	$D_C$ & $\mathsf L^2\mathsf T^{-1}$ & Diffusion coefficient of the damage-induced inflammatory signal. & Effective transport \\
	$r_C$ & $\mathsf S\mathsf T^{-1}$ & Signal-production coefficient per unit damaged-tissue fraction. & Effective rate \\
	$\mu_C$ & $\mathsf T^{-1}$ & Clearance or degradation rate constant of the inflammatory signal. & Effective rate \\
	\bottomrule
\end{tabular}
\end{table}

\paragraph{Nondimensionalization.}

We scale time with the effective macrophage-activity turnover time, space
with the corresponding motility length, tissue with the local capacity, and
the macrophage-activity density with its maximal recruitment--turnover level;
the signal is scaled with the half-saturation level
for macrophage recruitment:
\begin{align*}
	t = \mu_M \tau, \quad  x= \sqrt{\dfrac{\mu_M}{D_M}} \,\xi,
\quad h=\dfrac{H}{K_{\rm cap}}, \quad d=\dfrac{D}{K_{\rm cap}}, \quad
	m=\dfrac{\mu_M M}{r_M}, \quad c=\dfrac{C}{k_M}.
\end{align*}
Thus $\mu_M^{-1}$ is the reference time and
$\sqrt{D_M/\mu_M}$ is the characteristic motility length over one effective
macrophage-activity turnover time. The scale $r_M/\mu_M$ is the
activity level approached under sustained, fully saturated recruitment when
transport is neglected. The physical domain is transformed into
$\Omega=\sqrt{\mu_M/D_M}\,\Omega_{\rm phys}$.

The differential operators transform according to
\[
\partial_\tau=\mu_M\partial_t,\qquad
\nabla_\xi=\sqrt{\frac{\mu_M}{D_M}}\,\nabla_x,\qquad
\Delta_\xi=\frac{\mu_M}{D_M}\,\Delta_x.
\]
For the nontrivial chemotactic coefficient, one obtains
\[
\frac{k_M}{D_M}\chi(k_Mc)
=\frac{\bar\chi k_M}{D_M}
\frac{(k_\chi/k_M)^2}{(k_\chi/k_M+c)^2},
\]
which verifies the scaling of the chemotactic term below.

The resulting system is
\begin{equation}\label{eq:dimensionless}
	\begin{cases}
		\displaystyle
		\partial_t h = \sigma\,(1-h-d) + \rho\, d - \alpha\, m h,\\[0.7em]
		\displaystyle
		\partial_t d = \alpha \, m h - (\delta+\rho)\, d,\\[0.7em]
		\displaystyle
		\partial_t m = \Delta m
		- \chi_0\,\nabla\cdot\left(\frac{\kappa^2}{(\kappa+c)^2}\, m\,\nabla c\right)
		+ \frac{c}{1+c} - m,\\[0.7em]
		\displaystyle
		\partial_t c = D_c\,\Delta c + r\, d - \mu\, c,
	\end{cases}
\end{equation}
with homogeneous Neumann conditions for $m$ and $c$, and
\begin{equation}\label{eq:groups}
	\begin{aligned}
		&\sigma=\frac{s_{\rm rec}}{\mu_M},\qquad
		\rho=\frac{r_H}{\mu_M},\qquad
		\delta=\frac{\mu_D}{\mu_M},\qquad
		\mu=\frac{\mu_C}{\mu_M},\qquad
		D_c=\frac{D_C}{D_M},\\[0.4em]
		&\alpha=\frac{a_0 \, r_M}{\mu_M^{2}},\qquad
		r=\frac{r_C\,K_{\rm cap}}{\mu_M\,k_M},\qquad
		\chi_0=\frac{\bar\chi\,k_M}{D_M},\qquad
		\kappa=\frac{k_\chi}{k_M}.
	\end{aligned}
\end{equation}

\paragraph{PDE setting and structural assumptions.}
For the bounded-domain analysis, $\Omega\subset\mathbb R^n$, $n\ge1$, is a
bounded connected domain with boundary of class $C^{2+\eta_0}$ for some
$0<\eta_0<1$.
The admissible initial data satisfy
\[
h_{in}\ge0,\qquad d_{in}\ge0,\qquad h_{in}+d_{in}\le1,
\qquad m_{in}\ge0,\qquad c_{in}\ge0,
\]
with the regularity required for strong and classical solutions specified in
Section~\ref{wp:main_section}. Throughout the analysis,
\[
\sigma,\rho,\alpha,\delta,D_c,r,\mu,\kappa>0,
\qquad \chi_0\ge0;
\]
the case $\chi_0=0$ is the chemotaxis-free limiting case. The upper spreading
bound in Section~\ref{sec:invasion} is instead established on the whole line.

The dimensionless groups have direct mechanistic interpretations. The
parameters $\sigma$, $\rho$ and $\delta$ compare, respectively, homeostatic
reconstruction, direct repair and clearance of active damage with effective
macrophage-activity turnover. Hence $\delta+\rho$ is the total exit rate from the active-damage
compartment, while only $\rho$ returns tissue directly to $h$. The parameter
$\alpha$ measures macrophage-mediated damage at the activity-density scale
$r_M/\mu_M$, relative to the macrophage-activity turnover rate. At a spatially
homogeneous equilibrium, $r/\mu$ is the proportionality factor between signal
and active damage. The parameters $D_c$ and $\kappa$ compare signal diffusion
with effective macrophage-activity motility, and the chemotactic attenuation
scale with the recruitment half-saturation scale, respectively. The parameter
$\chi_0$ scales the low-signal chemotactic sensitivity by the recruitment
signal scale and the random motility coefficient.

The gain of the linear damage--signal--macrophage-activity feedback at the
damage-free state is defined in Section~\ref{sec:local} as the damage
reproduction number $\mathcal R_D$. The reconstruction rate $\sigma$ does
not enter this linear gain, whereas $\rho$ contributes to the removal of
active damage. Signal diffusion affects spatial modes and the leading-edge
calculation; chemotaxis does not enter the linearization at the damage-free
state, but remains present in the nonlinear spatial dynamics.

\section{Homogeneous equilibria and local stability}
\label{sec:local}
This section identifies the spatially homogeneous equilibria of the model and
determines their local stability within the homogeneous reaction system. The
main objective is to establish how the damage-amplification number $\mathcal R_D$
separates two regimes: resolution towards the damage-free equilibrium $E_0$,
and persistence through a unique positive equilibrium $E_+$. We first specify
the admissible state space, then classify the equilibria and analyze their
stability. All stability statements in this section concern homogeneous
perturbations; they do not assert global convergence of the ODE system or
stability of the PDE with respect to spatially heterogeneous perturbations.
Global resolution, spatial persistence, and the effect of nonzero spatial
modes are treated later.

We focus on the reaction terms of \eqref{eq:dimensionless}, namely on the
corresponding spatially homogeneous dynamical system:
\begin{equation}
	\label{eq:ode}
	\begin{cases}
		\displaystyle
		\dot{h} =
		 \sigma\,(1-h-d) + \rho\, d - \alpha\, m h,
		\\[1em]
		\displaystyle
		\dot{d}=
		 \alpha \, m h - (\delta+\rho)\, d,
		\\[1em]
		\displaystyle
		\dot{m} =
		\frac{c}{1+c} - m,
		\\[0.5em]
		\displaystyle
		\dot{c} = r d - \mu c.
	\end{cases}
\end{equation}

\subsection{Invariant region and global bounds}

We first show that the spatially homogeneous system preserves the biologically admissible state space. In particular, functional and damaged tissue remain nonnegative and cannot exceed the total available tissue fraction, while effective macrophage activity and inflammatory signal remain nonnegative.

\begin{proposition}[Invariant region and preliminary bounds]\label{prop:invariant_region}
	Consider system \eqref{eq:ode}. Then the set:
\begin{equation}\label{inv_set}
	\mathcal{D}=\{(h,d,m,c)\in\mathbb{R}^4:\ h\ge 0,\ d\ge 0,\ m\ge 0,\ c\ge 0,\ h+d\le 1\}
	\end{equation}
	is positively invariant. Moreover, for every $K_m\ge1$ and $K_c\ge r/\mu$ the compact set
	\[
	\mathcal{D}_{K_m,K_c}=\mathcal{D}\cap\{m\le K_m,\ c\le K_c\}
	\]
	is positively invariant as well.
\end{proposition}

\begin{proof}
	We verify that the vector field is inward pointing or tangent on each boundary component of $\mathcal{D}$.
	
	\emph{(i) Boundary $h=0$.} In this case,
	\[
	\dot h = \sigma(1-d)+\rho d> 0,
	\]
	Since $h\ge0$, $d\ge0$, and $h+d\le1$ on $\mathcal{D}$, we have
	$0\le d\le1$. Hence the vector field points inward to the boundary $h=0$.
	
	\emph{(ii) Boundary $d=0$.} In this case,
	\[
	\dot d=\alpha mh \ge 0,
	\]
	so the vector field points inward or is tangent to the boundary $d=0$.
	
	\emph{(iii) Boundary $m=0$.} In this case,
	\[
	\dot m=\dfrac{c}{1+c} \ge 0,
	\]
	hence trajectories cannot cross the boundary $m=0$ toward negative values.
	
	\emph{(iv) Boundary $c=0$.} In this case,
	\[
	\dot c = rd \ge 0,
	\]
	so trajectories cannot cross the boundary $c=0$ toward negative values.
	
	\emph{(v) Boundary $h+d=1$.} Let $w=h+d$. Summing the first two equations in \eqref{eq:ode} yields:
	\begin{equation}
	\dot w = \sigma(1-w)-\delta d. 
    \label{eq:w_reaction}
	\end{equation}
	Therefore, on the boundary $w=1$,
	\[
	\dot w|_{w=1}=-\delta d\le 0.
	\]
	Thus trajectories cannot cross the boundary $h+d=1$ outward.
	
	Since the vector field points inward or is tangent on every component of the boundary of
	$\mathcal{D}$, the set $\mathcal{D}$ is positively invariant.\\
    For the second statement, on $\{m=K_m\}$ we have $\dot m=c/(1+c)-K_m\le1-K_m\le0$, and on $\{c=K_c\}$ we have $\dot c=rd-\mu K_c\le r-\mu K_c\le0$, using $d\le1$.
\end{proof}

\begin{proposition}[Global bounds for the reaction system]\label{cor:global}
	Every solution of \eqref{eq:ode} with initial condition in $\mathcal{D}$ exists for all $t\ge0$,
	remains in $\mathcal{D}$ and is bounded; more precisely,
	\[
	m(t)\le\max\{1,m(0)\},\qquad c(t)\le\max\{r/\mu,\,c(0)\}\qquad\text{for all }t\ge0 .
	\]
	Moreover
	\[
	\limsup_{t\to\infty} c(t)\le\frac{r}{\mu},
	\qquad
	\limsup_{t\to\infty} m(t)\le\frac{r}{r+\mu}<1 .
	\]
\end{proposition}

\begin{proof}
	Let $K_m=\max\{1,m(0)\}$ and $K_c=\max\{r/\mu,c(0)\}$. The initial datum belongs to
	$\mathcal{D}_{K_m,K_c}$, which is positively invariant by Proposition~\ref{prop:invariant_region};
	hence the solution stays in a compact set, so it is global and bounded, and the stated pointwise
	bounds hold.

	Since $d\le1$ on $\mathcal{D}$, we have $\dot c\le r-\mu c$, and the comparison principle gives
	$\limsup_{t\to\infty}c\le r/\mu$. Let $\varepsilon>0$. Since $s\mapsto s/(1+s)$ is increasing,
	there is $t_\varepsilon$ such that
	\[
	\dot m\le\frac{r/\mu+\varepsilon}{1+r/\mu+\varepsilon}-m\qquad\text{for }t\ge t_\varepsilon ,
	\]
	whence $\limsup_{t\to\infty}m\le(r/\mu+\varepsilon)/(1+r/\mu+\varepsilon)$. Letting
	$\varepsilon\to0$ gives the last bound.
\end{proof}
These bounds provide the compactness and boundedness properties needed below
to classify the equilibria and analyze their local stability.
\subsection{Equilibria and damage-amplification threshold}

We now classify the spatially homogeneous equilibria of \eqref{eq:ode} and
determine their local stability within the homogeneous reaction system. The
analysis is organized around the damage reproduction number $\mathcal R_D$,
which measures the gain of the damage--signal--macrophage-activity-density feedback loop.

We first analyze the linear stability of the damage-free equilibrium $E_0$ within the homogeneous reaction system and
identify the threshold $\mathcal R_D=1$. Below threshold $E_0$ is locally
asymptotically stable and no persistent-damage equilibrium exists; at the
threshold it is nonhyperbolic; above threshold $E_0$ is unstable and the
persistent-damage equilibrium $E_+$ exists and is locally asymptotically
stable.

The global resolution result at and below threshold is stated separately
after the local stability analysis. 

\subsubsection{Damage-free equilibrium and threshold}

The system always admits the damage-free equilibrium
\begin{equation}\label{DF_eq}
	E_0=(1,0,0,0).
\end{equation}
This equilibrium represents a quiescent, locally functional
dystrophin-deficient state after resolution of active damage, rather than
genetically healthy muscle.

The Jacobian matrix of \eqref{eq:ode} at $E_0$ is
\[
J(E_0)=
\begin{pmatrix}
	-\sigma & -\sigma+\rho & -\alpha & 0\\[0.3em]
	0 & -(\delta+\rho) & \alpha & 0\\[0.3em]
	0 & 0 & -1 & 1\\
	0 & r & 0 & -\mu
\end{pmatrix}.
\]
The $h$-component contributes the strictly negative eigenvalue
$-\sigma$. Consequently, the stability of $E_0$ is determined by the
active $(d,m,c)$ subsystem, whose Jacobian is
\begin{equation}
	J_{\mathrm{act}}^0=
	\begin{pmatrix}
		-(\delta+\rho)&\alpha&0\\
		0&-1&1\\
		r&0&-\mu
	\end{pmatrix}.
	\label{Jacobianoridotto}
\end{equation}
Its characteristic polynomial is
\[
P_0(\lambda)
=\lambda^3+\pi_2\lambda^2+\pi_1\lambda+\pi_0,
\]
where
\[
\pi_2=\delta+\rho+1+\mu,
\qquad
\pi_1=(1+\mu)(\delta+\rho)+\mu,
\qquad
\pi_0=\mu(\delta+\rho)-\alpha r.
\]

The linearization identifies the following dimensionless threshold parameter
for this feedback loop.

\begin{definition}[Damage reproduction number of the reduced damage-amplification loop]\label{def:RD}
	\begin{equation}\label{eq:RD}
		\mathcal{R}_D
		:=\frac{\alpha r}{(\delta+\rho)\mu}
		=\frac{a_0\,r_M\,r_C\,K_{\mathrm{cap}}}
		{\mu_M\,\mu_C\,k_M\,(\mu_D+r_H)}.
	\end{equation}
\end{definition}
The quantity $\mathcal R_D$ is the linear amplification factor of the
damage--signal--macrophage-activity feedback loop near the damage-free
equilibrium. It combines the production and removal processes along the
cycle
$
D\longrightarrow C\longrightarrow M\longrightarrow D.
$
Thus, $\mathcal R_D<1$ corresponds to net attenuation of small
perturbations, whereas $\mathcal R_D>1$ corresponds to net amplification.
The threshold $\mathcal R_D=1$ therefore marks the change from local
resolution to linear instability of the damage-free state.

\begin{proposition}\label{prop:E0}
	The damage-free equilibrium $E_0$ is locally asymptotically stable if
	$\mathcal R_D<1$ and unstable if $\mathcal R_D>1$. Moreover, the spectral
	bound $\lambda_{\max}(J_{\mathrm{act}}^0)$ is a simple real eigenvalue with a strictly
	positive eigenvector, and
	\[
	\operatorname{sign}\lambda_{\max}(J_{\mathrm{act}}^0)
	=\operatorname{sign}(\mathcal R_D-1).
	\]
\end{proposition}

\begin{proof}
	Since
	\[
	\pi_0=(\delta+\rho)\mu(1-\mathcal R_D),
	\]
	the sign of the spectral bound changes at $\mathcal R_D=1$.
	
	Suppose first that $\mathcal R_D>1$. Then $P_0(0)<0$, whereas
	$P_0(\lambda)\to+\infty$ as $\lambda\to+\infty$. Hence $P_0$ has a
	positive real root $\lambda_+$. It follows directly from the definition of
	the spectral bound that
	\[
		\lambda_{\max}(J_{\mathrm{act}}^0)\ge\lambda_+>0.
	\]
	
	If $\mathcal R_D<1$, then $\pi_0,\pi_1,\pi_2>0$ and
	\[
	\pi_2\pi_1
	>(\delta+\rho+1+\mu)(\delta+\rho)\mu
	>(\delta+\rho)\mu
	>\pi_0.
	\]
	The Routh--Hurwitz criterion therefore implies
	\[
		\lambda_{\max}(J_{\mathrm{act}}^0)<0.
	\]
	
	At the threshold $\mathcal R_D=1$, the characteristic polynomial becomes
	\[
	P_0(\lambda)
	=\lambda\bigl(\lambda^2+\pi_2\lambda+\pi_1\bigr).
	\]
	Thus $\lambda=0$ is a simple eigenvalue, while the remaining two
	eigenvalues have negative real parts because $\pi_1,\pi_2>0$. Hence
	$\lambda_{\max}(J_{\mathrm{act}}^0)=0$ at the threshold.
	
	Finally, $J_{\mathrm{act}}^0$ is an irreducible Metzler matrix, with directed
	cycle $d\to c\to m\to d$. Perron--Frobenius theory therefore implies that
	its spectral bound is a real algebraically simple eigenvalue associated with
	a strictly positive eigenvector.
\end{proof}

The corresponding global resolution statement is proved later for the full
spatial system; its restriction to spatially homogeneous data gives the
global conclusion for \eqref{eq:ode}.

\subsubsection{Persistent-damage equilibrium}

Besides the damage-free equilibrium, a persistent-damage equilibrium with
$d^*>0$ may arise. At equilibrium, the last two equations of \eqref{eq:ode} yield:
\begin{equation}\label{mc_patho}
	c^*=\frac{r}{\mu}d^*,\qquad
	m^*=\frac{c^*}{1+c^*},
\end{equation}
while \eqref{eq:w_reaction} gives
\begin{equation}\label{eq:hstar_dstar}
	h^*=1-a_{\rm cap}\,d^*,\qquad\text{with}\qquad a_{\rm cap}:=1+\frac{\delta}{\sigma} .
\end{equation}
The coefficient $a_{\rm cap}$ is the amount of tissue withdrawn from the functional pool per unit of active damage: the damage itself, plus the capacity that has been cleared and is waiting to be reconstituted. Substituting \eqref{mc_patho}--\eqref{eq:hstar_dstar} into $\dot d=0$, and using $\alpha=\mathcal{R}_D\,(\delta+\rho)\mu/r$, the equilibrium condition becomes the identity
\begin{equation}\label{eq:identity}
	\mathcal{R}_D\,h^*=1+c^* ,
\end{equation}
which is linear in $d^*$ and yields the following result.

\begin{proposition}\label{patho}
	The damage-free state $E_0$ is the unique equilibrium of \eqref{eq:ode}
	with $d^*=0$. Moreover, system \eqref{eq:ode} admits an equilibrium with
	$d^*>0$ if and only if $\mathcal R_D>1$. In that case the equilibrium
	$E_+=(h^*,d^*,m^*,c^*)$ is unique and explicit:
	\begin{equation}\label{eq:Estar}
		d^*=
		\frac{1-\mathcal R_D^{-1}}
		{1+\delta/\sigma+(\delta+\rho)/\alpha},
		\qquad
		h^*=1-a_{\rm cap}d^*,
		\qquad
		c^*=\frac{r}{\mu}d^*,
		\qquad
		m^*=\frac{c^*}{1+c^*},
	\end{equation}
	and it belongs to $\mathcal D$, with $h^*>0$ and $h^*+d^*<1$.
	
	For fixed $\delta$, $\rho$ and $\sigma$, the equilibrium damage $d^*$
	is strictly increasing in $\mathcal R_D$ both when $\mathcal R_D$ is
	varied through $\alpha$ with $r/\mu$ fixed, and when it is varied
	through $r/\mu$ with $\alpha$ fixed. Moreover, $d^*\to0$ as
	$\mathcal R_D\to1^+$.
\end{proposition}

\begin{proof}
	If $d^*=0$, the equilibrium equations for $c$, $m$ and $h$, in this
	order, give $c^*=0$, $m^*=0$ and $h^*=1$. Hence the only equilibrium
	with $d^*=0$ is $E_0$.
	
	By \eqref{mc_patho}--\eqref{eq:hstar_dstar}, an equilibrium with
	$d^*>0$ must satisfy \eqref{eq:identity}, that is,
	\[
	\mathcal R_D(1-a_{\rm cap}d^*)
	=1+\frac{r}{\mu}d^*.
	\]
	Since
	\[
	\frac{r}{\mu}
	=\frac{\delta+\rho}{\alpha}\mathcal R_D,
	\]
	this equation yields the expression for $d^*$ in \eqref{eq:Estar}.
	It is positive if and only if $\mathcal R_D>1$. The remaining
	components are then uniquely determined by \eqref{mc_patho} and
	\eqref{eq:hstar_dstar}, and the resulting state satisfies the
	equilibrium equations.
	
	From \eqref{eq:identity},
	$h^*=(1+c^*)/\mathcal R_D>0$, while \eqref{eq:hstar_dstar} gives
	\[
	1-h^*-d^*=(a_{\rm cap}-1)d^*
	=\frac{\delta}{\sigma}d^*>0.
	\]
	Thus $h^*+d^*<1$, moreover $c^*>0$ and
	$m^*\in(0,1)$. Hence $E_+\in\mathcal D$.
	
	For fixed $r/\mu$, $\delta$, $\rho$ and $\sigma$, variation of
	$\mathcal R_D$ through $\alpha$ gives
	\[
	\left.\frac{\partial d^*}{\partial\mathcal R_D}\right|_{r/\mu}
	=\frac{a_{\rm cap}+r/\mu}
	{(a_{\rm cap}\mathcal R_D+r/\mu)^2}>0.
	\]
	For fixed $\alpha$, $\delta$, $\rho$ and $\sigma$, variation of
	$\mathcal R_D$ through $r/\mu$ gives instead
	\[
	\left.\frac{\partial d^*}{\partial\mathcal R_D}\right|_{\alpha}
	=\frac{1}
	{\bigl(1+\delta/\sigma+(\delta+\rho)/\alpha\bigr)
		\mathcal R_D^2}>0.
	\]
	Finally, all parameters being positive, \eqref{eq:Estar} implies
	\[
	0<d^*\le 1-\frac1{\mathcal R_D}.
	\]
	Therefore $d^*\to0$ as $\mathcal R_D\to1^+$.
\end{proof}
\begin{theorem}\label{theo:Estar_stable}
	Let $\mathcal{R}_D>1$. Then the persistent-damage equilibrium $E_+$ is
	locally asymptotically stable with respect to homogeneous perturbations of
	the reaction system \eqref{eq:ode}, for every choice of the remaining
	positive parameters. 
\end{theorem}
\begin{proof}
	Set
	\[
	 \xi_*:=\alpha m^*>0,\qquad g_*:=\alpha h^*>0,
	 \qquad \theta:=(1+c^*)^{-1}\in(0,1).
	\]
	The equilibrium identities imply
	\begin{equation}\label{eq:equilibrium_loop_identity}
	 g_* r\theta=(\delta+\rho)\mu.
	\end{equation}
	Let $(\widetilde h,\widetilde d,\widetilde m,\widetilde c)^T$ denote a
	perturbation of $E_+$ and introduce
	$\widetilde y:=\widetilde h+\widetilde d$. In the coordinates
	$(\widetilde y,\widetilde d,\widetilde m,\widetilde c)^T$, the
	linearized system is
	\begin{equation}\label{eq:Estar_transformed_linearization}
	 \frac{\mathrm d}{\mathrm dt}
	 \begin{pmatrix}\widetilde y\\\widetilde d\\\widetilde m\\\widetilde c\end{pmatrix}
	 =
	 \begin{pmatrix}
	 -\sigma & -\delta & 0 & 0\\
	 \xi_* & -(\xi_*+\delta+\rho) & g_* & 0\\
	 0 & 0 & -1 & \theta^2\\
	 0 & r & 0 & -\mu
	 \end{pmatrix}
	 \begin{pmatrix}\widetilde y\\\widetilde d\\\widetilde m\\\widetilde c\end{pmatrix}.
	\end{equation}
	This change of coordinates is invertible, so the matrix in
	\eqref{eq:Estar_transformed_linearization} has the same eigenvalues as
	$J(E_+)$. Suppose, by contradiction, that $\lambda$ is an eigenvalue
	with $\operatorname{Re}\lambda\ge0$, and let
	$(\widetilde y,\widetilde d,\widetilde m,\widetilde c)^T\ne0$ be a
	corresponding eigenvector. Necessarily $\widetilde d\ne0$: indeed, if
	$\widetilde d=0$, the
	fourth, third, and first eigenvalue equations, in this order, give
	$\widetilde c=\widetilde m=\widetilde y=0$, because $\lambda+\mu$,
	$\lambda+1$, and $\lambda+\sigma$
	do not vanish in the closed right half-plane.

	The eigenvalue equations now give
	\[
	 \widetilde y=-\frac{\delta}{\lambda+\sigma}\widetilde d,
	 \qquad
	 \widetilde m=\frac{r\theta^2}{(\lambda+1)(\lambda+\mu)}\widetilde d.
	\]
	Substitution into the second equation of
	\eqref{eq:Estar_transformed_linearization}, followed by division by
	$\widetilde d$, yields
	\begin{equation}\label{eq:Estar_spectral_identity}
	 \lambda+\xi_*+\delta+\rho+
	 \frac{\xi_*\delta}{\lambda+\sigma}
	 =\frac{g_*r\theta^2}{(\lambda+1)(\lambda+\mu)}.
	\end{equation}
	For $\operatorname{Re}\lambda\ge0$, the real part of the left-hand side
	satisfies
	\begin{align*}
	 \operatorname{Re}\left(
	 \lambda+\xi_*+\delta+\rho+
	 \frac{\xi_*\delta}{\lambda+\sigma}\right)
	 &=\operatorname{Re}\lambda+\xi_*+\delta+\rho
	 +\xi_*\delta\frac{\operatorname{Re}\lambda+\sigma}
	 {|\lambda+\sigma|^2}\\
	 &\ge\delta+\rho.
	\end{align*}
	Consequently, the modulus of the left-hand side of
	\eqref{eq:Estar_spectral_identity} is at least
	$\delta+\rho$. On the other hand, since
	$|\lambda+1|\ge1$ and $|\lambda+\mu|\ge\mu$, identity
	\eqref{eq:equilibrium_loop_identity} gives
	\[
	 \left|\frac{g_*r\theta^2}
	 {(\lambda+1)(\lambda+\mu)}\right|
	 \le\frac{g_*r\theta^2}{\mu}
	 =(\delta+\rho)\theta
	 <\delta+\rho.
	\]
	This contradicts \eqref{eq:Estar_spectral_identity}. Notice that only
	the nonnegativity of $\xi_*$ was used on the left. The sole strict
	inequality in the comparison is $\theta<1$, which is equivalent to
	$c^*>0$. Biologically, the
	cycle gain at $E_+$ is reduced by the recruitment-saturation factor
	$\theta=(1+c^*)^{-1}<1$. At the threshold $\mathcal R_D=1$, the positive
	equilibrium merges with $E_0$, so $c^*=0$ and $\theta=1$; this is the
	bifurcation value at which the strict gain reduction disappears.
	Thus $J(E_+)$ has
	no eigenvalue in the closed right half-plane, and all its eigenvalues
	have strictly negative real part. Since the vector field of
	\eqref{eq:ode} is continuously differentiable in a neighbourhood of
	$E_+$, the linearization principle implies that $E_+$ is locally
	asymptotically stable.
\end{proof}

\begin{remark}\label{rem:chronic_reading}
	For $\mathcal R_D>1$, the fraction of tissue capacity not occupied by
	functional tissue at the persistent-damage equilibrium $E_+$ is
	\[
	\mathcal L(\mathcal R_D):=1-h^*=a_{\rm cap}\,d^*
	=\frac{a_{\rm cap}(\mathcal R_D-1)}
	{a_{\rm cap}\mathcal R_D+r/\mu}.
	\]
	It consists of active damage $d^*$ and available capacity awaiting
	reconstitution,
	\[
	1-h^*-d^*=\frac{\delta}{\sigma}d^*.
	\]
	The ratio of the latter to active damage is therefore $\delta/\sigma$:
	the capacity awaiting reconstitution exceeds the active damage when
	$\delta>\sigma$.
\end{remark}

Figure~\ref{fig:equilibrium_regions} summarizes the homogeneous equilibrium
structure in the parameter ranges used below. In panel~(a), the straight
threshold curve $\mathcal R_D=1$ separates the region where only the
damage-free equilibrium $E_0$ exists from the region where the unique
persistent-damage equilibrium $E_+$ exists; below the threshold $E_0$ is
homogeneously locally stable, whereas above it $E_0$ is unstable and $E_+$ is
homogeneously locally stable. Panel~(b) shows the
equilibrium loss $\mathcal L(\mathcal R_D)$ defined in
Remark~\ref{rem:chronic_reading}: it increases continuously from zero
above threshold when $\mathcal R_D$ is varied through $\alpha$, with the
other parameters fixed.

\begin{figure}[!htbp]
\centering
\begin{subfigure}[t]{0.49\linewidth}
\centering
\includegraphics[width=\linewidth]{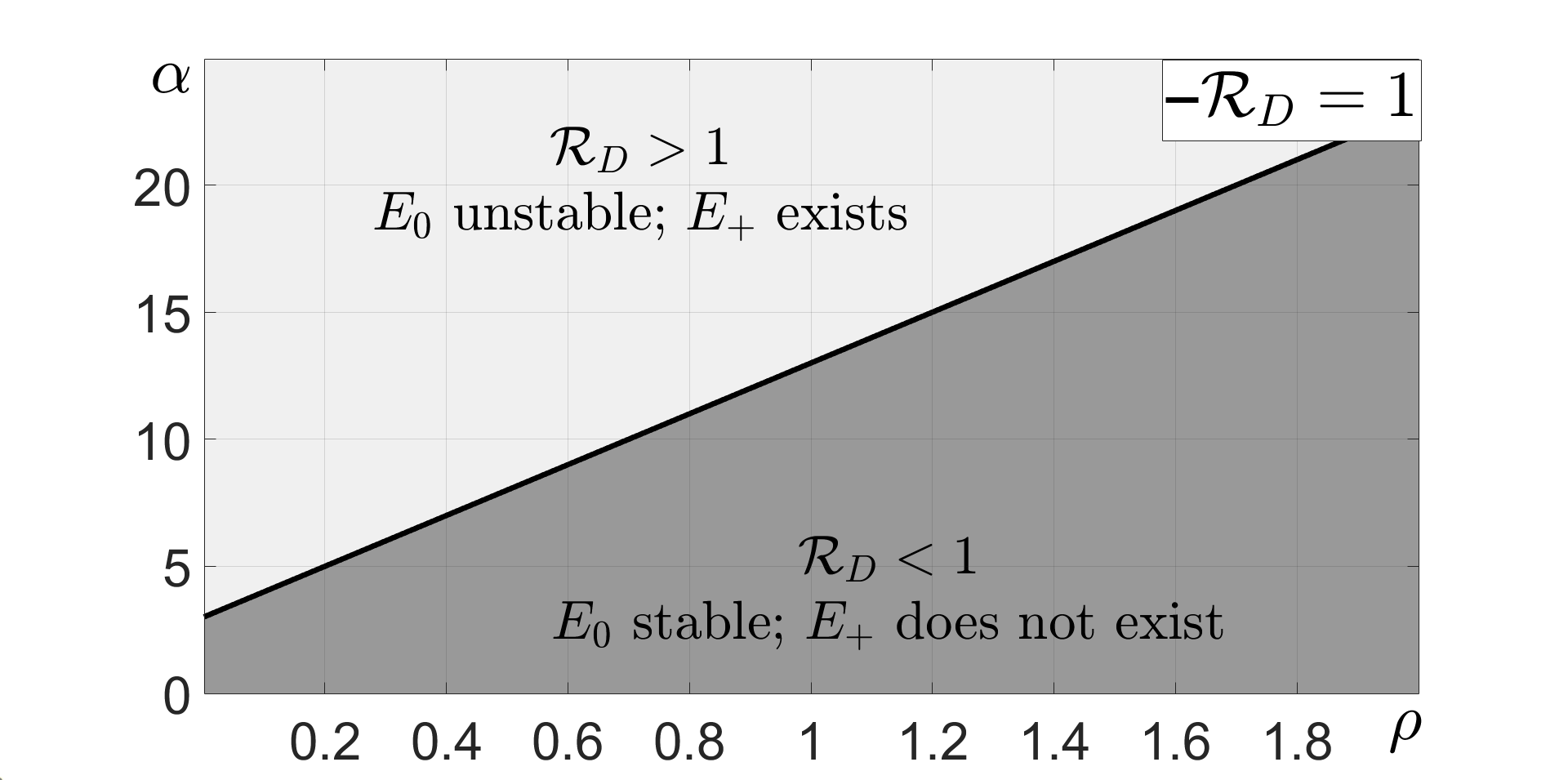}
\label{fig:equilibrium_regions_rho_alpha}
\caption{}
\end{subfigure}
\hfill
\begin{subfigure}[t]{0.49\linewidth}
\centering
\includegraphics[width=\linewidth]{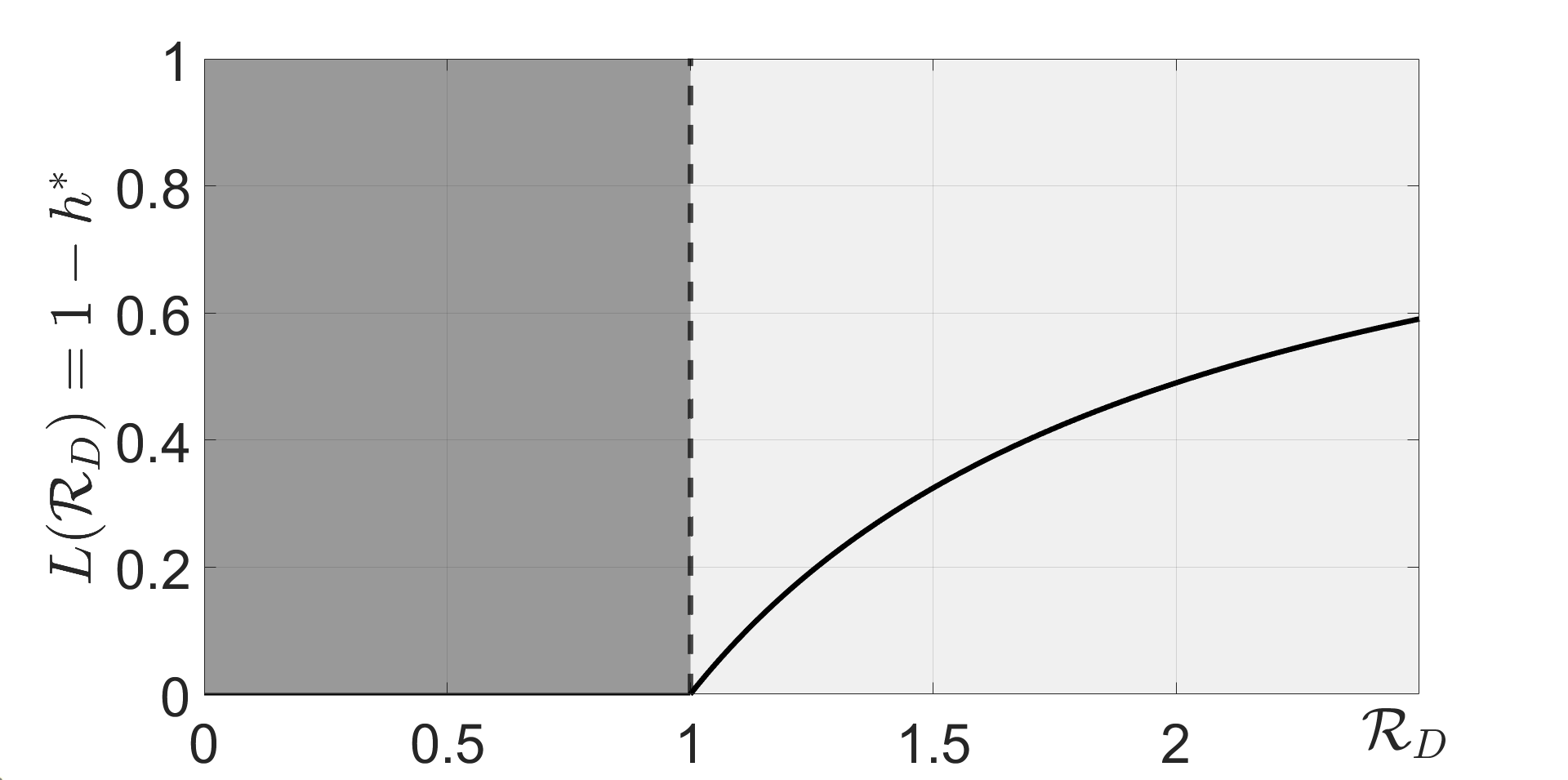}
\label{fig:equilibrium_loss_RD}
\caption{}
\end{subfigure}
\caption{Homogeneous equilibrium structure and persistent loss of functional
	tissue for the reference parameters $\sigma=0.2$, $\delta=0.3$, $\mu=10$,
	and $r=1$.
	Panel~(a) shows the regions of existence and local stability in the
	$(\rho,\alpha)$-plane. The threshold line
	$\mathcal R_D=1$ separates the region where the damage-free
	equilibrium $E_0$ is locally asymptotically stable and the persistent-damage
	equilibrium $E_+$ does not exist from the region where $E_0$ is unstable and
	$E_+$ exists and is locally asymptotically stable.
	All stability statements in this caption refer to homogeneous perturbations
	of the reaction system, not to stability with respect to nonzero spatial
	modes of the PDE.
	Panel~(b) adopts the reference value $\rho=0.5$ and shows the functional-tissue loss $\mathcal L(\mathcal R_D)$ at
	$E_+$ for $\mathcal R_D>1$, as given in Remark~\ref{rem:chronic_reading},
	extended by zero for $\mathcal R_D\le1$, corresponding to the loss at
	$E_0$. The two shaded regions in panel~(b) have the same meaning as the
	corresponding regions in panel~(a).}\label{fig:equilibrium_regions}
\end{figure}
\section{Well-posedness and global boundedness}\label{wp:main_section}

We now turn from the spatially homogeneous reaction system to the full
bounded-domain problem. The invariant tissue region identified in
Proposition~\ref{prop:invariant_region} remains invariant for the
nondiffusing tissue variables: Lemma~\ref{wp:lemma_region} establishes this
pointwise for a prescribed continuous nonnegative coefficient $m$. Using this
fact, the parabolic estimates
collected in Appendix~\ref{app:wp}, and a global
$L^\infty$ estimate for the effective macrophage-activity density, we prove
global existence, uniqueness, uniform boundedness and dissipativity. Spatial
regularity of the tissue variables is propagated from the initial data.

The proof has three stages. First, the nondiffusing tissue variables remain in
the invariant simplex. Second, this pointwise control closes the parabolic
estimates and prevents finite-time blow-up, without a smallness assumption on
$\chi_0$. Third, the same estimates are refined into an absorbing set and the
regularity consequences needed later. The construction details and the
technical bootstrap estimates are deferred to Appendix~\ref{app:wp} and the
appendices cited below.

\subsection{Functional setting and local well-posedness}

For completeness, we recall the dimensionless system:
\begin{equation}\label{wp:model}
	\begin{cases}
		h_t=\sigma(1-h-d)+\rho d-\alpha mh,\\
		d_t=\alpha mh-(\delta+\rho)d,\\
		m_t=\Delta m-\chi_0\nabla\cdot\bigl(\phi(c)m\nabla c\bigr)+f(c)-m,\\
		c_t=D_c\Delta c+rd-\mu c.
	\end{cases}
\end{equation}
Let $\Omega\subset\mathbb R^n$, $n\ge1$, be a bounded connected open set with boundary of class $C^{2+\eta_0}$, where $0<\eta_0<1$. For $0\le s<t$, we set $\Omega_{s,t}=(s,t)\times\Omega$ and $\Omega_T=\Omega_{0,T}$. We use the usual Lebesgue and Sobolev spaces and write
\[
W^{1,2}_p(\Omega_{s,t})
=W^{1,p}((s,t);L^p(\Omega))\cap L^p((s,t);W^{2,p}(\Omega)).
\]
Throughout the analysis $n+2<p<\infty$ is fixed. We denote by
\[
E_p=\{v\in W^{2-2/p,p}(\Omega):\ \partial_{\mathbf n}v=0
\text{ on }\partial\Omega\}
\]
the Neumann trace space with its Sobolev norm.  To distinguish spatial
regularity from time regularity, for $0\le s<t$ we use the notation
\begin{equation}\label{wp:C01x_definition}
	C_x^{0,1}(\overline{\Omega_{s,t}})
	:=C\bigl([s,t];C^1(\overline\Omega)\bigr),
\end{equation}
endowed with the norm
\[
	\|v\|_{C_x^{0,1}(\overline{\Omega_{s,t}})}
	:=\sup_{s\le\tau\le t}
	\bigl(\|v(\tau)\|_{L^\infty(\Omega)}
	+\|\nabla_xv(\tau)\|_{L^\infty(\Omega)}\bigr).
\]
Thus $v$ and its spatial gradient are jointly continuous on the closed
space--time cylinder; no time derivative or Lipschitz continuity in time
is asserted.
We set
\begin{equation}\label{wp:phi_eta}
	\eta=\delta+\rho,\qquad f(c)=\frac{c}{1+c},\qquad
	\phi(c)=\frac{\kappa^2}{(\kappa+c)^2}.
\end{equation}
The following bounds are immediate from \eqref{wp:phi_eta} for \(c\ge0\):
\begin{equation}
0\le f(c)\le\min\{c,1\},\quad 0<f'(c)\le1,\quad
0<\phi(c)\le1,\quad
|\phi'(c)|\le\frac{2}{\kappa},\quad
|\phi''(c)|\le\frac{6}{\kappa^2}.\nonumber
\end{equation}

We impose the boundary and initial conditions
\begin{equation}\label{wp:bc}
	\partial_{\mathbf n}m=\partial_{\mathbf n}c=0
	\quad\text{on }(0,\infty)\times\partial\Omega,
\end{equation}
\begin{equation}\label{wp:ic}
	(h,d,m,c)(0,\cdot)=(h_{in},d_{in},m_{in},c_{in}).
\end{equation}
Since $\partial_{\mathbf n}c=0$, the macrophage condition in \eqref{wp:bc} is equivalent to the no-flux condition
\[
\bigl(\nabla m-\chi_0\phi(c)m\nabla c\bigr)\cdot\mathbf n=0.
\]
No boundary condition is prescribed for the non-diffusing variables $h,d$.

All parameters $\sigma,\rho,\alpha,\delta,D_c,r,\mu,\kappa$ are positive,
whereas the chemotactic sensitivity satisfies $\chi_0\ge0$. The case
$\chi_0=0$ corresponds to the system without chemotaxis. The initial data
satisfy the following assumptions.
\begin{assumption}\label{wp:assumptions}
	We assume that $h_{in},d_{in}\in L^\infty(\Omega)$, with
	\[
	h_{in}\ge0,\qquad d_{in}\ge0,\qquad h_{in}+d_{in}\le1
	\quad\text{a.e. in }\Omega,
	\]
	and that $m_{in},c_{in}\in E_p$ are non-negative.
\end{assumption}
We equip the space
\begin{equation}\label{wp:phase}
	\mathcal X=[L^\infty(\Omega)]^2\times E_p\times E_p
\end{equation}
with the product norm and denote its admissible subset by
\[
\mathcal P=\{(h,d,m,c)\in\mathcal X:\ h,d,m,c\ge0,\ h+d\le1\}.
\]
The space $\mathcal X$ is Banach and $\mathcal P$ is a closed, hence complete,
subset of $\mathcal X$; here we use the embedding
$E_p\hookrightarrow C(\overline\Omega)$, valid because $p>n+2$.
We write $U=(h,d,m,c)$ and $U_{in}=(h_{in},d_{in},m_{in},c_{in})$. Here and below inequalities for the tissue variables are understood almost everywhere in $\Omega$. Unless otherwise stated, constants may depend on $p,n,\Omega$, the parameters and the norm of the initial data in $\mathcal X$.

No smallness assumption is imposed on the chemotactic sensitivity $\chi_0$.
The estimates below are nevertheless allowed to depend on the prescribed
value of $\chi_0$, as part of the parameter set.

In this section, a strong solution means a solution in the class
\eqref{wp:solution_class} that satisfies the equations almost everywhere
and the Neumann boundary conditions in the trace sense.

The following lemma establishes the positivity and invariant-region
properties of the nondiffusing tissue variables. It is used in the local
construction below, while the parabolic estimates required for the
transported variables are collected in Appendix~\ref{app:wp}.

\begin{lemma}[Positivity and invariant tissue region]
\label{wp:lemma_region}
Let $T>0$ and let $m\in C(\overline{\Omega_T})$ satisfy $m\ge0$. For
initial data $h_{in},d_{in}\in L^\infty(\Omega)$ such that
\[
h_{in}\ge0,\qquad d_{in}\ge0,\qquad h_{in}+d_{in}\le1
\quad\text{a.e. in }\Omega,
\]
the first two equations of \eqref{wp:model}, with
\[
h(0,\cdot)=h_{in},\qquad d(0,\cdot)=d_{in},
\]
admit a unique solution
\[
(h,d)\in C^1([0,T];L^\infty(\Omega))^2.
\]
Moreover, for every $t\in[0,T]$,
\[
h(t,x)\ge0,\qquad d(t,x)\ge0,\qquad h(t,x)+d(t,x)\le1
\]
for almost every $x\in\Omega$.
\end{lemma}

\begin{proof}
For almost every fixed $x\in\Omega$, the first two equations form a
linear ODE system for $(h(t,x),d(t,x))$ with continuous coefficient
$m(t,x)$. The corresponding integral equations therefore define a unique
solution, and the coefficient bounds imply
$(h,d)\in C^1([0,T];L^\infty(\Omega))^2$.

Let
\[
\Sigma:=\{(h,d)\in\mathbb R^2:
h\ge0,\ d\ge0,\ h+d\le1\}.
\]
The inward-pointing inequalities on the three faces of $\Sigma$ were
established in the proof of Proposition~\ref{prop:invariant_region} and
hold for every $m\ge0$. For almost every fixed $x\in\Omega$, the
prescribed coefficient $m(\cdot,x)$ is continuous and nonnegative, and
the tissue vector field is locally Lipschitz in $(h,d)$. The invariance
criterion for nonautonomous ODEs therefore implies that
$(h(t,x),d(t,x))\in\Sigma$ for every $t\in[0,T]$.
\end{proof}

With the tissue region controlled by Lemma~\ref{wp:lemma_region}, we can
now construct the coupled local solution. The parabolic estimates needed for
the signal and macrophage equations are stated in Appendix~\ref{app:wp}.

\begin{theorem}\label{wp:local}
	Under Assumption~\ref{wp:assumptions}, system \eqref{wp:model}, with \eqref{wp:bc}--\eqref{wp:ic}, admits a unique strong solution on a maximal interval $[0,T_{\max})$. This solution is non-negative and satisfies $h\ge0$, $d\ge0$ and $h+d\le1$ almost everywhere in $\Omega$ for every $0\le t<T_{\max}$. For every $T<T_{\max}$,
	\begin{equation}\label{wp:solution_class}
		\begin{split}
			&h,d\in C^1([0,T];L^\infty(\Omega)),\\
			&m,c\in W^{1,2}_p(\Omega_T)\cap C([0,T];E_p)
			\cap C_x^{0,1}(\overline{\Omega_T}).
		\end{split}
	\end{equation}
	The local existence time is uniform for data in a bounded subset of $\mathcal P$. If $T_{\max}<\infty$, then
	\begin{equation}\label{wp:continuation_criterion}
		\limsup_{t\uparrow T_{\max}}\|(h,d,m,c)(t)\|_{\mathcal X}=\infty.
	\end{equation}
\end{theorem}

\begin{proof}
	The detailed construction and difference estimates are given in
	Appendix~\ref{app:local_construction}. In brief, the tissue ODEs, the signal
	equation and the linear drift--diffusion equation for $m$ define a fixed-point
	map on a closed subset of $C([0,T];C(\overline\Omega))$. The invariant-region
	lemma and maximal regularity make the map well defined, while the estimates
	in Appendix~\ref{app:local_construction} give a contraction for sufficiently
	small $T$. Banach's theorem yields local existence and uniqueness, uniformly
	for data in bounded admissible subsets. If the $\mathcal X$-norm remained
	bounded at a finite maximal time, the same construction would restart the
	solution and extend it, which proves \eqref{wp:continuation_criterion}.
\end{proof}

\subsection{Global existence and regularity}

\begin{theorem}\label{wp:theorem}
	Under Assumption~\ref{wp:assumptions}, system \eqref{wp:model}, with boundary and initial conditions \eqref{wp:bc}--\eqref{wp:ic}, admits a unique global strong solution which is non-negative in each component. For every $T>0$, the solution belongs to the class \eqref{wp:solution_class}.
	The equations hold almost everywhere, the chemotactic term satisfies
	\[
	\nabla\cdot\bigl(\phi(c)m\nabla c\bigr)\in L^p(\Omega_T),
	\]
	and the tissue fractions remain in the physical region:
	\begin{equation}\label{wp:region}
		h\ge0,\qquad d\ge0,\qquad h+d\le1.
	\end{equation}
	For every $R>0$ there exists a constant $C_R>0$, depending only on
	$R$, the parameters, $p,n$ and $\Omega$, such that every admissible
	initial datum satisfying $\|U_{in}\|_{\mathcal X}\le R$ generates a
	solution for which
	\begin{equation}\label{wp:uniform}
		\sup_{t\ge0}\big(\|m(t)\|_{E_p}+\|c(t)\|_{E_p}
		+\|m(t)\|_{W^{1,\infty}}+\|c(t)\|_{W^{1,\infty}}\big)\le C_R,
	\end{equation}
	and
	\begin{equation}\label{wp:uniform_slabs}
		\sup_{s\ge0}\big(\|m\|_{W^{1,2}_p(\Omega_{s,s+1})}
		+\|c\|_{W^{1,2}_p(\Omega_{s,s+1})}\big)\le C_R.
	\end{equation}
	Moreover, the solution depends locally Lipschitz continuously on the initial data in $\mathcal X$, on every bounded time interval.
	
	If, in addition, $h_{in},d_{in}\in W^{1,\infty}(\Omega)$ and
	$\|h_{in}\|_{W^{1,\infty}}+\|d_{in}\|_{W^{1,\infty}}\le R_1$, then
	$h,d\in$ \\$C^1([0,T];W^{1,\infty}(\Omega))$ for every $T>0$, and
	\begin{equation}\label{wp:all_gradient_bound}
		\sup_{t\ge0}\|(h,d,m,c)(t)\|_{[W^{1,\infty}(\Omega)]^4}
		\le C_{R,R_1}.
	\end{equation}
	Here $C_{R,R_1}$ depends only on $R,R_1$ and the structural data listed
	above.

Finally, if $h_{in},d_{in}\in C(\overline\Omega)$, then $h,d\in C^1([0,T];C(\overline\Omega))$. If $h_{in},d_{in}\in C^{\eta_0}(\overline\Omega)$ and $m_{in},c_{in}\in C^{2+\eta_0}(\overline\Omega)$ satisfy \eqref{wp:bc} at $t=0$, the solution is classical: for some $\vartheta_{\rm reg}>0$,
	\[
	m,c\in C^{2+\vartheta_{\rm reg},1+\vartheta_{\rm reg}/2}(\overline{\Omega_T}),
	\]
	and all four equations hold pointwise with continuous time derivatives. If the tissue data also belong to $C^2(\overline\Omega)$, then $h,d\in C^1([0,T];C^2(\overline\Omega))$.
\end{theorem}

\begin{proof}
	The technical bootstrap is detailed in Appendix~\ref{app:global_bootstrap}.
	We record the estimates that drive it. Since $0\le d\le1$, comparison in
	the signal equation gives
	\begin{equation}\label{wp:main_signal_bound}
		\|c(t)\|_\infty\le e^{-\mu t}\|c_{in}\|_\infty
		+\frac r\mu(1-e^{-\mu t}).
	\end{equation}
	Using the bounded source $rd$, \eqref{wp:main_signal_bound} and the initial
	regularity $c_{in}\in E_p$, the parabolic estimates in
	Appendix~\ref{app:global_bootstrap} give uniform $L^\infty$ bounds for
	$\nabla c$ and $L^p$ bounds for $\Delta c$ on unit time slabs within the
	local existence interval. Since $\phi$ and $\phi'$ are bounded on
	$[0,\infty)$, the drift $b$ is bounded in $L^\infty$ and
	$k=\nabla\cdot b$ is bounded in $L^p$ on these slabs. More explicitly,
	\begin{equation}\label{wp:drift}
	 b=\chi_0\phi(c)\nabla c,\qquad
	 k=\nabla\cdot b
	 =\chi_0\bigl(\phi'(c)|\nabla c|^2+\phi(c)\Delta c\bigr).
	\end{equation}

	Integrating the macrophage equation and using the no-flux condition gives
	\begin{equation}\label{wp:mass_eq}
		\frac d{dt}\int_\Omega m+\int_\Omega m=\int_\Omega f(c)\le|\Omega|,
	\end{equation}
	so that $m$ has a uniform $L^1$ bound on bounded sets of initial data.
	For $s<t$, the variation-of-constants formula is
	\begin{equation}\label{wp:m_duhamel}
	\begin{split}
	 m(t)={}&e^{(t-s)(\Delta-1)}m(s)
	 +\int_s^t e^{(t-\tau)(\Delta-1)}f(c(\tau))\,d\tau\\
	 &-\int_s^t e^{(t-\tau)(\Delta-1)}\nabla\cdot(bm)(\tau)\,d\tau.
	\end{split}
	\end{equation}
	The heat-flux estimate applied to this formula then implies,
	for some $q>n$ and $\vartheta=\frac12+\frac n{2q}<1$,
	\[
	Y_T:=\sup_{0\le t\le T}\|m(t)\|_\infty
	\le C+C\,Y_T^{1-1/q}.
	\]
	Young's inequality yields a bound for $Y_T$ independent of $T$.
	The signal estimates and the cutoff estimate for $m$ give uniform
	parabolic and trace bounds on unit time slabs. Thus the
	$\mathcal X$-norm remains bounded on every finite interval, and the
	continuation criterion in Theorem~\ref{wp:local} gives global existence.

	The same estimates prove \eqref{wp:uniform}--\eqref{wp:uniform_slabs}.
	The chemotactic term belongs to $L^p(\Omega_T)$ by the product expansion
	shown in Appendix~\ref{app:global_bootstrap}. The tissue-gradient estimate
	is obtained from the compartmental evolution and is proved in
	Appendix~\ref{app:tissue}; continuous dependence and classical regularity
	are proved in Appendix~\ref{app:dependence}. The constants are uniform for
	initial data in bounded subsets, which completes the proof.
\end{proof}
\begin{corollary}[Propagation of tissue regularity]
\label{wp:higher_tissue_main}
	Under the hypotheses of Theorem~\ref{wp:theorem}, assume in addition that
	$h_{in},d_{in}\in W^{2,p}(\Omega)$. Then, for every $T>0$,
	\begin{equation}\label{wp:tissue_second_class}
		h,d\in C([0,T];W^{2,p}(\Omega))
	\cap W^{1,p}(0,T;W^{2,p}(\Omega)).
	\end{equation}
	More precisely, for every $R,R_2>0$ there is a constant
	$C^{(2)}_{R,R_2}>0$, depending only on $R,R_2$ and the structural
	data, such that
	\begin{equation}\label{wp:tissue_second_bound}
		\sup_{t\ge0}\big(\|h(t)\|_{W^{2,p}}+\|d(t)\|_{W^{2,p}}\big)
		\le C^{(2)}_{R,R_2}
	\end{equation}
	whenever $\|U_{in}\|_{\mathcal X}\le R$ and
	$\|h_{in}\|_{W^{2,p}}+\|d_{in}\|_{W^{2,p}}\le R_2$.
	In particular, under this additional assumption,
	$(h,d,m,c)\in[W^{1,2}_p(\Omega_T)]^4$.
	No boundary condition is imposed on $h_{in}$ or $d_{in}$.
\end{corollary}
\begin{proof}
	The proof is given in Appendix~\ref{app:tissue}.
\end{proof}

\begin{remark}\label{wp:no_tissue_smoothing}
	The spatial regularity of the tissue variables in Theorem~\ref{wp:theorem} and Corollary~\ref{wp:higher_tissue_main} is inherited from their initial data. In general, these variables do not acquire additional spatial regularity at positive times. Indeed, for $d_{in}=m_{in}=c_{in}=0$ and any admissible $h_{in}$, the solution is
	\[
	d=m=c=0,\qquad h(t,x)=1+e^{-\sigma t}(h_{in}(x)-1).
	\]
	Thus spatial discontinuities in the initial data persist at every finite time. Choosing $h_{in}\in W^{2-2/p,p}\setminus W^{2,p}$ also shows that the assumption $h_{in},d_{in}\in W^{2,p}$ cannot in general be replaced by the parabolic trace regularity $W^{2-2/p,p}$ for the second-derivative conclusion.
\end{remark}

\subsection{Dissipativity and asymptotic bounds}

The estimates above hold for all times, with constants depending on a
bounded set of initial data. We now prove an eventual bound independent
of that set. Global uniqueness and the continuous dependence estimate
\eqref{wp:data_dependence}, together with time continuity in
\eqref{wp:solution_class}, define a continuous semiflow
\[
\Phi:\mathbb R_+\times\mathcal P\longrightarrow\mathcal P,
\qquad \Phi(t)U_{in}=U(t;U_{in}).
\]

\begin{definition}[Bounded absorbing set]\label{wp:def_absorbing}
	A bounded set $B_0\subset\mathcal P$ is called \emph{absorbing} for
	$\Phi$ if, for every bounded set $B\subset\mathcal P$, there exists a
	time $T_B^{\rm abs}\ge0$ such that
	\begin{equation}\label{wp:absorbing_definition}
		\Phi(t)B\subset B_0\qquad\text{for every }t\ge T_B^{\rm abs}.
	\end{equation}
	Positive invariance of $B_0$ from time zero is not required.
\end{definition}

\begin{proposition}\label{wp:absorbing}
	There is a constant $K_*>0$, depending only on $p,n,\Omega$ and the
	parameters, and in particular independent of the initial datum, such
	that every solution of Theorem~\ref{wp:theorem} satisfies
	\begin{equation}\label{wp:eventual_strong}
		\limsup_{t\to\infty}\big(\|m(t)\|_{E_p}+\|c(t)\|_{E_p}
		+\|m(t)\|_{W^{1,\infty}}+\|c(t)\|_{W^{1,\infty}}\big)\le K_*.
	\end{equation}
	In particular, define
	\begin{equation}\label{wp:B0}
	\begin{split}
	B_0:=\bigl\{(h,d,m,c)\in\mathcal P:\;&
	\|m\|_{E_p}+\|c\|_{E_p}\\
	&+\|m\|_{W^{1,\infty}}+\|c\|_{W^{1,\infty}}
	\le K_*+1\bigr\}
	\end{split}
	\end{equation}
	Because $p>n+2$, the embedding $E_p\hookrightarrow W^{1,\infty}(\Omega)$
	is continuous; hence $B_0$ is closed in $\mathcal X$. The invariant tissue
	region supplies the remaining bounds on $h$ and $d$, so $B_0$ is bounded in
	the full phase space and is absorbing in the sense of
	Definition~\ref{wp:def_absorbing}. More precisely, for every bounded
	$B\subset\mathcal P$ there exists $T_B^{\rm abs}\ge0$ such that
	\begin{equation}\label{wp:B0_absorption}
		\Phi(t)B\subset B_0\qquad(t\ge T_B^{\rm abs}).
	\end{equation}
	Moreover,
	\begin{equation}\label{wp:asymptotic_basic}
		\limsup_{t\to\infty}\|c(t)\|_\infty\le\frac r\mu,
		\qquad
		\limsup_{t\to\infty}\frac1{|\Omega|}\int_\Omega m(t,x)\,dx
		\le\frac{r}{r+\mu}.
	\end{equation}
\end{proposition}

\begin{proof}
	The technical entrance-time argument is given in
	Appendix~\ref{app:absorbing_bootstrap}. The estimates
	\eqref{wp:main_signal_bound} and \eqref{wp:mass_eq} give the asymptotic
	bounds in \eqref{wp:asymptotic_basic}. The cutoff estimate applied to the
	signal equation then yields an eventual data-independent bound for
	$\nabla c$ and hence for the drift. The mass estimate similarly bounds the
	$L^1$ norm of $m$. The
	Duhamel formula \eqref{wp:m_duhamel}, combined with the heat-flux estimate
	on successive unit intervals, yields an eventual $L^\infty$ bound for $m$
	independent of the initial datum. The cutoff estimate
	\eqref{wp:cutoff_bound}, applied first to $c$ and then to $m$, gives the
	eventual bounds in \eqref{wp:eventual_strong}. Hence $B_0$ is absorbing;
	the recurrence and its entrance-time estimate are recorded in the appendix.
\end{proof}
\begin{corollary}\label{wp:improved_bounds}
	Let $M$ be any finite upper bound for $\limsup_{t\to\infty}\|m(t)\|_\infty$, and set
	\[
	D_* =\frac{\alpha M}{\alpha M+\eta}<1.
	\]
	Then
	\begin{equation}\label{wp:improved_asymptotic}
		\begin{split}
			&\limsup_{t\to\infty}\|d(t)\|_\infty\le D_*,\qquad
			\limsup_{t\to\infty}\|c(t)\|_\infty\le\frac r\mu D_*,\\
			&\limsup_{t\to\infty}\frac1{|\Omega|}\int_\Omega m(t,x)\,dx
			\le\frac{rD_*}{\mu+rD_*}.
		\end{split}
	\end{equation}
\end{corollary}
\begin{proof}
	For every $\varepsilon>0$, eventually $m\le M+\varepsilon$. Since $h\le1-d$,
	\[
	d_t\le\alpha(M+\varepsilon)(1-d)-\eta d.
	\]
	Scalar comparison, followed by $\varepsilon\to0$, proves the first bound. Comparison in the signal equation and then \eqref{wp:mass_eq} give the other two.
\end{proof}

\begin{remark}[Quantitative content of the dissipative bounds]
\label{wp:remark_iterated_bounds}
	For the general chemotactic PDE, Proposition~\ref{wp:absorbing} and
	Corollary~\ref{wp:improved_bounds} provide uniform eventual bounds, although
	the absorbing constant is not evaluated in closed form. For spatially
	homogeneous solutions, a sharper explicit iteration is available; its
	details are recorded in Appendix~\ref{app:ode_iteration}.
\end{remark}

\section{Nonlinear extinction and decay rates}\label{wp:extinction_section}

This section proves the global resolution result below and at the
damage-amplification threshold. We consider the damage-free equilibrium
$E_0=(1,0,0,0)$ and use the number $\mathcal R_D$ introduced in
Definition~\ref{def:RD}; throughout the section we set
$\eta:=\delta+\rho$. The no-flux conditions remove diffusion and chemotaxis
from the integrated mass equations. Below threshold, a cooperative comparison
argument gives exponential decay; at the threshold, the zero eigenvalue
requires a nonlinear dissipation estimate and yields algebraic decay. Thus the
resulting trichotomy is quantitative: below threshold the
active variables decay exponentially, at the threshold they satisfy an
algebraic decay bound of order $(1+t)^{-1}$, whereas above threshold global
extinction for all admissible data is impossible. The last statement is a
consequence of the positive equilibrium constructed in Section~\ref{sec:local};
uniform persistence of nonzero active data is established separately in
Section~\ref{pdep:sec:pde_persistence}.
\begin{theorem}\label{wp:extinction}
	Let Assumption~\ref{wp:assumptions} hold and set $E_0=(1,0,0,0)$.
	
	If $\mathcal R_D<1$, let $\omega_{\rm dec}$ be the unique number in $(0,\min\{\eta,1,\mu\})$ satisfying
	\begin{equation}\label{wp:spectral_rate}
		(\eta-\omega_{\rm dec})(1-\omega_{\rm dec})(\mu-\omega_{\rm dec})=\alpha r.
	\end{equation}
	Then
	\begin{equation}\label{wp:subcritical_rate}
		\|U(t)-E_0\|_{\mathcal X}\le
		\begin{cases}
			C e^{-\min\{\sigma,\omega_{\rm dec}\}t},&\sigma\ne\omega_{\rm dec},\\
			C(1+t)e^{-\omega_{\rm dec}t},&\sigma=\omega_{\rm dec}.
		\end{cases}
	\end{equation}
	If $\mathcal R_D=1$, then
	\begin{equation}\label{wp:critical_rate}
		\|U(t)-E_0\|_{\mathcal X}\le\dfrac{C}{1+t}.
	\end{equation}
	The constants can be chosen uniformly for initial data in any bounded subset of $\mathcal P$. In particular, for $\mathcal R_D\le1$,
	\begin{equation}\label{wp:extinction_limit}
		\|h(t)-1\|_\infty+\|d(t)\|_\infty
		+\|m(t)\|_\infty+\|c(t)\|_\infty\longrightarrow0.
	\end{equation}
	No restriction on $\chi_0$ is required.
\end{theorem}
\begin{proof}
	For $\mathcal R_D<1$, put
	$\mathcal M_d(t)=\int_\Omega d(t,x)\,dx$,
	$\mathcal M_m(t)=\int_\Omega m(t,x)\,dx$,
	$\mathcal M_c(t)=\int_\Omega c(t,x)\,dx$, and
	$\mathbf Y=(\mathcal M_d,\mathcal M_m,\mathcal M_c)^T$.
	Since $h\le1$ and $f(c)\le c$,
	\[
	\mathbf Y'\le J_{\mathrm{act}}^0\mathbf Y.
	\]
	The left-hand side of \eqref{wp:spectral_rate} decreases strictly from $\eta\mu$ to zero on $[0,\min\{\eta,1,\mu\}]$. Its root $\omega_{\rm dec}$ is therefore unique, and $v=(\alpha/(\eta-\omega_{\rm dec}),1,1-\omega_{\rm dec})^T>0$ satisfies $J_{\mathrm{act}}^0v=-\omega_{\rm dec}v$. Cooperative comparison gives
	\begin{equation}\label{wp:mass_decay}
	0\le \mathbf Y(t)\le Ce^{-\omega_{\rm dec} t}v.
	\end{equation}
	Lemma~\ref{wp:decay_lifting} transfers \eqref{wp:mass_decay} to the strong norms with the same exponential weight. For $z=1-h-d$,
	
	\begin{equation}\label{wp:z_representation}
		z(t)=e^{-\sigma t}z_{in}+\delta\int_0^t e^{-\sigma(t-s)}d(s)\,ds.
	\end{equation}
	Evaluating the convolution and using $1-h=d+z$ proves \eqref{wp:subcritical_rate}, including resonance.
	
	At $\mathcal R_D=1$, define
	\begin{equation}\label{wp:critical_lyapunov}
		V(t)=\int_\Omega\left(d+\alpha m+\frac\eta r c\right).
	\end{equation}
	Since $\alpha r=\eta\mu$, integration gives
	\begin{equation}\label{wp:critical_dissipation}
		V'=-\alpha\int_\Omega\left(m(1-h)+\frac{c^2}{1+c}\right)\le0.
	\end{equation}
	We show that this dissipation controls all three masses after a fixed time delay. For data in a bounded admissible set, write $m\le K_{m,\mathcal B}$ and $c\le K_{c,\mathcal B}$. Put $\beta_{\rm t}=\min\{\sigma,\rho\}$ and $\Lambda_{\mathcal B}=\beta_{\rm t}+\alpha K_{m,\mathcal B}$. Since $h_t\ge \beta_{\rm t}-\Lambda_{\mathcal B}h$,
	\begin{equation}\label{wp:h_lower}
		h(t,x)\ge h_*:=\frac{\beta_{\rm t}}{\Lambda_{\mathcal B}}(1-e^{-\Lambda_{\mathcal B}})>0,
		\qquad t\ge1.
	\end{equation}
	Set $\mathcal M_d(t)=\int_\Omega d(t,x)\,dx$,
	$\mathcal M_m(t)=\int_\Omega m(t,x)\,dx$ and
	$\mathcal M_c(t)=\int_\Omega c(t,x)\,dx$. Their vector
	$\mathbf Y=(\mathcal M_d,\mathcal M_m,\mathcal M_c)^T$ satisfies
	\[
	\mathbf Y'\ge J_-\mathbf Y,\qquad J_-=\begin{pmatrix}-\eta&\alpha h_*&0\\0&-1&0\\r&0&-\mu\end{pmatrix}.
	\]
	The matrix $J_-$ is Metzler, so its semigroup is positive. Moreover, the
	directed paths $d\to c$ and $m\to d\to c$ imply that every entry in the
	last row of $e^{\tau J_-}$ is strictly positive for $\tau>0$, including
	the diagonal entry corresponding to $c$. This remains true when removal
	rates coincide. By continuity on the compact interval $1\le\tau\le2$,
	these entries have a strictly positive minimum.
	Using this minimum and the positive weights in $V$ gives
	\begin{equation}\label{wp:mass_observation}
		\mathcal M_c(t+\tau)\ge k_0V(t),\qquad t\ge1,\quad1\le\tau\le2.
	\end{equation}
	By Cauchy--Schwarz, $-V'(s)\ge\alpha \mathcal M_c(s)^2/(|\Omega|(1+K_{c,\mathcal B}))$. Integrating over $[t+1,t+2]$ yields
	\begin{equation}\label{wp:discrete_critical_decay}
		V(t)-V(t+2)\ge V(t+1)-V(t+2)\ge\frac{\alpha k_0^2}{|\Omega|(1+K_{c,\mathcal B})}V(t)^2
		=k_1V(t)^2.
	\end{equation}
	Set $v_j=V(1+2j)$. While $v_j,v_{j+1}>0$, $v_j-v_{j+1}\ge k_1v_j^2$ gives $v_{j+1}^{-1}-v_j^{-1}\ge k_1v_j/v_{j+1}\ge k_1$, hence $v_j\le(v_0^{-1}+k_1j)^{-1}$. If a term vanishes, every later term vanishes. Monotonicity fills the intervening times and the uniform bounds cover $[0,1]$. Thus $V(t)\le C/(1+t)$, uniformly on bounded sets of data. Lemma~\ref{wp:decay_lifting} and \eqref{wp:z_representation} prove \eqref{wp:critical_rate}.
\end{proof}
\begin{remark}[Sharpness of the extinction threshold]\label{wp:threshold_sharp}
	When $\mathcal R_D>1$, the equilibrium $E_+$ of
	Proposition~\ref{patho}, viewed as a spatially constant state, belongs to
	the PDE phase space $\mathcal P$ and is stationary for \eqref{wp:model}.
	Thus the global-extinction conclusion of Theorem~\ref{wp:extinction}
	cannot be extended above threshold. The complementary persistence result
	is proved in Section~\ref{pdep:sec:pde_persistence}.
\end{remark}

\section{Uniform persistence above threshold on bounded domains}\label{pdep:sec:pde_persistence}

This section establishes the complementary supercritical result on bounded
connected domains. Under the assumptions of
Section~\ref{wp:main_section}, if $\RD>1$ and at least one active component
is initially nonzero, the solution is uniformly persistent. The persistence
bounds depend on the model parameters, the domain and the functional setting,
but not on the initial state; the time after which they hold may depend on
the solution. The main conclusion is componentwise and spatial: after a
transient, active damage is bounded away from zero almost everywhere, while
effective macrophage activity and inflammatory signal are bounded away from
zero throughout the domain.

The proof proceeds in four steps. We first establish uniform weak persistence
with respect to the extinction boundary. We then prove asymptotic compactness
of the ODE--PDE semiflow and use it to obtain strong persistence. Next, we
derive separate lower bounds for the masses of the active components.
Finally, heat-kernel positivity and scalar comparison yield the pointwise
lower bounds. This result is stronger than persistence of spatial averages
and provides the bounded-domain PDE counterpart of the threshold
classification for the spatially homogeneous system in Section~2.

\subsection{Persistence framework and main result}

The phase space $\mathcal P$ and the continuous semiflow $\Phi$ are those
of Section~\ref{wp:main_section}. Set $K_{\rm abs}:=K_*+1$. For every
bounded set $B\subset\mathcal P$, the absorbing property
\eqref{wp:B0_absorption} ensures that its trajectories lie in the set
$B_0$ defined by~\eqref{wp:B0} for all $t\ge T_B^{\rm abs}$.
Throughout this section, constants depending only on the data may depend
on the model parameters and on $p,n,\Omega$, but not on the initial state.

Define the persistence function $\mathfrak p:\mathcal P\to[0,\infty)$ by
\begin{equation}\label{pdep:eq:persistence_function}
	\mathfrak p(U)=\norm d_\infty+\norm m_\infty+\norm c_\infty,
	\qquad
	\mathcal P_0=\mathfrak p^{-1}(0),\quad
	\mathcal P_+=\mathcal P\setminus\mathcal P_0.
\end{equation}
Thus
\[
\mathcal P_0=
\{(h,0,0,0):h\in L^\infty(\Omega),\ 0\le h\le1
\text{ a.e. in }\Omega\}.
\]
The function $\mathfrak p$ is Lipschitz continuous with respect to the
topology of $\mathcal X$. The set $\mathcal P_+$ consists precisely of
the admissible states for which at least one active component is nonzero;
no uniform positive distance from $\mathcal P_0$ is imposed on the
initial data.

\begin{theorem}[Pointwise uniform persistence]
	\label{pdep:thm:pointwise}
	Under Assumption~\ref{wp:assumptions}, let $\RD>1$. There are constants
	$\underline\varepsilon_d,\underline\varepsilon_m,
	\underline\varepsilon_c>0$,
	depending only on the parameters and on $p,n,\Om$, such that every
	solution issued from $\mathcal P_+$ satisfies
	\[
	\liminf_{t\to\infty}\,\operatorname*{ess\,inf}_{x\in\Om}d(t,x)
	\ge\underline\varepsilon_d,
	\qquad
	\liminf_{t\to\infty}\,\min_{x\in\overline\Om}m(t,x)
	\ge\underline\varepsilon_m,
	\qquad
	\liminf_{t\to\infty}\,\min_{x\in\overline\Om}c(t,x)
	\ge\underline\varepsilon_c .
	\]
	Moreover, with $A=\alpha K_{\rm abs}$,
	$\beta_{\rm t}=\min\{\sigma,\rho\}$ and
	$\underline h_A=\beta_{\rm t}/(\beta_{\rm t}+A)$,
	\[
	\underline h_A\le
	\liminf_{t\to\infty}\,\operatorname*{ess\,inf}_{x\in\Om}h(t,x)
	\le
	\limsup_{t\to\infty}\,\operatorname*{ess\,sup}_{x\in\Om}h(t,x)
	\le1-\underline\varepsilon_d .
	\]
\end{theorem}

The minima for $m$ and $c$ are well defined because
$E_p\hookrightarrow C(\overline\Omega)$. For any fixed $\theta\in(0,1)$,
the active components satisfy the stated lower bounds multiplied by $\theta$
after a time that may depend on the initial state and on $\theta$.

The proof uses the absorbing estimates from Section~3 and
Appendix~\ref{app:absorbing_bootstrap}. For asymptotic compactness, it
also uses the uniform contraction of the tissue evolution proved in
Appendix~\ref{app:tissue}.

\subsection{Extinction boundary and uniform weak persistence}

For a solution of the system, denote the masses of the active components by
\[
\mathcal M_d(t)=\int_\Om d(t,x)\,dx,\qquad
\mathcal M_m(t)=\int_\Om m(t,x)\,dx,\qquad
\mathcal M_c(t)=\int_\Om c(t,x)\,dx.
\]

\begin{lemma}[Invariant decomposition]
	\label{pdep:lem:invariance}
	The sets $\mathcal P_0$ and $\mathcal P_+$ are positively invariant. On
	$\mathcal P_0$,
	\[
	h(t,x)=1+e^{-\sigma t}(h_{\rm in}(x)-1)
	\qquad\text{for a.e. }x\in\Om.
	\]
	The only complete bounded trajectory contained in $\mathcal P_0$ is
	$\Ezero=(1,0,0,0)$.
\end{lemma}

\begin{proof}
	If the initial state belongs to $\mathcal P_0$, uniqueness gives
	$d=m=c=0$, and the equation for $h$ yields the displayed formula.
	Thus $\mathcal P_0$ is positively invariant.
	
	For an admissible solution, the no-flux conditions and nonnegativity give,
	for almost every $t>0$,
	\[
	\mathcal M_d'\ge-\eta\mathcal M_d,\qquad
	\mathcal M_m'\ge-\mathcal M_m,\qquad
	\mathcal M_c'\ge-\mu\mathcal M_c.
	\]
	A nonzero active component has positive mass, which therefore remains
	positive at every finite time. Hence $\mathcal P_+$ is positively invariant.
	
	Finally, let $U(t)$ be a complete trajectory contained in $\mathcal P_0$.
	For $s<t$, the formula for $h$ and $0\le h(s)\le1$ imply
	\[
	\norm{1-h(t)}_\infty
	=e^{-\sigma(t-s)}\norm{1-h(s)}_\infty
	\le e^{-\sigma(t-s)}.
	\]
	Letting $s\to-\infty$ gives $h(t)=1$ in $L^\infty(\Om)$ for every
	$t\in\mathbb R$. Thus the trajectory is the equilibrium $\Ezero$.
\end{proof}

\begin{lemma}[Uniform weak persistence]
	\label{pdep:lem:weak}
	If $\RD>1$, there exists $\varepsilon_w>0$, depending only on $\delta$,
	$\sigma$ and $\RD$, such that
	\begin{equation}
		\label{pdep:eq:weak}
		\limsup_{t\to\infty}\mathfrak p(\Phi(t)U_{\rm in})
		\ge\varepsilon_w
		\qquad (U_{\rm in}\in\mathcal P_+).
	\end{equation}
\end{lemma}

\begin{proof}
	Choose $\underline h,\underline g\in(0,1)$, depending only on $\RD$, such
	that
	\begin{equation}
		\label{pdep:eq:hg}
		\underline h\,\underline g\,\RD>1.
	\end{equation}
	For example, one may take
	$\underline h=\underline g=\RD^{-1/3}$.
	
	Put $z=1-h-d$. For almost every $x\in\Om$,
	\begin{equation}
		\label{pdep:eq:z}
		z_t=\delta d-\sigma z,\qquad
		z(t)=e^{-\sigma(t-t_0)}z(t_0)
		+\delta\int_{t_0}^t e^{-\sigma(t-s)}d(s)\,ds.
	\end{equation}
	Choose $\varepsilon_w>0$ sufficiently small that
	\begin{equation}
		\label{pdep:eq:epsilon_choice}
		1-\varepsilon_w-\frac{\delta}{\sigma}\varepsilon_w
		>\underline h,\qquad
		\frac1{1+\varepsilon_w}>\underline g.
	\end{equation}
	This choice depends only on $\delta$, $\sigma$ and $\RD$.
	
	Suppose that~\eqref{pdep:eq:weak} fails for some active solution. Choose
	$\widehat\varepsilon$ strictly between its limiting superior and
	$\varepsilon_w$. There is $t_0$ such that
	\[
	\norm{d(t)}_\infty+\norm{m(t)}_\infty+\norm{c(t)}_\infty
	<\widehat\varepsilon,\qquad t\ge t_0.
	\]
	Since $0\le z\le1$, formula~\eqref{pdep:eq:z} gives
	\[
	\norm{z(t)}_\infty
	\le e^{-\sigma(t-t_0)}
	+\frac{\delta}{\sigma}\widehat\varepsilon
	\qquad(t\ge t_0).
	\]
	Using $\widehat\varepsilon<\varepsilon_w$
	and~\eqref{pdep:eq:epsilon_choice}, we obtain, for some $T\ge t_0$,
	\[
	h(t,x)\ge\underline h,\qquad
	f(c(t,x))\ge\underline g\,c(t,x)
	\]
	for almost every $x\in\Om$ and all $t\ge T$.
	
	Integration of the active equations now gives, componentwise,
	\[
	\mathbf Y'\ge A_*\mathbf Y,\qquad
	\mathbf Y=(\mathcal M_d,\mathcal M_m,\mathcal M_c)^T,\qquad
	A_*=\begin{pmatrix}
		-\eta&\alpha\underline h&0\\
		0&-1&\underline g\\
		r&0&-\mu
	\end{pmatrix}
	\]
	for almost every $t\ge T$. By~\eqref{pdep:eq:hg}, the characteristic
	equation of $A_*$,
	\[
	(\lambda+\eta)(\lambda+1)(\lambda+\mu)
	=\alpha r\underline h\underline g,
	\]
	has a positive real root $\lambda_+$. Hence the spectral bound satisfies
	$\lambda_{\max}(A_*)\ge\lambda_+>0$. Since $A_*$ is irreducible and
	Metzler, it has a strictly positive left eigenvector $\ell$ associated with
	$\lambda_{\max}(A_*)$. Consequently,
	\[
	\frac d{dt}(\ell^T\mathbf Y)
	\ge\lambda_{\max}(A_*)\ell^T\mathbf Y.
	\]
	Lemma~\ref{pdep:lem:invariance} gives $\mathbf Y(T)\ne0$, so
	$\ell^T\mathbf Y(T)>0$ and the preceding inequality forces
	$\ell^T\mathbf Y(t)$ to grow exponentially.
	
	On the other hand, the assumed bound on the active components gives
	\[
	\ell^T\mathbf Y(t)
	\le |\Om|\max_i\ell_i\,\widehat\varepsilon
	\qquad(t\ge T),
	\]
	a contradiction. This proves~\eqref{pdep:eq:weak}.
\end{proof}

\subsection{Asymptotic compactness of the semiflow}

The set $B_0$ constructed in Section~\ref{wp:main_section} is bounded and
absorbing, but is not assumed to be positively invariant. To construct a
compact global attractor, we establish asymptotic compactness of the
semiflow. The next lemmas provide the parabolic smoothing and the
contraction of the nondiffusing tissue variables needed for this step.

\begin{lemma}[Interior higher-integrability smoothing]
\label{pdep:lem:smoothing}
Set $q:=2p$.  There is a constant $C_q>0$, depending only on the data, and
for every bounded $B\subset\mathcal P$ a time $T_B^{\rm sm}>0$, such that
every solution starting in $B$ satisfies
\begin{equation}
\label{pdep:eq:qslab}
 \norm m_{W^{1,2}_q((t-1,t)\times\Om)}
 +\norm c_{W^{1,2}_q((t-1,t)\times\Om)}
 +\norm{m(t)}_{E_q}+\norm{c(t)}_{E_q}\le C_q
\end{equation}
for $t\ge T_B^{\rm sm}$.  Here
$E_q=W^{2-2/q,q}(\Om)$ with the Neumann trace condition.
\end{lemma}

\begin{proof}
Fix a bounded initial set $B$ and take $t\ge T_B^{\rm abs}+3$, so that
the orbit lies in $B_0$ on $[t-2,t]$. Choose a smooth cutoff
$\zeta$ which vanishes at $t-2$, equals one on $[t-3/2,t]$, and has
uniformly bounded derivative. Then $w=\zeta c$ solves
\[
 w_t-D_c\Delta w+\mu w=\zeta rd+\zeta'c,\qquad w(t-2)=0,
 \qquad\partial_{\mathbf n}w=0.
\]
The source is uniformly bounded in $L^q$, since $0\le d\le1$ and $c$
is bounded. By maximal $L^q$-regularity, with the estimate
\eqref{wp:maxreg} applied for $q$ in place of $p$, there is a solution
in $W^{1,2}_q$; as $q>p$, this solution also belongs to $W^{1,2}_p$ and
agrees with $w$ by uniqueness. The trace estimate therefore gives
\begin{equation}\label{pdep:eq:cq}
 \norm c_{W^{1,2}_q((t-3/2,t)\times\Om)}
 +\sup_{t-3/2\le s\le t}\norm{c(s)}_{E_q}\le C_q.
\end{equation}
Using the coefficient formulas \eqref{wp:drift} and the parabolic
embedding for $q>n+2$, we obtain
\begin{equation}\label{pdep:eq:bkq}
 \norm b_{L^\infty((t-3/2,t)\times\Om)}
 +\norm k_{L^q((t-3/2,t)\times\Om)}\le C_q.
\end{equation}
Take a second cutoff $\xi$ which vanishes at $t-3/2$, equals one on
$[t-1,t]$, and has uniformly bounded derivative. The function $v=\xi m$ satisfies
\[
 v_t-\Delta v=-b\cdot\nabla v-(1+k)v+\xi f(c)+\xi'm,
 \qquad v(t-3/2)=0,\qquad\partial_{\mathbf n}v=0.
\]
The right-hand side is uniformly bounded in $L^q$: the definition
\eqref{wp:B0} already bounds $m$ and $\nabla m$ in $L^\infty$,
while \eqref{pdep:eq:bkq} controls $k$ in $L^q$. Here the right-hand side is evaluated on the already constructed solution;
no $W^{1,2}_q$ regularity of $m$ is used in this estimate. Maximal regularity and
the same uniqueness argument give
\begin{equation}\label{pdep:eq:mq}
 \norm m_{W^{1,2}_q((t-1,t)\times\Om)}
 +\sup_{t-1\le s\le t}\norm{m(s)}_{E_q}\le C_q.
\end{equation}
Combining \eqref{pdep:eq:cq} and \eqref{pdep:eq:mq} proves the assertion
with $T_B^{\rm sm}=T_B^{\rm abs}+3$. Only this waiting time depends on $B$.
\end{proof}

\begin{lemma}[Uniform contraction of the tissue evolution]
\label{pdep:lem:tissue_contraction}
Put $u=(h,d,1-h-d)^T$.  At almost every $x$,
\begin{equation}\label{wp:Q}
 u_t=Q(\alpha m)u,\qquad
 Q(a)=\begin{pmatrix}
 -a&\rho&\sigma\\
 a&-\eta&0\\
 0&\delta&-\sigma
 \end{pmatrix}.
\end{equation}
With $A$ as in Theorem~\ref{pdep:thm:pointwise}, there exist $C_0,\omega>0$, depending only
on $\rho$, $\sigma$, $\eta$ and $A$, such that, on every interval
$[s,t]$ where
\[
0\le\alpha m(\tau,x)\le A,
\qquad s\le\tau\le t,\quad x\in\overline\Om,
\]
the evolution matrix $\mathcal U_m(t,s;x)$ satisfies
\begin{equation}
	\label{pdep:eq:contraction}
	|\mathcal U_m(t,s;x)v|_1
	\le C_0e^{-\omega(t-s)}|v|_1,
	\qquad {\bf1}^Tv=0.
\end{equation}
The constants are uniform in $s$, $t$, $x$ and in the trajectory.
\end{lemma}

\begin{proof}
Apply estimate~\eqref{wp:contraction} from Appendix~\ref{app:tissue}
with $A=\alpha K_{\rm abs}$. Its proof requires only $0\le\alpha m\le A$.
Hence $C_0,\omega$ are uniform for all trajectories after absorption,
although the absorption time may depend on the initial bounded set. For trajectories issued from a bounded set $B\subset\mathcal P$, the
coefficient bound holds on every interval $[s,t]$ with
$s\ge T_B^{\rm abs}$.
\end{proof}

\begin{proposition}[Asymptotic compactness]
\label{pdep:prop:compactness}
The semiflow $\Phi$ is asymptotically compact in $\mathcal X$.  Together with the existence of the bounded absorbing set $B_0$, it
therefore possesses a compact global attractor $\mathcal A\subset\mathcal P$.
\end{proposition}

\begin{proof}
Let $U_j\in B$, with $B\subset\mathcal P$ bounded, and let $t_j\to\infty$.
Set $q=2p$.  Lemma~\ref{pdep:lem:smoothing} and the compact embedding
$E_q\Subset E_p$ show that $m_j(t_j)$ and $c_j(t_j)$ are precompact in
$E_p$.  Here compactness follows from Rellich's theorem, since
$2-2/q>2-2/p$ and $q>p$.  On every fixed translated interval $[-T,0]$,
finitely many overlapping applications of~\eqref{pdep:eq:qslab}, followed by
the parabolic H\"older embedding, give a common modulus of continuity in
space and time.  Arzel\`a--Ascoli and a diagonal extraction therefore give
\[
 m_j(t_j+\cdot)\longrightarrow \bar m
 \quad\text{uniformly on }[-T,0]\times\overline\Om.
\]

It remains to treat the nondiffusing tissue.  Let $e=(1,0,0)^T$ and fix
$T>0$; we consider only indices $j$ so large that $t_j-T\ge T_B^{\rm sm}$, so that
the orbit is in the absorbing set on $[t_j-T,t_j]$ and
Lemma~\ref{pdep:lem:tissue_contraction} applies there.  Since
$u_j(t_j-T)-e$ has zero component sum,
\begin{align}
u_j(t_j)
={}&\mathcal U_{m_j}(t_j,t_j-T)e \notag\\
&+\mathcal U_{m_j}(t_j,t_j-T)
       \bigl(u_j(t_j-T)-e\bigr).
\label{pdep:eq:tissue_split}
\end{align}
The second term is bounded in $[L^\infty(\Om)]^3$ by
$2C_0e^{-\omega T}$.  For fixed $T$, the first term is precompact in
$[C(\overline\Om)]^3$: indeed, continuous dependence of a finite-dimensional
linear evolution operator on its coefficient gives
\[
\sup_{x\in\Om}
\left|\mathcal U_{m_j}(t_j,t_j-T;x)e
-\mathcal U_{m_k}(t_k,t_k-T;x)e\right|
\le C_T\sup_{\substack{-T\le s\le0\\x\in\Om}}
|m_j(t_j+s,x)-m_k(t_k+s,x)|,
\]
where $C_T$ depends on $T$ and the data but not on $j$ or $k$.
Given $\varepsilon>0$, first choose $T$ so that
$2C_0e^{-\omega T}<\varepsilon/2$, and then cover the precompact family of
first terms by finitely many $\varepsilon/2$-balls.  Formula
~\eqref{pdep:eq:tissue_split} proves total boundedness of
$\{u_j(t_j)\}$ in $[L^\infty]^3$.  Hence $(h_j(t_j),d_j(t_j))$ is
precompact in $[L^\infty]^2$.  Combining this with the parabolic components
proves asymptotic compactness in $\mathcal X$.

Define the omega-limit set of the closed bounded absorbing set by
\[
 \mathcal A=\bigcap_{s\ge0}
 \overline{\bigcup_{t\ge s}\Phi(t)B_0}^{\,\mathcal X}.
\]
Asymptotic compactness makes this set nonempty and compact: limits of
$\Phi(t_j)V_j$, with $V_j\in B_0$ and $t_j\to\infty$, belong to
$\mathcal A$, and every sequence in $\mathcal A$ is approximated by such
a sequence. Since $B_0$ is bounded and absorbing, there is
$T_{B_0}^{\rm abs}$ such that $\Phi(t)B_0\subset B_0$ for all
$t\ge T_{B_0}^{\rm abs}$. The closedness of $B_0$ therefore gives
$\mathcal A\subset B_0$, without requiring $B_0$ to be positively invariant.
Continuity gives $\Phi(\tau)\mathcal A\subset\mathcal A$. For the reverse
inclusion, approximate $a\in\mathcal A$ by $\Phi(t_j)V_j$ with
$t_j\ge j+\tau$ and extract a limit of $\Phi(t_j-\tau)V_j$; this limit
belongs to $\mathcal A$ and maps to $a$. Thus
\begin{equation}\label{pdep:eq:attractor_invariance}
 \Phi(\tau)\mathcal A=\mathcal A\qquad(\tau\ge0).
\end{equation}
If $\mathcal A$ did not attract $B_0$, a sequence $\Phi(t_j)V_j$ staying
a fixed positive distance from $\mathcal A$ would have a subsequence
converging into $\mathcal A$, a contradiction. Absorption and the
semiflow property extend attraction to every bounded subset of
$\mathcal P$.
\end{proof}

\subsection{Strong and componentwise persistence}

\begin{theorem}[Uniform persistence of the PDE system]
\label{pdep:thm:pde_persistence}
Assume $\RD>1$.  There exists $\varepsilon_*>0$, depending only on the
parameters and on $p$, $n$, $\Om$, such that every solution with
$U_{\rm in}\in\mathcal P_+$ satisfies
\begin{equation}
\label{pdep:eq:pde_persistence}
 \liminf_{t\to\infty}
 \bigl(\norm{d(t)}_\infty+\norm{m(t)}_\infty+\norm{c(t)}_\infty\bigr)
 \ge\varepsilon_*.
\end{equation}
\end{theorem}

\begin{proof}
	We use the classical weak-to-strong persistence argument
	\cite{SmithThieme2011,Thieme1993}, recording the relevant hypotheses in our
	setting. The phase space is complete, $\Phi$ is continuous, $\mathcal P_+$
	is positively invariant by Lemma~\ref{pdep:lem:invariance}, and
	Proposition~\ref{pdep:prop:compactness} supplies a compact global attractor
	$\mathcal A$.
	
	Let $a=\varepsilon_w/2$, where $\varepsilon_w$ is given by
	Lemma~\ref{pdep:lem:weak}. If the conclusion failed, we could choose initial
	states $U_j\in\mathcal P_+$ such that
	\[
	\liminf_{t\to\infty}\mathfrak p(\Phi(t)U_j)<\frac{1}{2j}.
	\]
	For $j$ sufficiently large that $1/j<a$, weak persistence supplies
	arbitrarily late times at which $\mathfrak p(\Phi(t)U_j)>a$, followed by
	times at which $\mathfrak p(\Phi(t)U_j)<1/j$.
	
	Choose such times after the orbit has been absorbed into $B_0$ and remains
	within distance $1/j$ of $\mathcal A$. Let $t_j^-$ be the first time after
	the chosen time above $a$ at which
	$\mathfrak p(\Phi(t_j^-)U_j)=1/j$, and let $t_j^+$ be the last time before
	$t_j^-$ at which $\mathfrak p(\Phi(t_j^+)U_j)=a$. Set
	\[
	V_j=\Phi(t_j^+)U_j,\qquad \tau_j=t_j^--t_j^+.
	\]
	By continuity,
	\[
	\mathfrak p(V_j)=a,\qquad
	\mathfrak p(\Phi(\tau_j)V_j)=\frac1j,\qquad
	\mathfrak p(\Phi(t)V_j)\le a
	\quad(0\le t\le\tau_j).
	\]
	
	The points $V_j$ have distance at most $1/j$ from $\mathcal A$. By
	compactness of $\mathcal A$, a subsequence converges to some
	$V\in\mathcal A$. In particular, $\mathfrak p(V)=a>0$, so $V\in\mathcal P_+$.
	
	If a subsequence of $\tau_j$ were bounded, we could assume
	$\tau_j\to\tau<\infty$; joint continuity would then give
	$\mathfrak p(\Phi(\tau)V)=0$, contrary to positive invariance of
	$\mathcal P_+$ in Lemma~\ref{pdep:lem:invariance}. Hence, after extraction,
	$\tau_j\to\infty$. For each fixed $t\ge0$, the inequality along the
	crossing segment and continuity now give
	$\mathfrak p(\Phi(t)V)\le a$, contradicting Lemma~\ref{pdep:lem:weak}.
	This proves uniform strong persistence.
\end{proof}

\begin{theorem}[Componentwise mass persistence]
\label{pdep:thm:componentwise}
Assume $\RD>1$. There are constants $q_d,q_m,q_c>0$, depending only on
the parameters and on $p$, $n$, $\Om$, such that every solution issued
from $\mathcal P_+$ satisfies
\begin{equation}\label{pdep:eq:component_mass_persistence}
 \liminf_{t\to\infty}\int_\Om d(t,x)\,dx\ge q_d,
 \quad
 \liminf_{t\to\infty}\int_\Om m(t,x)\,dx\ge q_m,
 \quad
 \liminf_{t\to\infty}\int_\Om c(t,x)\,dx\ge q_c.
\end{equation}
\end{theorem}

\begin{proof}
We first record a fact about complete trajectories in the global attractor.
The argument is by contradiction: if one component had arbitrarily small
mass along a sequence of active solutions, compactness would produce a
complete attractor trajectory with that component missing, while the tissue
and heat equations will be used to rule this out.
Since $\mathcal A\subset B_0$, invariance of $\mathcal A$ implies
\[
 0\le \alpha m(t,x)\le A
 \qquad(t\in\mathbb R,\ x\in\Om)
\]
on every complete trajectory in $\mathcal A$. With $\beta_{\rm t}$ as in
Theorem~\ref{pdep:thm:pointwise}, put $z=1-h-d$. Since $z,d\ge0$ and
$z+d=1-h$,
\[
 h_t=\sigma z+\rho d-\alpha mh
 \ge \beta_{\rm t}-(\beta_{\rm t}+A)h.
\]
Integration from $s<t$ gives
\[
 h(t,x)\ge e^{-(\beta_{\rm t}+A)(t-s)}h(s,x)
 +\frac{\beta_{\rm t}}{\beta_{\rm t}+A}
  \bigl(1-e^{-(\beta_{\rm t}+A)(t-s)}\bigr).
\]
Letting $s\to-\infty$ shows that every complete attractor trajectory obeys
\begin{equation}
\label{pdep:eq:hpositive}
 h(t,x)\ge \underline h_A
 \quad\text{for a.e. }x\in\Om\text{ and every }t\in\mathbb R.
\end{equation}

Suppose, contrary to \eqref{pdep:eq:component_mass_persistence}, that no
uniform lower bound exists for the spatial integral of one of the three
components, denoted by $v\in\{d,m,c\}$. Then for every $j\ge1$ there is
$U_j\in\mathcal P_+$ whose orbit satisfies
\[
 \liminf_{t\to\infty}\int_\Om v_j(t,x)\,dx<\frac1j.
\]
The data $U_j$ need not
lie in a bounded subset of $\mathcal X$, so
Proposition~\ref{pdep:prop:compactness} cannot be applied to them directly; we
first move the initial points into the absorbing set.  Since the singleton
$\{U_j\}$ is bounded, there is $\tau_j\ge0$ with
$\Phi(t)U_j\in B_0$ for $t\ge\tau_j$.  By
Theorem~\ref{pdep:thm:pde_persistence} there is $T_j\ge0$ with
$\mathfrak p(\Phi(t)U_j)\ge\varepsilon_*/2$ for $t\ge T_j$.  Choose
\[
 t_j\ge\max\{\tau_j,T_j\}+2j
 \quad\text{with}\quad
 \int_\Om v_j(t_j,x)\,dx<\frac1j,
\]
and set
\[
 \widetilde U_j=\Phi(t_j-j)U_j\in B_0\cap\mathcal P_+ .
\]
Then, for every $s\ge-j$,
\begin{equation}
\label{pdep:eq:reparam}
 \Phi(t_j+s)U_j=\Phi(j+s)\widetilde U_j,
 \qquad
 \mathfrak p\bigl(\Phi(j+s)\widetilde U_j\bigr)\ge\frac{\varepsilon_*}{2},
\end{equation}
and the $v$-component of $\Phi(j)\widetilde U_j$ has spatial integral
smaller than $1/j$.

The points $\widetilde U_j$ lie in the bounded set $B_0$, and for each
fixed integer $k\ge0$ the times $j-k$ tend to infinity.
Proposition~\ref{pdep:prop:compactness} therefore shows that
$\{\Phi(j-k)\widetilde U_j\}_{j>k}$ is precompact in $\mathcal X$ for each
$k$.  A diagonal extraction gives a subsequence, not relabelled, and points
$W_k\in\mathcal A$ such that
$\Phi(j-k)\widetilde U_j\to W_k$ in $\mathcal X$ for every $k$.  By
continuity of $\Phi(1)$, $\Phi(1)W_{k+1}=W_k$, so that
\[
 U_*(s)=\Phi(s+k)W_k\qquad(s\ge-k)
\]
is a well-defined complete trajectory in $\mathcal A$, and continuity of
$\Phi(s+k)$ gives
\[
 \Phi(t_j+s)U_j=\Phi(j+s)\widetilde U_j\longrightarrow U_*(s)
 \quad\text{in }\mathcal X
\]
for every fixed $s\in\mathbb R$. Continuity of $\mathfrak p$ in
$\mathcal X$ and convergence of the components in $L^1(\Om)$, together
with~\eqref{pdep:eq:reparam}, yield
\begin{equation}
\label{pdep:eq:complete_positive}
 \mathfrak p(U_*(s))\ge\frac{\varepsilon_*}{2}\quad(s\in\mathbb R),
 \qquad \int_\Om v_*(0,x)\,dx=0.
\end{equation}
Hence $v_*(0)=0$ almost everywhere in $\Om$.

Let $D_*,M_*,C_*$ be the spatial integrals of $d_*,m_*,c_*$.  They are
continuously differentiable and satisfy
\begin{equation}
\label{pdep:eq:massidentities}
 D_*'=\alpha\int_\Om m_*h_*-\eta D_*,\qquad
 M_*'=\int_\Om f(c_*)-M_*,\qquad
 C_*'=rD_*-\mu C_*.
\end{equation}
If $D_*(0)=0$, then the nonnegative differentiable function $D_*$ has a
minimum at zero, so $D_*'(0)=0$. Equations~\eqref{pdep:eq:massidentities}
and~\eqref{pdep:eq:hpositive} give $M_*(0)=0$; the minimum property for
$M_*$ then gives $C_*(0)=0$, because $f(\xi)>0$ for $\xi>0$. If instead
$M_*(0)=0$, the minimum property for $M_*$ first gives $C_*(0)=0$, and
that for $C_*$ gives $D_*(0)=0$. Finally, if $C_*(0)=0$, the $C_*$
identity gives $D_*(0)=0$, reducing to the first case. Thus all three
masses vanish at time zero and, by non-negativity,
$d_*(0)=m_*(0)=c_*(0)=0$ almost everywhere,
contradicting \eqref{pdep:eq:complete_positive}. Repeating the argument
for each $v\in\{d,m,c\}$ proves~\eqref{pdep:eq:component_mass_persistence}.
\end{proof}

\subsection{Pointwise persistence and its consequences}

We complete the proof of Theorem~\ref{pdep:thm:pointwise}. The lower
bound for the mass of $d$ makes $c$ uniformly positive through the
Neumann heat kernel. The source $f(c)$ then gives a lower bound for $m$,
which feeds the tissue equation for $d$. Two auxiliary estimates justify
the kernel bound and scalar comparison in the presence of chemotaxis.

Throughout this section $\mathcal A_c=D_c\Delta-\mu$ denotes the Neumann realization
on $L^q(\Om)$, with $q=2p$, and $T_B^{\rm abs}$ is, as in the proof of
Lemma~\ref{pdep:lem:smoothing}, a time after which every orbit issued from the
bounded set $B$ lies in $B_0$.

\begin{lemma}[Positivity of the Neumann heat kernel]
\label{pdep:lem:kernel}
There is $\kappa_0>0$, depending only on $D_c$ and $\Om$, such that
\[
 \bigl(e^{\tau D_c\Delta}g\bigr)(x)\ge\kappa_0\int_\Om g(y)\,dy
 \qquad(x\in\overline\Om,\ \tfrac12\le\tau\le1)
\]
for every nonnegative $g\in L^\infty(\Om)$.
\end{lemma}

\begin{proof}
	The Neumann heat kernel $G(\tau,x,y)$ is jointly continuous for $\tau>0$
	up to the boundary and is positive in the interior of the connected domain;
	see \cite[Sections~4.2 and~6.2]{Ouhabaz2005} and
	\cite[Chapter~IV]{LSU88}. For fixed $y\in\Om$ and $\tau>0$, the strong
	maximum principle and the boundary point lemma, together with the
	homogeneous Neumann condition, give (see, for example,
	\cite[Chapter~3]{ProtterWeinberger1984})
	\[
	G(\tau,x,y)>0
	\qquad(x\in\overline\Om,\ y\in\Om).
	\]
	By symmetry of the Neumann heat kernel, the same conclusion holds when
	$x\in\Om$ and $y\in\overline\Om$. The semigroup property then extends
	strict positivity to both boundary variables: for
	$x,y\in\overline\Om$,
	\[
	G(\tau,x,y)
	=\int_\Om G(\tau/2,x,z)G(\tau/2,z,y)\,dz>0.
	\]
	Therefore, by continuity and compactness,
	\[
	\kappa_0=\min_{\substack{1/2\le\tau\le1\\x,y\in\overline\Om}}
	G(\tau,x,y)>0,
	\]
	which proves the assertion by integration against $g\ge0$.
\end{proof}

\begin{lemma}[Bounded time derivative and Laplacian of $c$]
\label{pdep:lem:laplacian}
There are constants $K_1,K_2>0$ and $K_3\ge0$, depending only on the data, such that
every solution starting in a bounded set $B\subset\mathcal P$ satisfies,
for $t\ge T_B^{\rm abs}+1$,
\begin{equation}
\label{pdep:eq:laplacian}
 \norm{c_t(t)}_\infty\le K_1,\qquad
 \norm{\Delta c(t)}_\infty\le K_2,\qquad
 \norm{k(t)}_\infty\le K_3,
\end{equation}
where $k=\nabla\cdot\bigl(\chi_0\phi(c)\nabla c\bigr)$ as in
Lemma~\ref{pdep:lem:smoothing}.
\end{lemma}

\begin{proof}
The semiflow is continuous in $\mathcal X$, so $h,d$ are continuous in time
with values in $L^\infty(\Om)$ and $m$ is continuous with values in
$C(\overline\Om)$.  The tissue equation
$d_t=\alpha mh-\eta d$ then shows that
$d\in C^1([0,\infty);L^\infty(\Om))$, and for $t\ge T_B^{\rm abs}$
\[
 0\le d\le1,\qquad \norm{d_t(t)}_\infty\le\alpha K_{\rm abs}+\eta .
\]
Fix $s\ge T_B^{\rm abs}$.  Since $c$ is the strong solution of a linear problem,
for $t>s$
\[
 c(t)=e^{(t-s)\mathcal A_c}c(s)+r\int_0^{t-s}e^{\tau\mathcal A_c}d(t-\tau)\,d\tau .
\]
Because the forcing is continuously differentiable in time with values in
$L^q(\Om)$, the homogeneous term is differentiable for $t>s$ by
analyticity of the semigroup.  The Bochner-valued Leibniz rule applied to
the convolution term is justified by $d\in C^1([s,\infty);L^q(\Om))$ and
gives the following identity in $L^q(\Om)$; this is the standard
inhomogeneous abstract Cauchy-problem formula, for which we refer to
\cite[Chapter~4, Section~4.2]{Pazy1983}:
\begin{equation}
\label{pdep:eq:ct}
 c_t(t)=\mathcal A_c e^{(t-s)\mathcal A_c}c(s)+re^{(t-s)\mathcal A_c}d(s)
 +r\int_0^{t-s}e^{\tau\mathcal A_c}d_t(t-\tau)\,d\tau .
\end{equation}
	Take $t\in[s+1,s+2]$.  The semigroup $e^{\tau\mathcal A_c}$ is a contraction on
$L^\infty(\Om)$, so the last two terms of~\eqref{pdep:eq:ct} are bounded in
$L^\infty$ by $r+2r(\alpha K_{\rm abs}+\eta)$.  For the first term write
\[
 \mathcal A_c e^{(t-s)\mathcal A_c}=e^{(t-s-1/2)\mathcal A_c}\,
\bigl(\mathcal A_c e^{\mathcal A_c/2}\bigr).
\]
The operator $\mathcal A_c e^{\mathcal A_c/2}$ is bounded on $L^q(\Om)$ by analyticity of the
semigroup, and $e^{\vartheta\mathcal A_c}$ maps $L^q(\Om)$ boundedly into
	$L^\infty(\Om)$, uniformly for $\tfrac12\le\vartheta\le\tfrac32$, because
$q>n/2$.  Hence
$\norm{\mathcal A_c e^{(t-s)\mathcal A_c}c(s)}_\infty\le C\norm{c(s)}_q\le C|\Om|^{1/q}K_{\rm abs}$.
Since every $t\ge T_B^{\rm abs}+1$ is of the form $t=s+1$ with
$s\ge T_B^{\rm abs}$, this
proves the first bound in~\eqref{pdep:eq:laplacian}.

The equation for $c$ gives
$D_c\Delta c=c_t+\mu c-rd$ almost everywhere, whence the second bound with
$K_2=(K_1+\mu K_{\rm abs}+r)/D_c$.  Finally
$k=\chi_0\bigl(\phi'(c)|\nabla c|^2+\phi(c)\Delta c\bigr)$ with
$0<\phi\le1$ and $|\phi'|\le2/\kappa$ on $[0,\infty)$, so the third bound
holds with $K_3=\chi_0(2K_{\rm abs}^2/\kappa+K_2)$.
\end{proof}

\begin{proof}[Proof of Theorem~\ref{pdep:thm:pointwise}]
Fix $U_{\rm in}\in\mathcal P_+$.  The singleton $\{U_{\rm in}\}$ is
bounded, so by Theorem~\ref{pdep:thm:componentwise} and
Lemma~\ref{pdep:lem:laplacian} there is a time $T_0$, depending on the orbit,
such that for $t\ge T_0$ the orbit lies in the absorbing set $B_0$,
the bounds~\eqref{pdep:eq:laplacian} hold, and
\begin{equation}
\label{pdep:eq:dmass}
 \int_\Om d(t,x)\,dx\ge\frac{q_d}2 .
\end{equation}
All constants below depend only on the data; only the waiting times depend
on the orbit.

\emph{Step 1: the inflammatory signal.}  Let $t\ge T_0+1$.  The variation of
constants formula on $[t-1,t]$, the positivity of the semigroup, and
$c(t-1)\ge0$, $d\ge0$ give
\[
 c(t)\ge r\int_{1/2}^{1}e^{-\mu\tau}e^{\tau D_c\Delta}d(t-\tau)\,d\tau .
\]
By Lemma~\ref{pdep:lem:kernel} and~\eqref{pdep:eq:dmass},
\begin{equation}
\label{pdep:eq:cpoint}
 c(t,x)\ge\underline\varepsilon_c
 :=\frac{r e^{-\mu}\kappa_0\,q_d}{4}
 \qquad(x\in\overline\Om,\ t\ge T_0+1).
\end{equation}

\emph{Step 2: the effective macrophage-activity density.}  Put $T_1=T_0+1$ and
$f_0=f(\underline\varepsilon_c)>0$.  Since $f$ is increasing,
\eqref{pdep:eq:cpoint} gives $f(c)\ge f_0$ for $t\ge T_1$.  With
$b=\chi_0\phi(c)\nabla c$, the equation for $m$ reads
\[
 m_t=\nabla\cdot(\nabla m-bm)+f(c)-m,
 \qquad\partial_{\mathbf n}m=0,\qquad b\cdot\mathbf n=0
 \ \text{ on }\partial\Om,
\]
the last identity because $\partial_{\mathbf n}c=0$.  Consider the
spatially constant function
\[
 \underline m(t)=\frac{f_0}{K_3+1}
 \Bigl(1-e^{-(K_3+1)(t-T_1)}\Bigr),
 \qquad
 \underline m'+(K_3+1)\underline m=f_0,\quad\underline m(T_1)=0 .
\]
Since $\underline m\ge0$ and $k=\nabla\cdot b\le K_3$ almost everywhere
by~\eqref{pdep:eq:laplacian},
\[
 \underline m'+\underline m+\nabla\cdot(b\,\underline m)
 =\underline m'+(1+k)\underline m\le f_0\le f(c),
\]
so $\underline m$ is a subsolution.  Let
$w=\underline m-m$.  Subtracting the two relations,
\[
 w_t\le\Delta w-\nabla\cdot(bw)-w
 \qquad\text{a.e. in }(T_1,\infty)\times\Om,
 \qquad w(T_1)\le0 .
\]
Multiply by $w^+=\max\{w,0\}$ and integrate over $\Om$.  The boundary
terms vanish because $\partial_{\mathbf n}w=0$ and
$b\cdot\mathbf n=0$, and Young's inequality gives
\[
 \frac12\frac d{dt}\int_\Om(w^+)^2
 \le-\int_\Om|\nabla w^+|^2+\int_\Om w^+\,b\cdot\nabla w^+
 -\int_\Om(w^+)^2
 \le\frac{\norm b_\infty^2}{4}\int_\Om(w^+)^2 .
\]
Since $w^+(T_1)=0$, Gronwall's inequality yields $w^+\equiv0$, that is,
$m\ge\underline m$ for $t\ge T_1$.  Letting $t\to\infty$,
\begin{equation}
\label{pdep:eq:mpoint}
 \liminf_{t\to\infty}\,\min_{x\in\overline\Om}m(t,x)
 \ge\underline\varepsilon_m:=\frac{f_0}{K_3+1}
 =\frac{f(\underline\varepsilon_c)}{K_3+1}.
\end{equation}

\emph{Step 3: the tissue.} For $t\ge T_0$ we have $0\le\alpha m\le A$.
Applying the comparison argument from the proof of
Theorem~\ref{pdep:thm:componentwise}, starting at $T_0$, gives
\[
 h(t,x)\ge
 \underline h_A\bigl(1-e^{-(\beta_{\rm t}+A)(t-T_0)}\bigr)
 \qquad\text{for a.e. }x\in\Om.
\]
Consequently,
\[
 \liminf_{t\to\infty}\operatorname*{ess\,inf}_{x\in\Om}h(t,x)
 \ge\underline h_A.
\]
Let $\theta_0\in(0,1)$.  By~\eqref{pdep:eq:mpoint} and the bound just proved
there is $T_{\theta_0}$ such that
$m\ge\theta_0\underline\varepsilon_m$ and $h\ge\theta_0\underline h_A$ for
$t\ge T_{\theta_0}$.  Then, for almost every $x$,
\[
 d_t=\alpha mh-\eta d
 \ge\alpha\theta_0^2\underline\varepsilon_m\underline h_A-\eta d
 \qquad(t\ge T_{\theta_0}),
\]
and therefore
$\liminf_{t\to\infty}\operatorname*{ess\,inf}_xd(t,x)
\ge\alpha\theta_0^2\underline\varepsilon_m\underline h_A/\eta$.  Since
$\theta_0<1$ is arbitrary, the assertion for $d$ holds with
\[
 \underline\varepsilon_d
 :=\frac{\alpha\,\underline\varepsilon_m\,\underline h_A}{\eta}.
\]
Finally $h\le1-d$ gives the upper bound for $h$.
\end{proof}

\begin{remark}
	The lower bounds obtained in the proof depend on the mass-persistence
	constant $q_d$, which is obtained by compactness; they are therefore not
	fully explicit in terms of the parameters. Their dependence on $\chi_0$
	is allowed through $q_d$, $K_{\rm abs}$ and $K_3$, although the threshold
	itself is independent of chemotaxis. The comparison argument applies to
	the scalar equation for $m$ with $c$ fixed and does not require
	monotonicity of the full system.
\end{remark}

\begin{corollary}[Persistence of the spatially homogeneous system]
\label{cor:ode_persistence}
Assume $\RD>1$. There exist constants
$\varepsilon_d^{\rm ODE},\varepsilon_m^{\rm ODE},
\varepsilon_c^{\rm ODE}>0$, independent of the initial state, such that
every solution of~\eqref{eq:ode} issued from $\mathcal D$ with
$d(0)+m(0)+c(0)>0$ satisfies
\[
 \liminf_{t\to\infty}d(t)\ge\varepsilon_d^{\rm ODE},\qquad
 \liminf_{t\to\infty}m(t)\ge\varepsilon_m^{\rm ODE},\qquad
 \liminf_{t\to\infty}c(t)\ge\varepsilon_c^{\rm ODE}.
\]
\end{corollary}

\begin{proof}
Fix any bounded connected smooth domain $\Omega$ and regard an admissible
ODE initial state as a spatially constant element of $\mathcal P_+$. The
solution of~\eqref{eq:ode}, extended constantly in space, is a strong
solution of~\eqref{wp:model}: both diffusion terms and the chemotactic
flux vanish identically. Uniqueness in Theorem~\ref{wp:theorem} therefore
implies that the PDE solution issued from these data remains spatially
homogeneous and coincides with the ODE solution. The assertion now follows
from Theorem~\ref{pdep:thm:pointwise}. Since the domain and all functional
setting parameters are fixed once and for all, the resulting constants are
uniform over all admissible active ODE initial states.
\end{proof}

\begin{remark}[Biological reading of the persistence results]
In a bounded connected region with no-flux boundaries, spatially homogeneous
parameters, signal diffusion and sustained regenerative capacity, the
positive lower bound for damage may be small. Its positivity neither excludes
spatial heterogeneity nor implies clinically significant damage at every
location. Persistence does not imply convergence to $E_+$.
\end{remark}

\section{Linear stability and spreading analysis}\label{sec:invasion}

This section distinguishes the linear leading-edge prediction, the nonlinear
upper spreading bound and the speed selected by the nonlinear dynamics.
We first linearize the system at the damage-free equilibrium $E_0$ and show
that diffusion cannot generate a Turing or oscillatory instability there.
We then derive the threshold for linear growth of localized perturbations and
the associated candidate linear spreading speed by a marginal-stability
calculation. Finally, when $\chi_0=0$, we prove a nonlinear upper spreading
bound for compactly supported active initial data on the whole line. We stress here that with this analysis, we do not establish the
existence of travelling waves, a matching lower bound, and equality between
the candidate linear speed and a speed selected by the full nonlinear system.

The leading-edge analysis concerns small perturbations of $E_0$. Spatial
stability of the positive equilibrium and nonlinear pattern formation away
from $E_0$ are separate questions.

\subsection{Spatial modes and stability of the damage-free state}\label{subsec:noTuring}

We introduce the perturbation
\begin{equation}
	\boldsymbol{\zeta}
	= (\zeta_h,\zeta_d,\zeta_m,\zeta_c)^{\mathsf T}
	= (h,d,m,c)^{\mathsf T}-(1,0,0,0)^{\mathsf T}.
\end{equation}
Since the chemotactic flux contains the product $m\nabla c$, it is
quadratic in $\boldsymbol{\zeta}$ near $E_0$ and does not contribute at
linear order. Moreover,
\[
\frac{c}{1+c}=c+O(c^2).
\]
Consequently, retaining only linear terms in \eqref{eq:dimensionless} gives
\begin{equation}
	\partial_t\boldsymbol{\zeta}
	=J(E_0)\boldsymbol{\zeta}
	+\mathscr D\Delta\boldsymbol{\zeta},
	\label{eq:matriciale}
\end{equation}
where $J(E_0)$ is the reaction Jacobian defined in Section~\ref{sec:local}
and
\begin{equation}
	\mathscr D=
	\begin{pmatrix}
		0&0&0&0\\
		0&0&0&0\\
		0&0&1&0\\
		0&0&0&D_c
	\end{pmatrix}
	\label{eq:diffusion_matrix}
\end{equation}
is the diffusion matrix of the dimensionless system. The active $(d,m,c)$
block of $J(E_0)$ is the matrix $J_{\mathrm{act}}^0$ introduced in
Section~\ref{sec:local}.

We first investigate whether spatial heterogeneity can arise through Turing
instability. For the one-dimensional calculation used below, we consider
modes of the form $e^{\mathrm{i}kx+\lambda t}$. On a bounded domain with
homogeneous Neumann boundary conditions, $k^2$ is understood as an
eigenvalue of $-\Delta$; on the whole line, it ranges continuously over
$[0,\infty)$. In several spatial dimensions, the same calculation applies
with $k^2$ replaced by $|\mathbf k|^2$. Such mechanisms would lead to
pattern formation from small perturbations of the homogeneous equilibrium.

Substituting this ansatz into \eqref{eq:matriciale}, the eigenvalues $\lambda$
are determined by the characteristic equation
\begin{equation*}
	\det\big(\lambda I-J(E_0)+k^2\mathscr D\big)=0 .
\end{equation*}
Expanding the determinant along its first column, which contains the single non-zero entry $\lambda+\sigma$, we obtain
\begin{equation}\label{eq:factor_k}
	(\lambda+\sigma)\Big[(\lambda+\delta+\rho)(\lambda+1+k^2)(\lambda+\mu+D_ck^2)-\alpha r\Big]=0 .
\end{equation}
The first factor has the negative root $\lambda=-\sigma$, and the
remaining dynamics is governed by the cubic factor
\[
\lambda^3+\pi_2(k^2)\lambda^2+\pi_1(k^2)\lambda+\pi_0(k^2),
\]
where
\begin{subequations}\label{eq:qk}
	\begin{align}
		\pi_2(k^2) &= \delta + \rho + \mu + 1 + k^2(1 + D_c), \\
		\pi_1(k^2) &= (\delta + \rho)(1 + k^2 + \mu + D_c \, k^2) + (1 + k^2)(\mu + D_c k^2), \\
		\pi_0(k^2) &= (\delta + \rho)\,(1 + k^2)(\mu + D_c k^2) - \alpha r.
	\end{align}
\end{subequations}
For every $k$, the coefficients $\pi_2(k^2)$ and $\pi_1(k^2)$ are
strictly positive. Moreover,
\[
\pi_2(k^2)\pi_1(k^2)
>(\delta+\rho)(1+k^2)(\mu+D_ck^2).
\]
Consequently, whenever $\pi_0(k^2)>0$, the third
Routh--Hurwitz condition
$\pi_2(k^2)\pi_1(k^2)>\pi_0(k^2)$ is automatically satisfied.
The cubic factor is therefore Hurwitz if and only if
$\pi_0(k^2)>0$; it has a zero eigenvalue when $\pi_0(k^2)=0$ and at
least one eigenvalue with positive real part when $\pi_0(k^2)<0$.
In particular, a Hopf crossing is impossible: for a cubic with positive
$\pi_0$, such a crossing would require
$\pi_2(k^2)\pi_1(k^2)=\pi_0(k^2)$, which cannot occur here.

Spatially heterogeneous modes correspond to nonzero Laplacian modes
$k\neq0$; such a mode is unstable when $\operatorname{Re}\lambda>0$. Therefore,
diffusion-driven instability can only occur if there exists $k>0$ such that
$\pi_0(k^2)<0$ while $\pi_0(0)>0$. This is impossible, since
\[
\frac{d \pi_0}{d(k^2)}=(\delta+\rho)\,\big(\mu+D_c+2D_ck^2\big)>0 ,
\]
so that $\pi_0(k^2)>\pi_0(0)$ for every $k\neq0$. Equivalently, introducing the wavenumber-dependent reproduction number
\begin{equation}\label{eq:RDk}
	\mathcal{R}_D(k)=\frac{\alpha r}{(\delta+\rho)\,(1+k^2)(\mu+D_ck^2)}=\frac{\mathcal{R}_D}{(1+k^2)\,\big(1+D_ck^2/\mu\big)} ,
\end{equation}
the mode $k$ is unstable if and only if $\mathcal{R}_D(k)>1$, and $\mathcal{R}_D(k)<\mathcal{R}_D$ for every $k\neq0$: diffusion of the effective macrophage-activity density and signal removes them from the region where they are produced, and can only reduce the gain of the inflammatory loop. Consequently, neither diffusion-driven nor oscillatory instability can
occur in the linearization at $E_0$.

\subsection{Linear invasion threshold and candidate spreading speed}\label{subsec:invasion}

In this subsection, we characterize the conditions under which localized
damage can grow at the leading edge of an otherwise damage-free state. The
threshold obtained here is a threshold for linear instability and leading-edge
growth. 

We first examine the spatially homogeneous mode $k=0$. In this case,
the characteristic equation \eqref{eq:factor_k} reduces to
\begin{equation}\label{eq:factor_0}
	\det(\lambda I-J(E_0))=(\lambda+\sigma)P_0(\lambda)=0,
	\qquad
	P_0(\lambda)
	=(\lambda+\delta+\rho)(\lambda+1)(\lambda+\mu)-\alpha r,
\end{equation}
where $P_0$ is the characteristic polynomial of the active Jacobian
$J_{\mathrm{act}}^0$ introduced in Section~\ref{sec:local}.

The coefficients $\pi_2(0)$ and $\pi_1(0)$ defined in \eqref{eq:qk} are strictly
positive. Consequently,
\[
P_0'(\lambda)
=3\lambda^2+2\pi_2(0)\lambda+\pi_1(0)>0
\qquad(\lambda\ge0).
\]
Thus $P_0$ is strictly increasing on $[0,\infty)$ and has a positive
real root if and only if $P_0(0)<0$. This root, when it exists, is
unique and simple. The corresponding threshold condition is
\begin{equation}\label{thresh_inv}
	\mu(\delta+\rho)-\alpha r<0
	\qquad\Longleftrightarrow\qquad
	\mathcal R_D>1.
\end{equation}

For $\mathcal R_D<1$, all roots of $P_0$ have negative real parts, as
shown in Section~\ref{sec:local}. At $\mathcal R_D=1$,
\[
P_0(\lambda)
=\lambda\bigl(\lambda^2+\pi_2(0)\lambda+\pi_1(0)\bigr),
\]
so that zero is a simple root and the other two roots have negative real
parts. For $\mathcal R_D>1$, the unique positive root $\lambda_+$ gives
the growth rate of the unstable homogeneous active mode. By
Proposition~\ref{prop:E0}, it coincides with the spectral bound of
$J_{\mathrm{act}}^0$.

Thus the damage-reproduction threshold also determines when the homogeneous active mode becomes linearly unstable. Above this threshold, the linearized dynamics predicts growth of localized active perturbations. A leading-edge analysis of the linearized system then yields a candidate linear spreading speed. This approach is commonly used in the classical theory of pulled fronts, where the propagation speed is determined by the linear growth of small perturbations ahead of the front, rather than by its nonlinear interior \cite{vanSaarloos2003}.
However, the speed predicted by the linearization does not necessarily coincide with the speed selected by the full nonlinear system. Such an identification requires additional mathematical conditions, as established for certain cooperative systems \cite{WeinbergerLewisLi02,LiangZhao07}. In fact, while the linearized active subsystem is cooperative, the full nonlinear model is not, because of the presence of the term $\alpha mh$. This term increases active damage while simultaneously reducing functional tissue. Consequently, the speed obtained from the leading-edge analysis is treated as a candidate linear spreading speed rather than a rigorously established nonlinear propagation speed. 

For our analysis we consider exponential modes of the
linearized system of the form
\begin{equation}
	\mathbf w(x,t)=\mathbf v\,e^{-\gamma(x-st)},
\end{equation}
where $s>0$ represents the propagation speed, $\gamma>0$ is the spatial decay rate of the front and $\mathbf v=(\hat h,\hat d,\hat m,\hat c)^T$ is the amplitude vector. Substituting this ansatz into the linearized system, we obtain the eigenvalue problem
\begin{equation}
	\big(\gamma^2\mathscr D+J(E_0)-\gamma s\,I\big)\mathbf v=0 .
\end{equation}
Non-trivial solutions exist only when
$\det(\gamma^2\mathscr D+J(E_0)-\gamma s\,I)=0$. Due to the block structure
of the system matrices, the determinant factorizes as
\begin{equation}\label{det_vel}
	(\gamma s+\sigma)\,p_3(\gamma,s)=0,
	\qquad
	p_3(\gamma,s)=(\gamma s+\delta+\rho)\big(\gamma s+1-\gamma^2\big)\big(\gamma s+\mu-D_c\gamma^2\big)-\alpha r .
\end{equation}
The first factor corresponds to a perturbation of the functional tissue component alone. It yields $s=-\sigma/\gamma$ and is associated with the decaying temporal mode $-\sigma$.

In contrast, the cubic factor $p_3(\gamma,s)$ captures the coupled dynamics of damaged tissue, effective macrophage-activity density and signal at the leading edge. The corresponding linear speed branch is determined by the roots of $p_3(\gamma,s)=0$. For fixed $\gamma>0$, this is a cubic equation in $s$, whose roots are $s=\lambda/\gamma$, where $\lambda$ ranges over the eigenvalues of
\begin{equation}\label{eq:Bgamma}
	\widehat B(\gamma)=J_{\mathrm{act}}^0+\gamma^2\,\mathrm{diag}(0,1,D_c)=
	\begin{pmatrix}
		-(\delta+\rho) & \alpha & 0\\
		0 & -1+\gamma^2 & 1\\
		r & 0 & -\mu+D_c\gamma^2
	\end{pmatrix}.
\end{equation}
Exactly as for $J_{\mathrm{act}}^0$ in Proposition~\ref{prop:E0}, $\widehat B(\gamma)$ is an irreducible Metzler matrix, so its largest eigenvalue $\lambda_1(\gamma)$ is real and simple, and it is the only eigenvalue with a positive eigenvector $\mathbf v_\gamma$; it is therefore the only one admissible for a front of non-negative densities. It is the largest real root of
\begin{equation}\label{eq:lambda1}
	\big(\lambda+\delta+\rho\big)\big(\lambda+1-\gamma^2\big)\big(\lambda+\mu-D_c\gamma^2\big)=\alpha r .
\end{equation}
A physically meaningful invasion requires the existence of a positive propagation speed. Since the diagonal entries of $\widehat B(\gamma)$ are non-decreasing in $\gamma$, and the largest eigenvalue of an irreducible Metzler matrix is increasing in each entry, we have $\lambda_1(\gamma)\ge\lambda_1(0)=\lambda_{\max}(J_{\mathrm{act}}^0)$, which is positive under the invasion condition \eqref{thresh_inv}. Hence, for every $\gamma>0$, the equation $p_3(\gamma,s)=0$ admits
the positive root
\begin{equation}
	s_+(\gamma)
	=\frac{\lambda_1(\gamma)}{\gamma}
	=\frac{1}{\gamma}\,\lambda_{\max}\!\left(
	J_{\mathrm{act}}^0+\gamma^2\operatorname{diag}(0,1,D_c)
	\right)>0,
	\qquad \gamma>0,
	\label{eq:spectral_speed_branch}
\end{equation}
where $\lambda_{\max}(A)$ denotes the spectral bound of a matrix $A$.
This defines the linear spreading-speed branch.
This spectral formulation is independent of the explicit cubic expansion and
is the appropriate formulation for the linear leading-edge problem.

The candidate linear spreading speed is the minimal value of this branch:
\begin{equation}\label{frontvelocity}
	s^*=\min_{\gamma>0}s_+(\gamma).
\end{equation}
This is the classical marginal-stability construction for the linearized
leading edge. The following result shows that the minimum is attained and
describes its dependence on the parameters.

\begin{proposition}\label{prop:sstar}
	Let $\mathcal{R}_D>1$. Then:
	\begin{enumerate}
		\item[(i)] the minimum in \eqref{frontvelocity} is attained at some $\gamma^*\in(0,\infty)$, and $s^*>0$;
		\item[(ii)] $s^*$ depends on $\alpha$ and $r$ only through the product
		$\alpha r=\mathcal{R}_D(\delta+\rho)\mu$. At fixed $\eta$, $\mu$,
		and $D_c$, it is strictly increasing in $\mathcal{R}_D$; at fixed
		$\eta$, $\mu$, and $\mathcal{R}_D$, it is strictly increasing in $D_c$.
		At fixed $\alpha$, $r$, $\mu$, and $D_c$, it is strictly decreasing
		in $\eta=\delta+\rho$ on the range where $\mathcal R_D>1$; in
		particular, it is strictly decreasing in each of $\delta$ and $\rho$
		when the other is held fixed;
		\item[(iii)] $s^*\le2\sqrt{\max\{1,D_c\}\;\lambda_{\max}(J_{\mathrm{act}}^0)}$. In particular, $s^*\to0$ as $\mathcal{R}_D\to1^+$ when $\alpha$ varies and $r$, $\delta$, $\rho$, $\mu$, and $D_c$ are fixed.
		\item[(iv)] every minimizer $\gamma^*$ satisfies the marginal-stability
		condition
		\begin{equation*}
			\lambda_1'(\gamma^*)=\frac{\lambda_1(\gamma^*)}{\gamma^*}=s^*,
		\end{equation*}
		or, equivalently,
		\begin{equation*}
			p_3(\gamma^*,s^*)=0,
			\qquad
			\partial_\gamma p_3(\gamma^*,s^*)=0,
		\end{equation*}
		where the partial derivative is taken at fixed $s$.
	\end{enumerate}
\end{proposition}

\begin{proof}
	(i) The function $\gamma\mapsto\lambda_1(\gamma)$ is continuous, as a simple eigenvalue depending continuously on the entries. As $\gamma\to0^+$, $\lambda_1(\gamma)\to\lambda_{\max}(J_{\mathrm{act}}^0)>0$, so $s_+(\gamma)\to+\infty$. As $\gamma\to\infty$, the largest eigenvalue of a Metzler matrix is not smaller than any of its diagonal entries, so $\lambda_1(\gamma)\ge\gamma^2-1$ and again $s_+(\gamma)\to+\infty$. Hence the infimum of $s_+$ is a minimum, attained at some $\gamma^*\in(0,\infty)$, and it is positive, since $s_+(\gamma)>0$ for every $\gamma$.
	
	(ii) By \eqref{eq:lambda1}, $\lambda_1(\gamma)$ depends on $\alpha$
	and $r$ only through $\alpha r$. With $\eta$, $\mu$, and $D_c$ fixed,
	increasing $\mathcal R_D$ increases this product. Equivalently, we may hold
	$r$ fixed and increase $\alpha$; irreducibility then implies a strict
	increase in $\lambda_1(\gamma)$ for every $\gamma>0$. With $\eta$, $\mu$,
	and $\mathcal R_D$ fixed, increasing $D_c$ strictly increases the $(3,3)$
	entry of $\widehat B(\gamma)$ for every $\gamma>0$ and hence also strictly
	increases $\lambda_1(\gamma)$. For either parameter, if $s_{+,2}>s_{+,1}$
	pointwise and $\gamma_2^*$ minimizes $s_{+,2}$, then
	\[
	s_2^*=s_{+,2}(\gamma_2^*)>s_{+,1}(\gamma_2^*)\ge s_1^*,
	\]
	where the minima exist by (i). Finally, at fixed $\alpha$, $r$, $\mu$, and
	$D_c$, increasing $\eta$ strictly decreases the $(1,1)$ entry of
	$\widehat B(\gamma)$ and hence $s_+(\gamma)$ for every $\gamma>0$. If
	$\eta_1<\eta_2$ both give $\mathcal R_D>1$ and $\gamma_1^*$ minimizes the
	branch at $\eta_1$, then
	$$
	s^*(\eta_2)\le s_+(\gamma_1^*;\eta_2)
	<s_+(\gamma_1^*;\eta_1)=s^*(\eta_1).
	$$
	This also gives the stated monotonicity in $\delta$ and $\rho$.
	
	(iii) Entrywise, $\widehat B(\gamma)\le J_{\mathrm{act}}^0+\max\{1,D_c\}\gamma^2I$, so that $\lambda_1(\gamma)\le\lambda_{\max}(J_{\mathrm{act}}^0)+\max\{1,D_c\}\gamma^2$. Minimizing $\big(\lambda_{\max}(J_{\mathrm{act}}^0)+\max\{1,D_c\}\gamma^2\big)/\gamma$ over $\gamma>0$ gives the bound. Finally, when $\alpha$ varies with $r$, $\delta$, $\rho$, $\mu$, and $D_c$ fixed, $\mathcal R_D\to1^+$ implies $\alpha\to(\delta+\rho)\mu/r$. Proposition~\ref{prop:E0} and continuity give $\lambda_{\max}(J_{\mathrm{act}}^0)\to0$, so the preceding bound yields $s^*\to0$.
	
	(iv) Since $\widehat B(\gamma)$ is an irreducible Metzler matrix with a
	simple dominant eigenvalue, $\lambda_1$ is continuously differentiable on
	$(0,\infty)$. At an interior minimizer of
	$s_+(\gamma)=\lambda_1(\gamma)/\gamma$, differentiation gives
	$\gamma^*\lambda_1'(\gamma^*)-\lambda_1(\gamma^*)=0$, which is the first
	condition in (iv). The relation $p_3(\gamma^*,s^*)=0$ follows from the
	definition of the branch. Since the corresponding eigenvalue is simple,
	$\partial_s p_3(\gamma^*,s^*)\ne0$; differentiating
	$p_3(\gamma,s_+(\gamma))=0$ and using $s_+'(\gamma^*)=0$ therefore gives
	$\partial_\gamma p_3(\gamma^*,s^*)=0$.
\end{proof}

In analogy with Fisher--KPP type systems, the minimization of the linear
leading-edge branch suggests a pulled-front mechanism
\cite{vanSaarloos2003}.

\subsection{Nonlinear upper spreading bound without chemotaxis}

For $\mathcal R_D>1$ and $\chi_0=0$, a comparison argument yields decay
outside every cone moving faster than $s^*$ for compactly supported active
initial data.

\begin{theorem}[Global upper spreading bound]\label{thm:upper_speed}
	Let $\mathcal R_D>1$ and $\chi_0=0$, and consider
	\eqref{eq:dimensionless} on $\mathbb R$.
	Assume that
	\[
	U_{in}=(h_{in},d_{in},m_{in},c_{in})\in BUC(\mathbb R)^4,
	\]
	where $BUC(\mathbb R)$ is the Banach space of bounded uniformly
	continuous functions, and that
	\[
	h_{in},d_{in},m_{in},c_{in}\ge0,\qquad
	h_{in}+d_{in}\le1.
	\]
	If $(d_{in},m_{in},c_{in})$ has compact support, then the unique global
	solution furnished by Lemma~\ref{app:whole_line_wp} satisfies, for every
	$s>s^*$,
	\[
	\lim_{t\to\infty}\ \sup_{|x|\ge st}\big(d+m+c\big)(t,x)=0 .
	\]
\end{theorem}

\begin{proof}
	By Lemma~\ref{app:whole_line_wp}, the solution exists globally, is
	non-negative and bounded, preserves $h+d\le1$, and has the classical
	regularity required in the comparison argument below.
	Set $D_0=\operatorname{diag}(0,1,D_c)$ and define the
	partially degenerate diffusion operator by
	\[
	\mathcal A_{\rm diff}(z_1,z_2,z_3)^T=(0,\partial_{xx}z_2,
	D_c\partial_{xx}z_3)^T.
	\]
	Since
	$0\le h\le1$ and $c/(1+c)\le c$, the vector
	$\mathbf u=(d,m,c)^T$ satisfies, componentwise and pointwise,
	\[
	\partial_t\mathbf u\le J_{\mathrm{act}}^0\mathbf u+\mathcal A_{\rm diff}\mathbf u .
	\]
	For every $\gamma>0$, the function
	\[
	\overline{\mathbf u}(t,x)
	=A e^{-\gamma(x-s_+(\gamma)t)}\mathbf v_\gamma
	\]
	is an exact solution of
	$\partial_t\overline{\mathbf u}=J_{\mathrm{act}}^0\overline{\mathbf u}
	+\mathcal A_{\rm diff}\overline{\mathbf u}$. Here
	$s_+(\gamma)=\lambda_1(\gamma)/\gamma$, and $\mathbf v_\gamma\gg0$
	is the Perron eigenvector of $\widehat B(\gamma)$. Since the active
	initial data are bounded and compactly supported, $A>0$ can be chosen so
	that $\overline{\mathbf u}(0,x)\ge\mathbf u(0,x)$ for every $x\in\mathbb R$.
	
	For completeness, we prove the comparison needed below, including the
	non-diffusing first component. Put
	$W=\overline{\mathbf u}-\mathbf u$. Then
	\[
	\partial_tW-\mathcal A_{\rm diff}W-J_{\mathrm{act}}^0W\ge0,
	\qquad W(0,\cdot)\ge0,
	\]
	and $W$ is bounded from below on every strip $[0,T]\times\mathbb R$.
	Let $\mathbf 1=(1,1,1)^T$, let $\psi(x)=1+x^2$, and choose
	\[
	\beta>\max_{1\le i\le3}(J_{\mathrm{act}}^0\mathbf 1)_i
	+2\max\{1,D_c\}.
	\]
	For $\varepsilon>0$, define
	\[
	Z(t,x)=W(t,x)+\varepsilon e^{\beta t}\psi(x)\mathbf 1.
	\]
	Since $\partial_{xx}\psi=2$, we have
	\begin{equation}\label{eq:strict_comparison_upper_speed}
		\partial_tZ-\mathcal A_{\rm diff}Z-J_{\mathrm{act}}^0Z
		\ge \varepsilon e^{\beta t}
		\left\{\psi(x)(\beta\mathbf 1-J_{\mathrm{act}}^0\mathbf 1)-2D_0\mathbf 1\right\}
		\gg0.
	\end{equation}
	Fix $T>0$. At $t=0$ all components of $Z$ are strictly positive, and,
	because $W$ is bounded from below on $[0,T]\times\mathbb R$,
	$Z(t,x)\to+\infty$ componentwise as $|x|\to\infty$, uniformly for
	$t\in[0,T]$. If a component of $Z$ became non-positive on this strip,
	there would therefore be a first contact $(t_0,x_0)$ with zero, in a
	component $i$, with $t_0>0$. At that point,
	\[
	Z(t_0,x_0)\ge0,\qquad Z_i(t_0,x_0)=0,\qquad
	\partial_tZ_i(t_0,x_0)\le0,
	\]
	and, when $(D_0)_{ii}>0$, also
	$\partial_{xx}Z_i(t_0,x_0)\ge0$; when $(D_0)_{ii}=0$, the latter term is
	absent. Since $J_{\mathrm{act}}^0$ is Metzler,
	$(J_{\mathrm{act}}^0Z)_i(t_0,x_0)=
	\sum_{j\ne i}(J_{\mathrm{act}}^0)_{ij}Z_j(t_0,x_0)\ge0$.
	Thus \eqref{eq:strict_comparison_upper_speed} yields
	\[
	\partial_tZ_i(t_0,x_0)
	>(\mathcal A_{\rm diff}Z)_i(t_0,x_0)+(J_{\mathrm{act}}^0Z)_i(t_0,x_0)
	\ge0,
	\]
	a contradiction. Hence $Z\ge0$ on $[0,T]\times\mathbb R$.
	Letting $\varepsilon\downarrow0$ gives $W\ge0$ there; since $T$ is
	arbitrary, $\mathbf u\le\overline{\mathbf u}$ for all $t\ge0$.
	
	Choose now $\gamma=\gamma^*$, so that $s_+(\gamma^*)=s^*$. For
	$x\ge st$ and $s>s^*$,
	\[
	0\le\mathbf u(t,x)
	\le A e^{-\gamma^*(x-s^*t)}\mathbf v_{\gamma^*}
	\le A e^{-\gamma^*(s-s^*)t}\mathbf v_{\gamma^*},
	\]
	which converges to zero componentwise. Repeating the argument for the
	reflected solution $x\mapsto-x$ gives the same conclusion on
	$x\le-st$. The asserted two-sided estimate follows.
\end{proof}

The whole-line construction in Lemma~\ref{app:whole_line_wp} is independent
of the bounded-domain well-posedness theory of Section~\ref{wp:main_section}.
It uses essentially the absence of chemotaxis: when $\chi_0=0$, the two
parabolic equations are semilinear heat equations and the invariant tissue
simplex supplies uniform bounds for their source terms.

In classical monostable reaction--diffusion problems, travelling waves and
spreading speeds are often linked through additional comparison and
monotonicity properties \cite{HadelerRothe75,LiangZhao07}.

\begin{remark}[Chemotaxis and pulled-front dynamics]
	\label{rem:pulled_conjecture}
	The leading-edge analysis suggests a pulled-front scenario, in which
	sufficiently localized perturbations with a nontrivial active component
	spread at the candidate linear speed $s^*$. Chemotactic transport acts
	nonlinearly and may modify the front profile and, potentially, the speed
	selected by the full system. Whether nonlinear propagation is governed by
	the leading edge, and how chemotaxis affects this behaviour, are explored
	numerically in Section~\ref{sec_num}.
\end{remark}

\begin{remark}[Biological interpretation of the linear spreading speed]
	By Proposition~\ref{prop:sstar}, the candidate speed depends on the feedback
	strength through $\mathcal R_D$ and on the signal diffusivity $D_c$. Signal
	spreading ahead of the damaged region promotes recruitment or activation of
	the effective macrophage activity in functional tissue. For fixed remaining
	parameters, $s^*$ is strictly increasing in $D_c$. It tends to zero as
	$\mathcal R_D\to1^+$ when $\alpha$ varies and the other parameters are fixed. The candidate linear speed is unaffected by the split
	of $\mathcal R_D$ between the intensity of immune-mediated damage $\alpha$
	and the production of signal $r$.
\end{remark}
\section{Numerical simulations}\label{sec_num}
In this Section we present some numerical simulations of the full nonlinear system to
validate the obtained analytical results. First, we describe
the numerical implementation and parameter choices, then we test threshold
dynamics depending on the parameter $\mathcal{R}_D$, the invasion profiles, and the measured front speed.
Next, we examine the dependence of the invasion speed on the  model
parameters and assess the influence of chemotaxis on the evolving front.

The numerical simulations were implemented in
FEniCS~\cite{LoggWells2010,LoggMardalWells2012}, using Finite Element discretization with continuous,
piecewise-linear finite elements  for all four
variables on a uniform spatial mesh. Time integration employed
the fully implicit backward Euler method, with an adaptive time
step. Temporal error was estimated by comparing
one full step with two half steps, using relative and absolute
tolerances of $10^{-3}$ and $10^{-6}$, respectively.
The coupled nonlinear systems were solved by Newton's method,
with relative and absolute tolerances
of $10^{-9}$ and $10^{-10}$.
Homogeneous no-flux boundary conditions were imposed. All results have been tested to be mesh-independent at the scales resolved in the simulations.

\subsection{Parameter interpretation and numerical choices}\label{parameters}

The dimensionless parameters of the model are defined in \eqref{eq:groups} and their interpretation is discussed in Section~\ref{subsec:dimensional}.
Concerning the  parameter values  used in the numerical experiments, since there is no directly calibration from DMD-specific experimental data, the simulations are designed to explore different dynamical regimes rather than to reproduce the progression of a particular patient or experimental cohort.
Table~\ref{tab:parameter_estimation} summarizes the parameter ranges considered in the numerical analysis and the rationale for their selection. These ranges are phenomenological and are chosen to investigate the effects of damage amplification, tissue repair, signal turnover and spatial transport. The damage-amplification threshold $\mathcal R_D$, defined in \eqref{eq:RD}, provides the main reference for the numerical experiments. By varying $\alpha$ or $r$ while keeping the remaining parameters fixed, we examine the transition between damage resolution and persistence. The transport parameters $D_c$, $\chi_0$ and $\kappa$, which do not enter the homogeneous threshold, are varied separately to investigate their influence on spatial propagation and chemotactic dynamics.
Finally, the ratio $r/\mu$ characterizes the signal level associated with active damage in the homogeneous quasi-steady-state approximation and provides a useful measure of the relative strength of signal production and clearance.

\begin{table}[!htbp]
	\centering
	\caption{Status of parameter estimation and values used for numerical
		exploration. A range marked phenomenological is a modelling choice designed
		to explore model behaviour, not a measurement.}
	\label{tab:parameter_estimation}
	\begin{tabular}{@{}L{0.09\linewidth}L{0.21\linewidth}L{0.24\linewidth}L{0.36\linewidth}@{}}
		\toprule
		Parameter & Estimation status & Numerical choice & Rationale \\
		\midrule
		$\mu$ & Phenomenological range motivated by a timescale ordering & $5$--$20$ & Signal turnover is taken faster than effective macrophage-activity turnover; no direct DMD-specific ratio is identified. \\
		$\delta$ & Phenomenological effective rate & $0.1$--$1$ & No direct DMD-specific clearance rate for the aggregate damage compartment is available. \\
		$\rho$ & Phenomenological effective rate & $0.1$--$2$ & Direct repair of actively damaged tissue is represented by one effective rate. \\
		$\sigma$ & Phenomenological effective rate & $0.05$--$0.5$ & Homeostatic reconstruction is not directly identified by the available data. \\
		$D_c$ & Phenomenological range motivated by faster signal transport & $10$--$100$ & Signal transport is taken faster than effective macrophage-activity motility; the ratio is an effective tissue-scale range, not a direct estimate. \\
		$r$ & Phenomenological feedback parameter & $0.1$--$10$ & Damage-induced signal production and recruitment are aggregated into one dimensionless gain. \\
		$\alpha$ & Phenomenological, varied through threshold & values chosen through $\mathcal R_D$ & The dystrophin-linked susceptibility is not directly identifiable as a single $a_0$ from the cited experiments. \\
		$\chi_0$ & Phenomenological chemotactic strength & $0$ and positive test values & The zero case is a mathematical limiting case used by the comparison theorem; positive values explore directed migration. \\
		$\kappa$ & Phenomenological saturation scale & $0.1$--$10$, reference value $1$ & No direct estimate of $k_\chi/k_M$ is available; the interval tests weak to strong attenuation. \\
		\bottomrule
	\end{tabular}
\end{table}

\paragraph{Interpretation of the computational domain.}
The spatial variable is nondimensionalized using the effective macrophage-activity motility length:
$$
\ell_M = \sqrt{\frac{D_M}{\mu_M}},
$$
so that a dimensionless domain of size $\ell$ corresponds to a physical length:
$$
\ell_{\mathrm{phys}} = \ell\,\ell_M.
$$

For a one-dimensional domain, it represents a scaled portion of muscle tissue along which damage propagates. Its length is chosen to be large compared with the characteristic front width, allowing the propagation speed to be measured before the advancing front is affected by the boundaries.

\subsection{Numerical investigation of threshold dynamics and front propagation speed}

In this section, we present the numerical outcomes of the simulations to examine the dynamics and compare it with the prediction of the analytical results. We first illustrate how a localized damage resolves below the threshold $\mathcal R_D=1$ while develops into a propagating front above it. We then assess the sensitivity of the front profiles and measured speeds to the initial perturbation, checking the numerical speeds against the leading-edge prediction $s^*$ defined in \eqref{frontvelocity}, and investigating the effects of model parameters on front propagation. Finally, we show how different parameter combinations can yield the same linear spreading prediction while producing different levels of persistent damage behind the front. We shall consider the chemotaxis-free case $\chi_0=0$, for which a nonlinear upper spreading bound is available. 

We begin by comparing the evolution of localized damage on either side of the threshold. Figure~\ref{fig:threshold} illustrates both regimes, starting from a Gaussian perturbation of the $d$-component localized near the left boundary of the computational domain.
For $\mathcal R_D=0.99$, below the threshold, the active components decay and the solution approaches the damage-free equilibrium $E_0$. For $\mathcal R_D=1.01$ and $\mathcal R_D=2$, above the threshold, the simulations instead reveal advancing profiles, with a region of persistent damage behind the front and the damage-free state ahead of it.

\begin{figure}[!htbp]
	\centering
	\begin{subfigure}[t]{0.22\linewidth}
		\includegraphics[width=\linewidth]{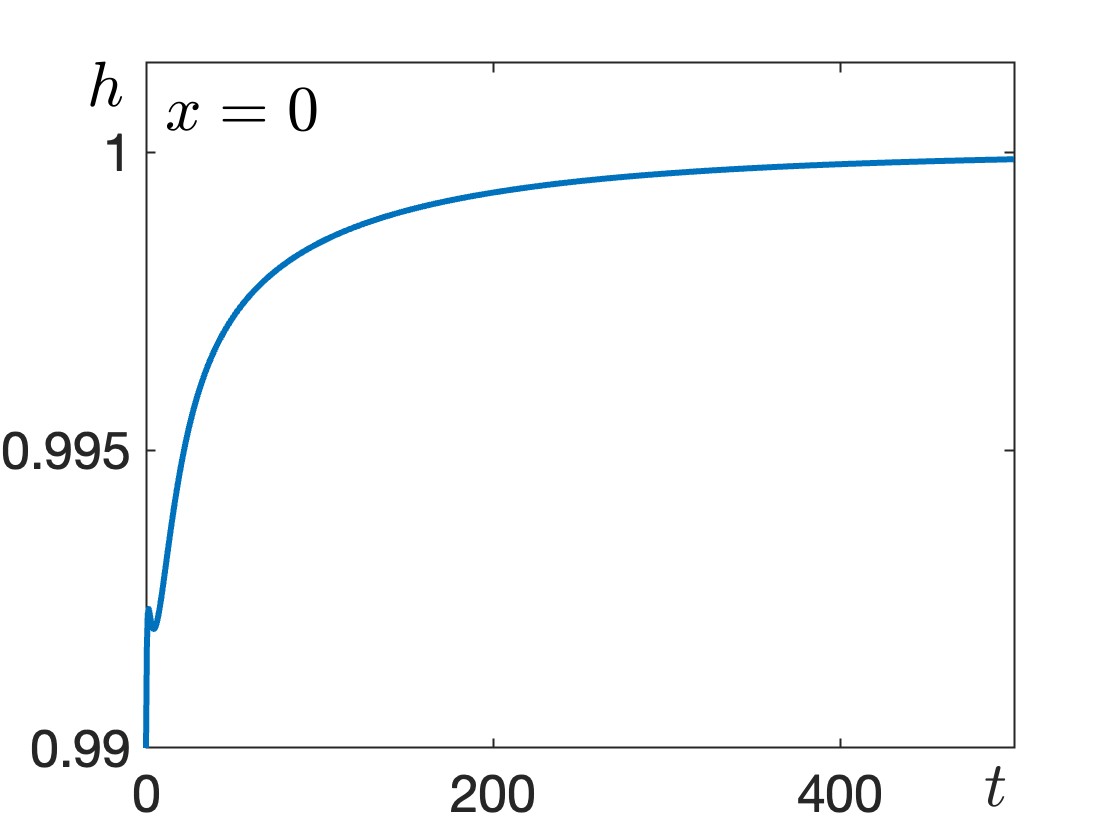}
		\caption{}
	\end{subfigure}	\begin{subfigure}[t]{0.22\linewidth}
		\includegraphics[width=\linewidth]{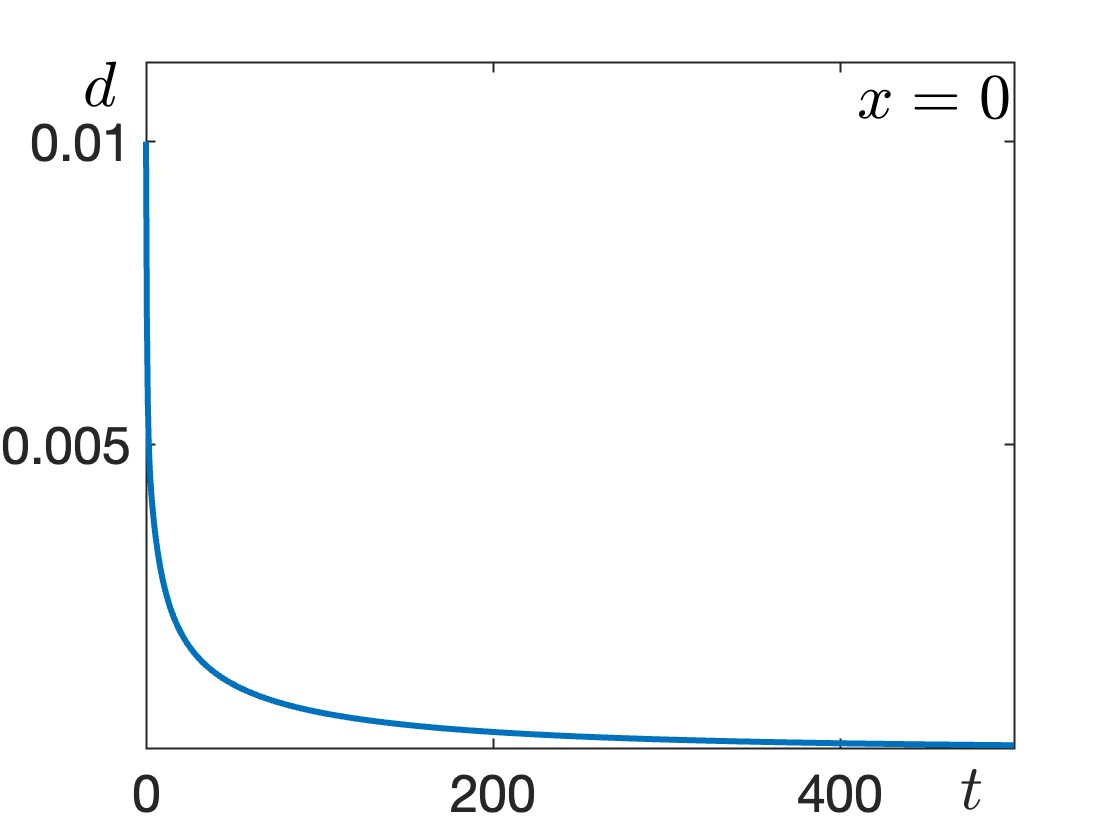}
		\caption{}
	\end{subfigure}	\begin{subfigure}[t]{0.22\linewidth}
		\includegraphics[width=\linewidth]{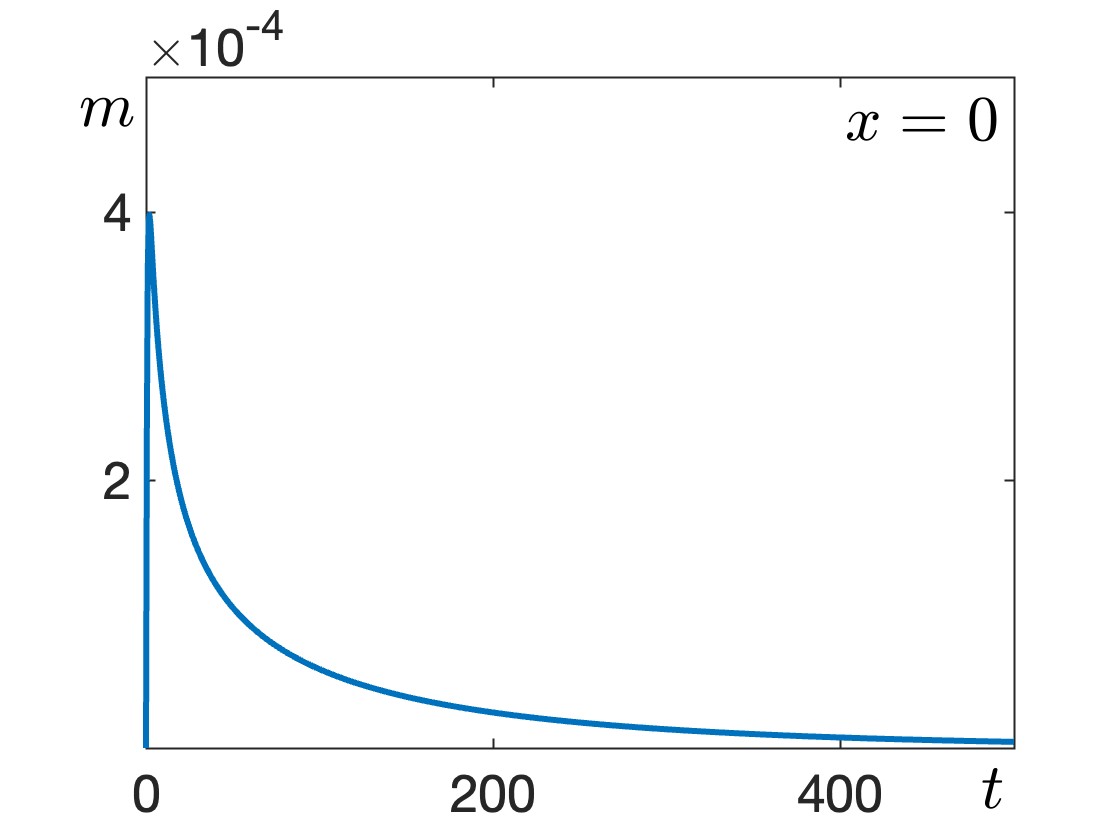}
		\caption{}
	\end{subfigure}	\begin{subfigure}[t]{0.22\linewidth}
		\includegraphics[width=\linewidth]{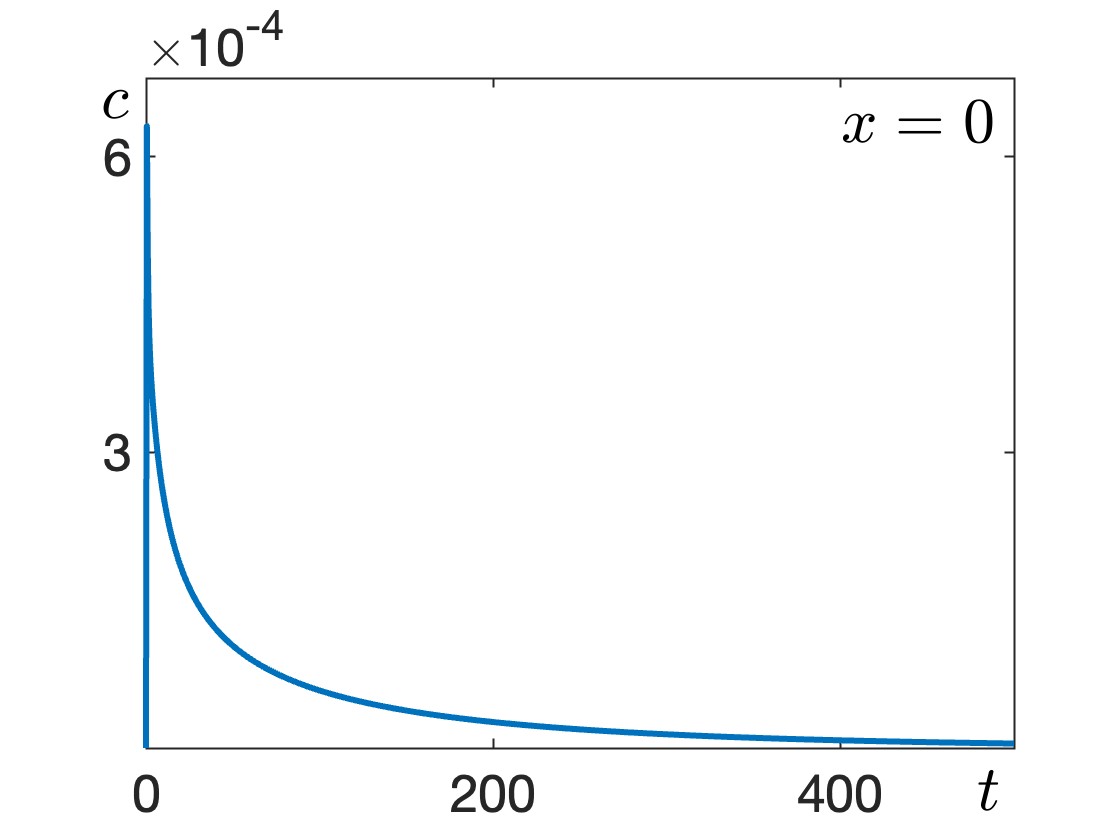}
		\caption{}
	\end{subfigure}
	\hfill
	\begin{subfigure}[t]{0.22\linewidth}
		\includegraphics[width=\linewidth]{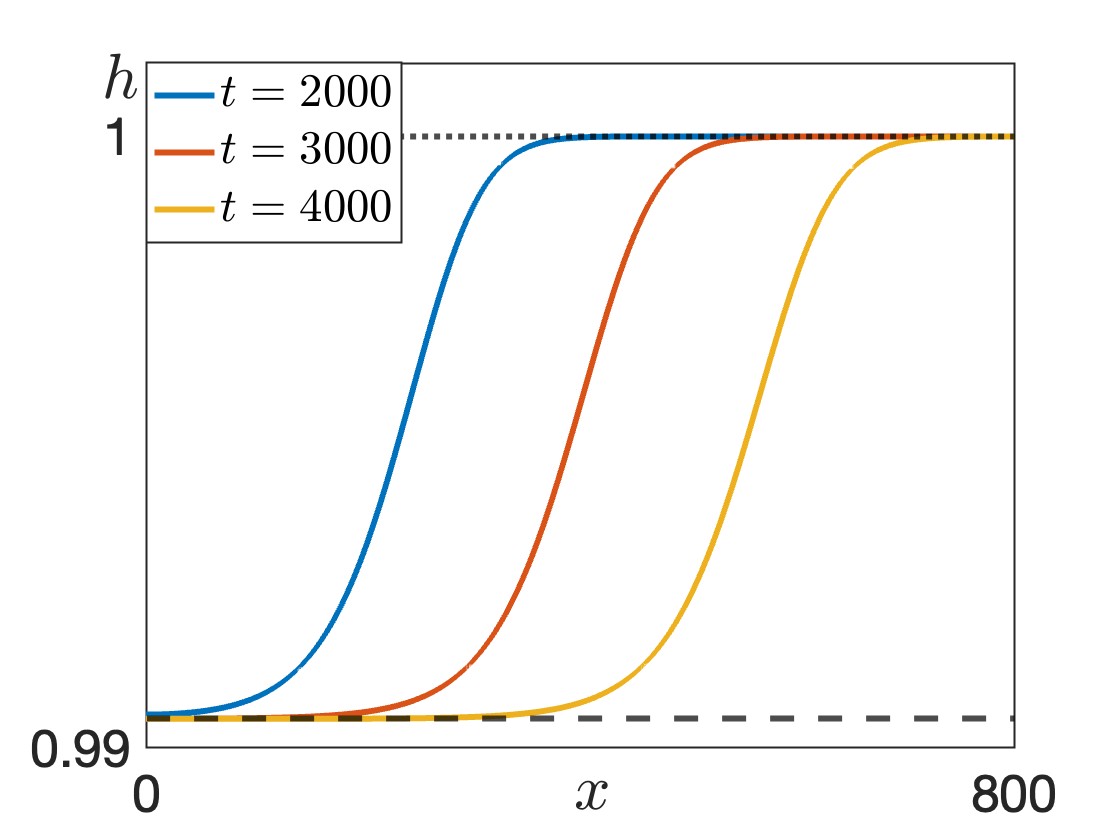}
		\caption{}
	\end{subfigure}	\begin{subfigure}[t]{0.22\linewidth}
		\includegraphics[width=\linewidth]{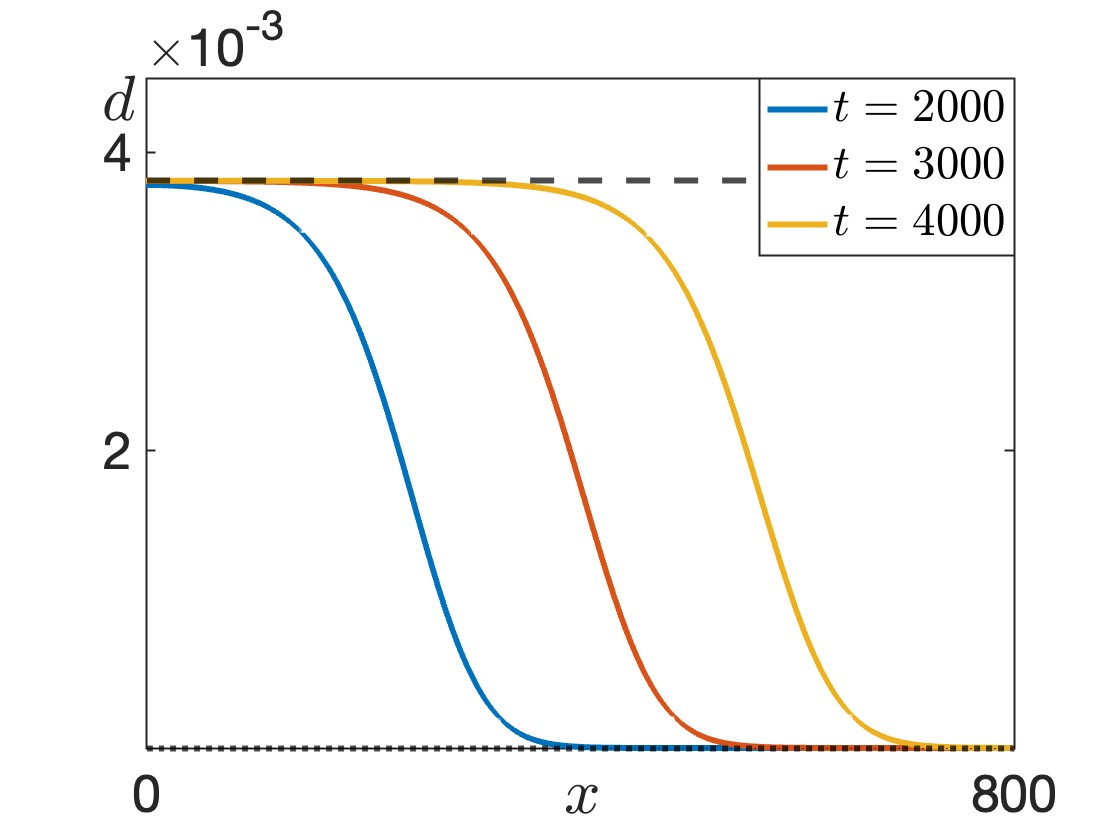}
		\caption{}
	\end{subfigure}	\begin{subfigure}[t]{0.22\linewidth}
		\includegraphics[width=\linewidth]{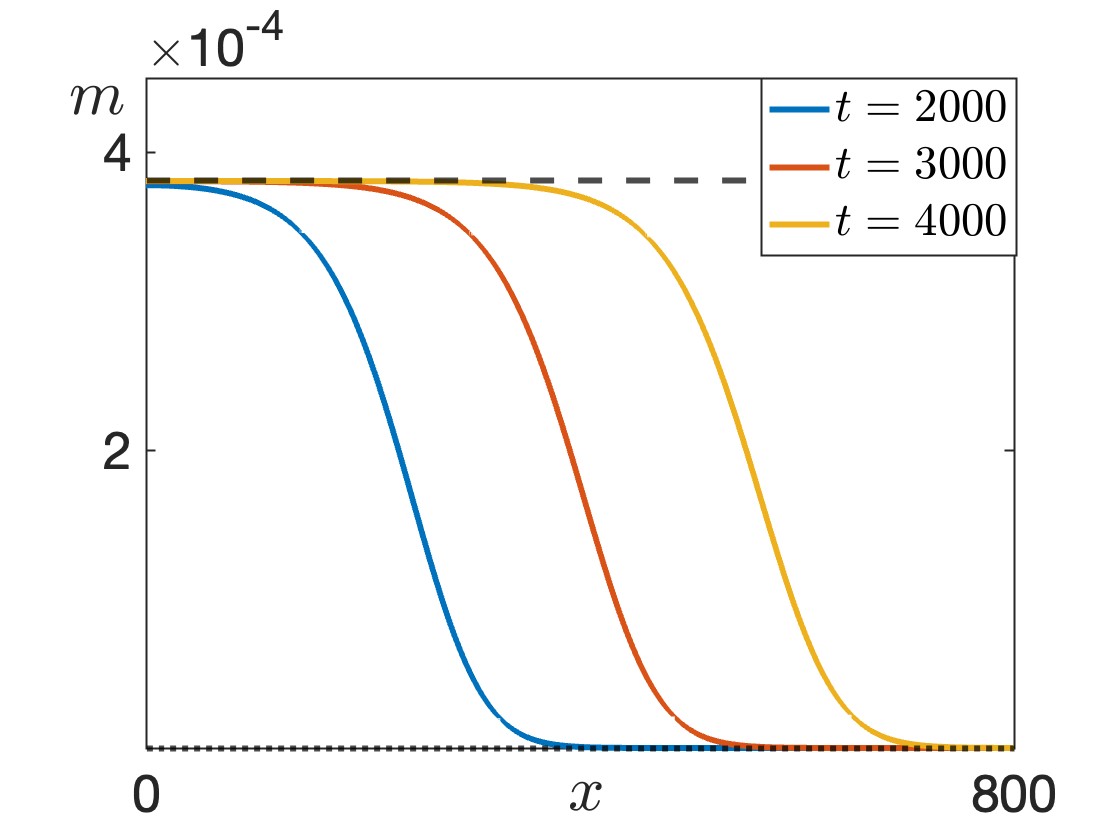}
		\caption{}
	\end{subfigure}
	\begin{subfigure}[t]{0.22\linewidth}
		\includegraphics[width=\linewidth]{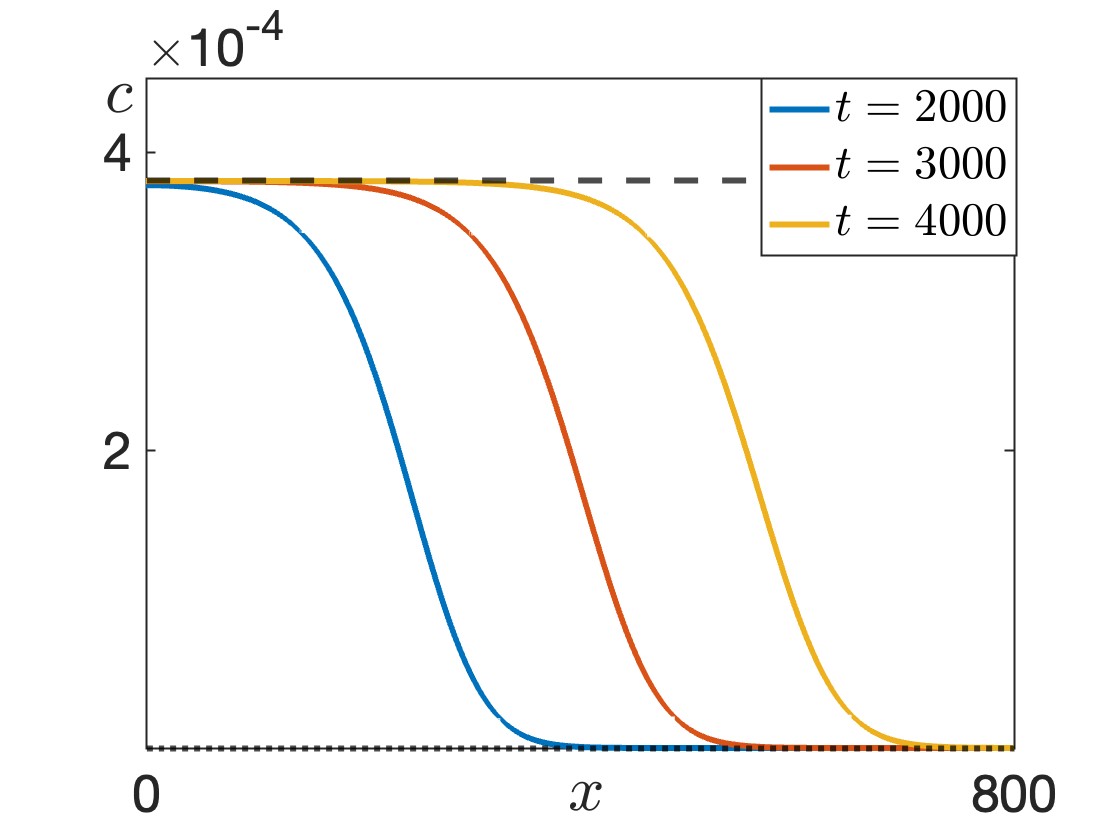}
		\caption{}
	\end{subfigure}
	\hfill
	\begin{subfigure}[t]{0.22\linewidth}
		\includegraphics[width=\linewidth]{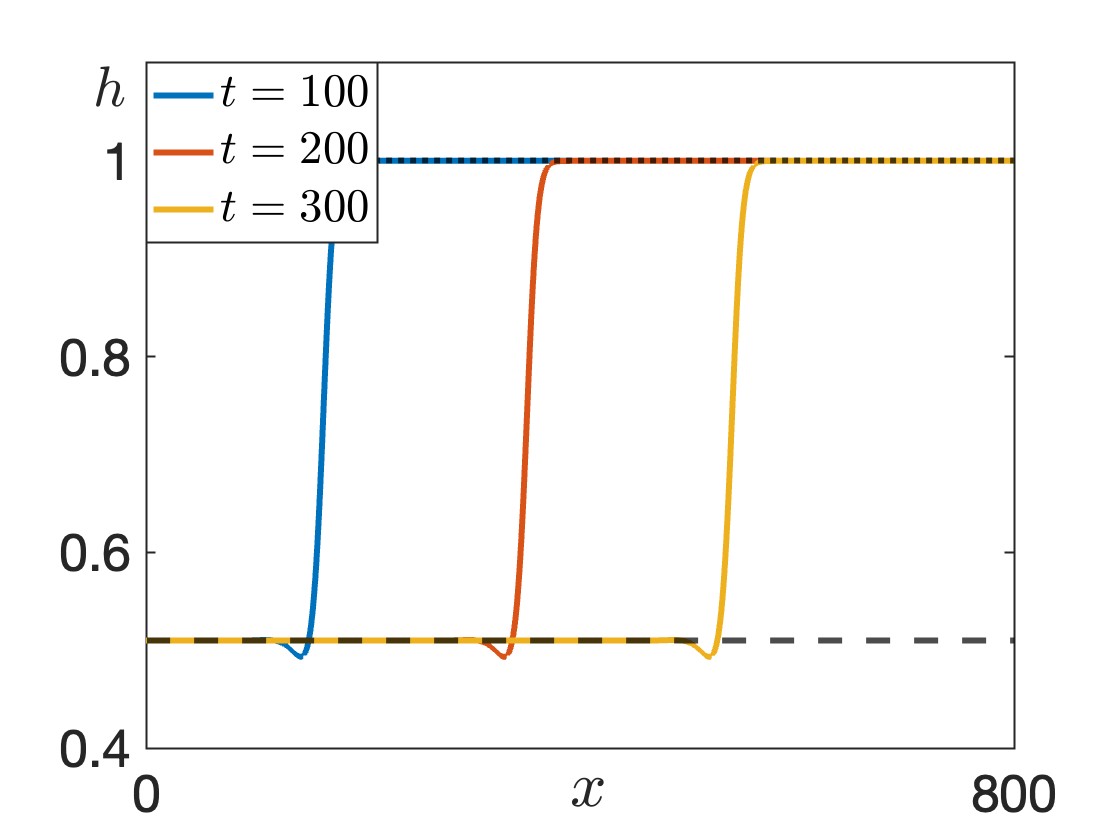}
		\caption{}
	\end{subfigure}	\begin{subfigure}[t]{0.22\linewidth}
		\includegraphics[width=\linewidth]{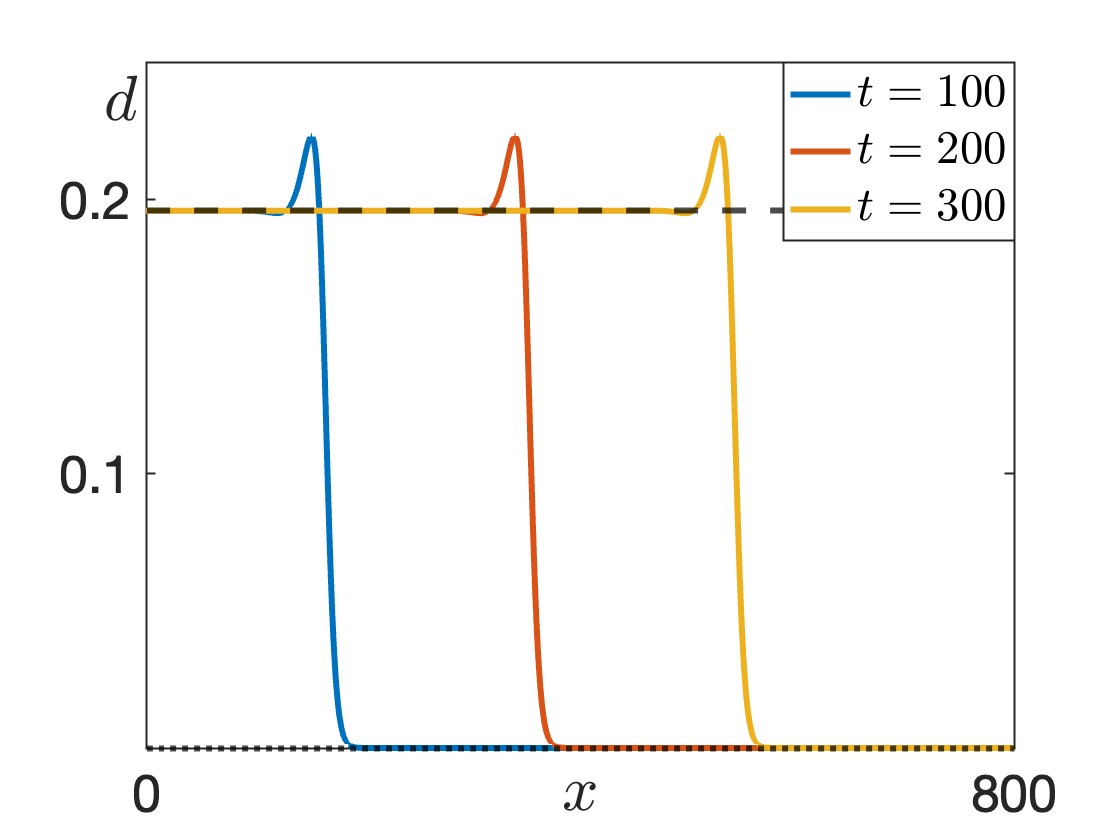}
		\caption{}
	\end{subfigure}	\begin{subfigure}[t]{0.22\linewidth}
		\includegraphics[width=\linewidth]{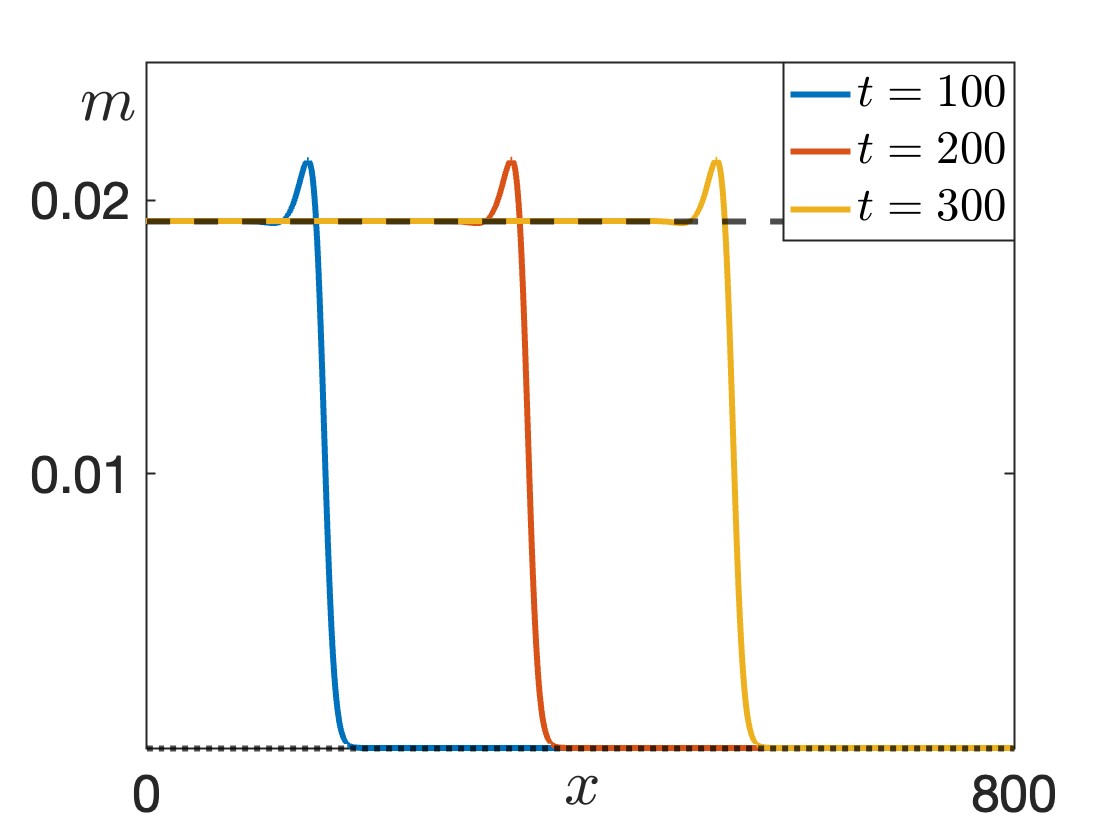}
		\caption{}
	\end{subfigure}
	\begin{subfigure}[t]{0.22\linewidth}
		\includegraphics[width=\linewidth]{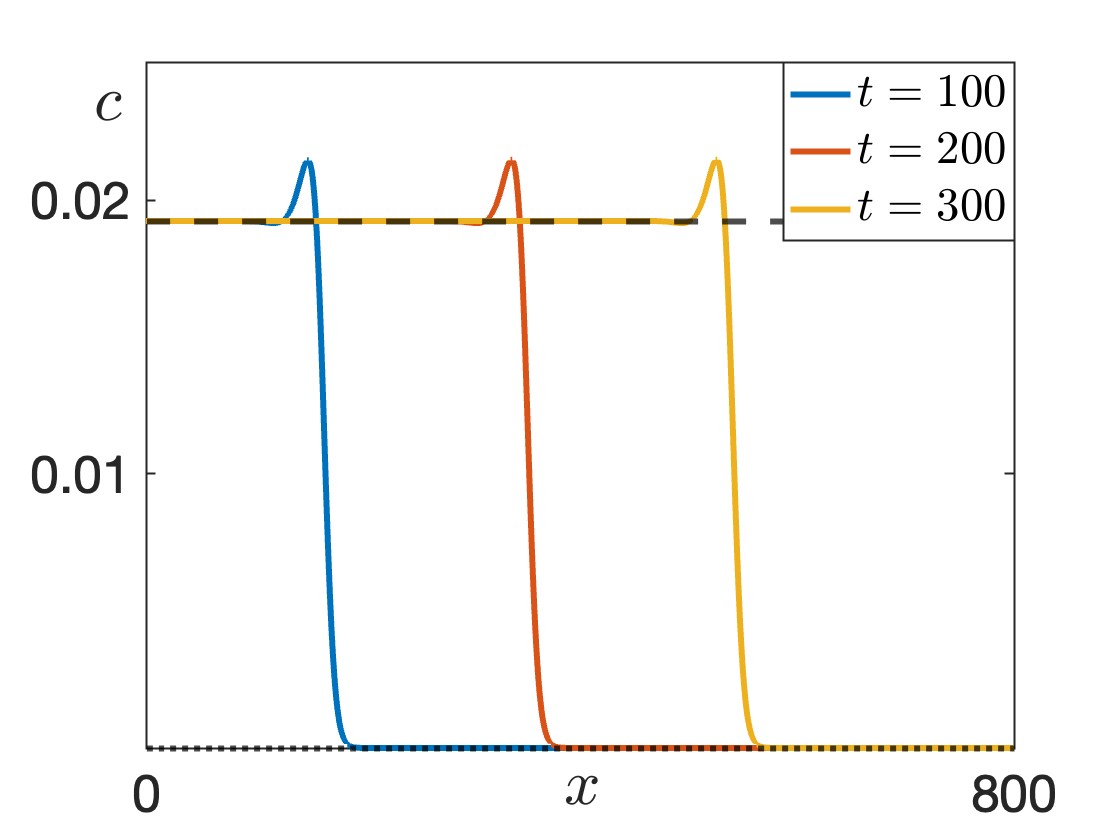}
		\caption{}
	\end{subfigure}
	\caption{Threshold and invasion dynamics for $\chi_0=0$, with
		$\sigma=0.2$, $\delta=0.3$, $\mu=10$, $D_c=30,$ $r=1$, $\rho=0.5$ and $\ell=800$.
		In each case $\alpha=8\mathcal R_D$. Panels~(a)--(d) display the four componentwise time traces below threshold, for
		$\mathcal R_D=0.99$;
		panels~(e)--(h) and (i)--(l) give the corresponding spatial profiles above
		threshold for $\mathcal R_D=1.01$ and $\mathcal R_D=2$, respectively.}
	\label{fig:threshold}
\end{figure}

We now show that different initial $d$-profiles lead to similar dynamical outcomes. We consider three initializations having the same maximal damage initial intensity: the Gaussian profile used for the simulations in Figure~\ref{fig:threshold}, a compactly supported profile, and one aligned with the dominant eigenvector of the linearized active subsystem, that is an initial
$d$-profile with an exponentially decaying perturbation. We denote by $x_f(t)$ the tracked front position and by $s_{\mathrm{num}}$ the speed estimated from its evolution after the initial transient. Figure~\ref{fig:D1_profile_robustness} shows the $d
$-profiles in the moving coordinate $x-s_{\mathrm{num}}t$ at three values of $\mathcal R_D>1$. The close overlap suggests that, in the tested regimes, the observed front shape is largely insensitive to these choices of initialization. 

\begin{figure}[!htbp]
	\centering
	\begin{subfigure}[t]{0.32\linewidth}
		\centering
		\includegraphics[width=\linewidth]{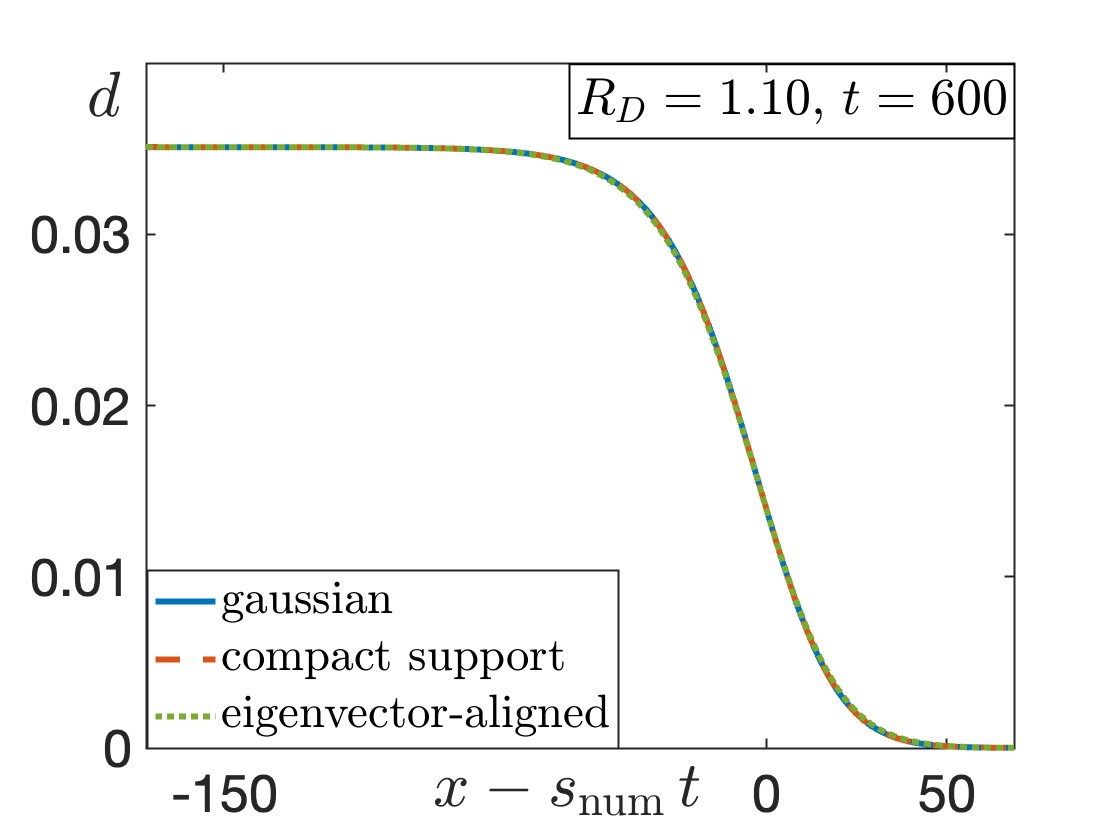}
		\caption{$\mathcal R_D=1.10$}
		\label{fig:D1_profile_RD110}
	\end{subfigure}
	\hfill
	\begin{subfigure}[t]{0.32\linewidth}
		\centering
		\includegraphics[width=\linewidth]{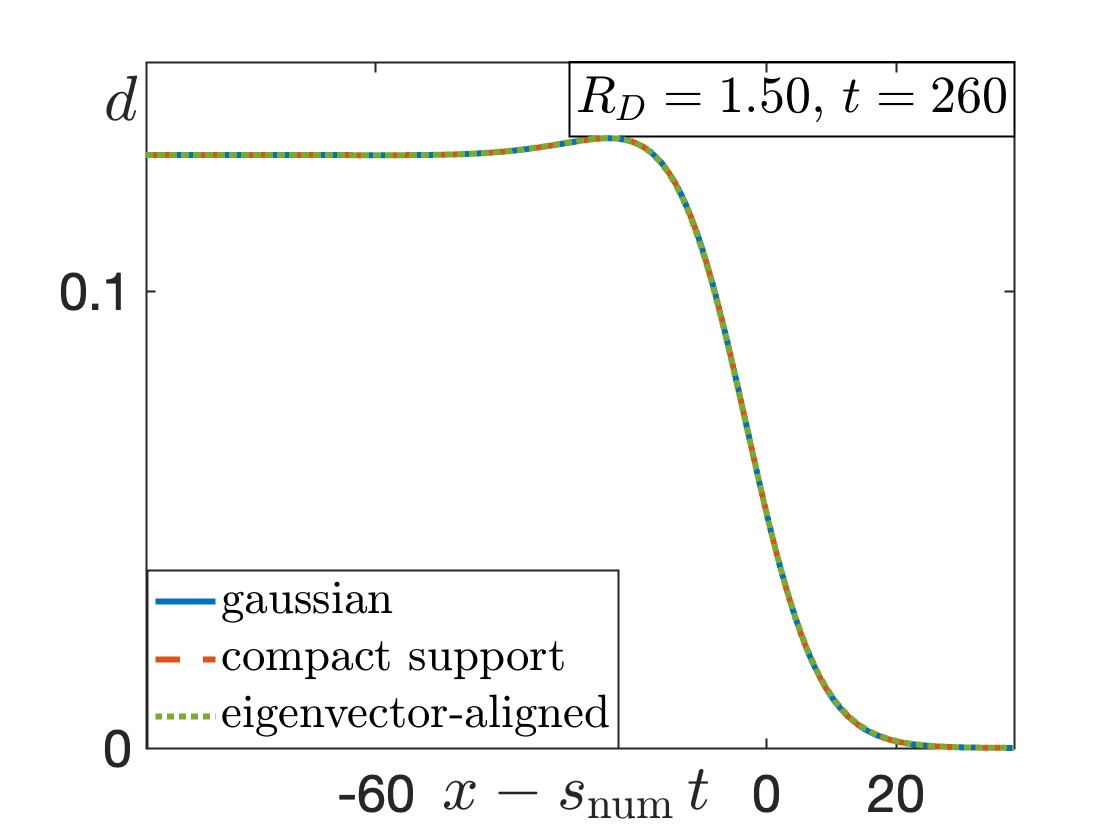}
		\caption{$\mathcal R_D=1.50$}
		\label{fig:D1_profile_RD150}
	\end{subfigure}
	\hfill
	\begin{subfigure}[t]{0.32\linewidth}
		\centering
		\includegraphics[width=\linewidth]{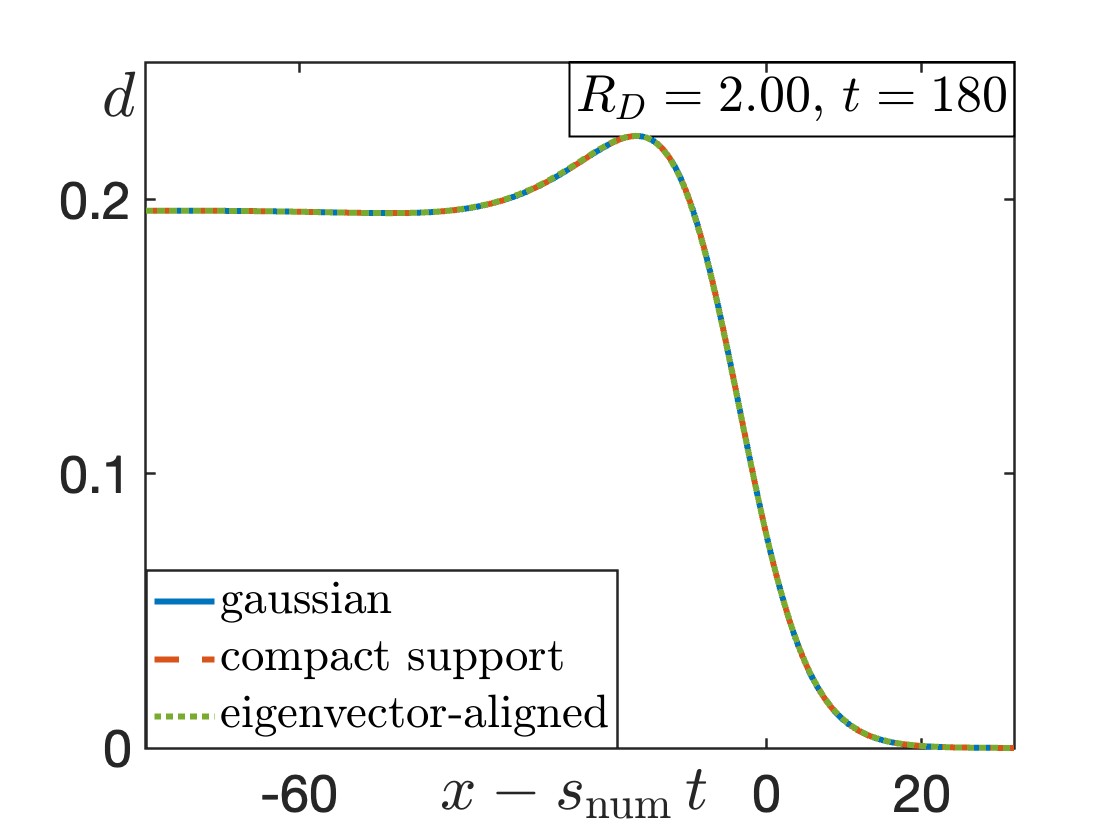}
		\caption{$\mathcal R_D=2.00$}
		\label{fig:D1_profile_RD200}
	\end{subfigure}
	\caption{$d$-profiles in the moving coordinate $x-s_{\mathrm{num}}t$ for three values of $\mathcal R_D$, obtained from three different initial perturbations of the damaged-tissue profile: a Gaussian perturbation, a compactly supported perturbation, and an exponentially decaying perturbation aligned with the dominant eigenvector of the linearized active subsystem. The remaining parameters are chosen as in Figure~\ref{fig:threshold} with domain size $\ell=12000$.}
	\label{fig:D1_profile_robustness}
\end{figure}
Figure~\ref{fig:D3_speed_transients} displays the time-dependent speed estimates obtained from the three different initial conditions.  The dashed line in each panel indicates the theoretical speed $s^*$. The estimates differ during the initial transient but become closer at later times, suggesting that the numerical front speed is almost independent of the chosen initialization.

\begin{figure}[!htbp]
	\centering
	\begin{subfigure}[t]{0.32\linewidth}
		\centering
		\includegraphics[width=\linewidth]{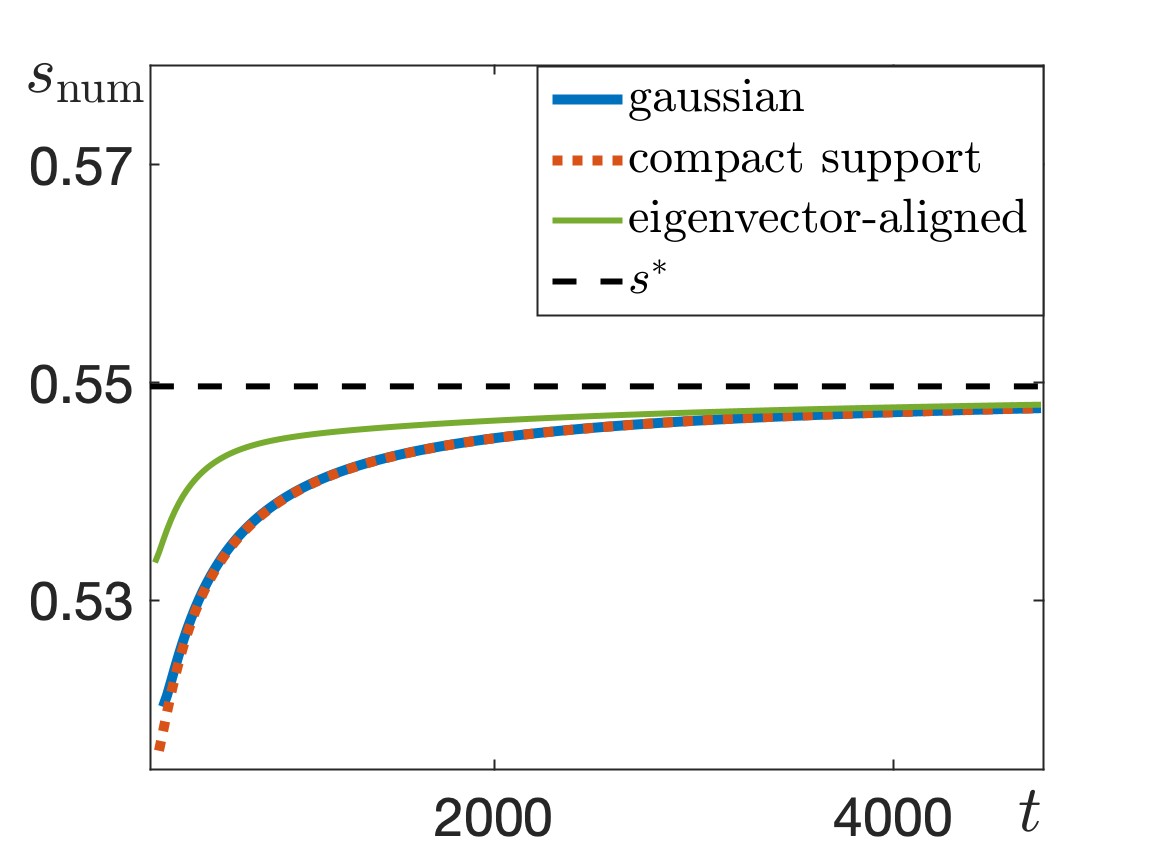}
		\caption{$\mathcal R_D=1.10$}
		\label{fig:D3_speed_RD110}
	\end{subfigure}
	\hfill
	\begin{subfigure}[t]{0.32\linewidth}
		\centering
		\includegraphics[width=\linewidth]{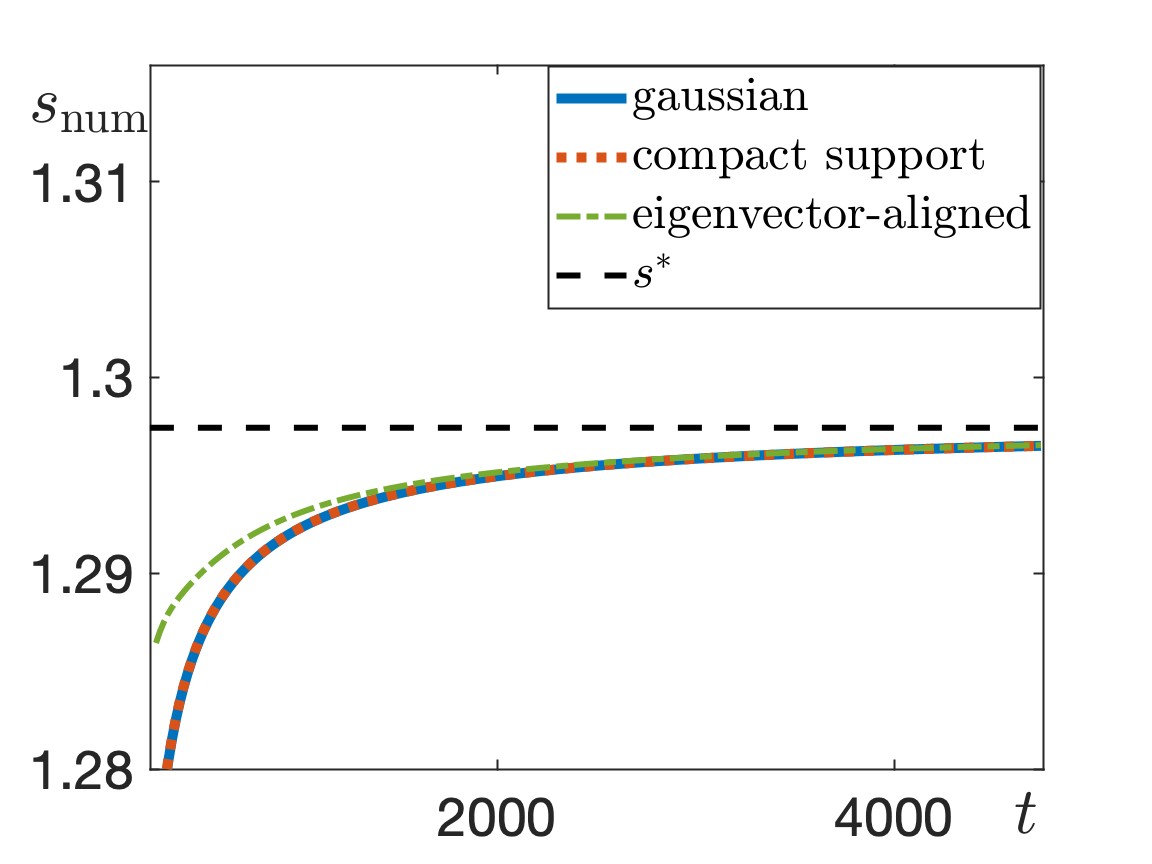}
		\caption{$\mathcal R_D=1.50$}
		\label{fig:D3_speed_RD150}
	\end{subfigure}
	\hfill
	\begin{subfigure}[t]{0.32\linewidth}
		\centering
		\includegraphics[width=\linewidth]{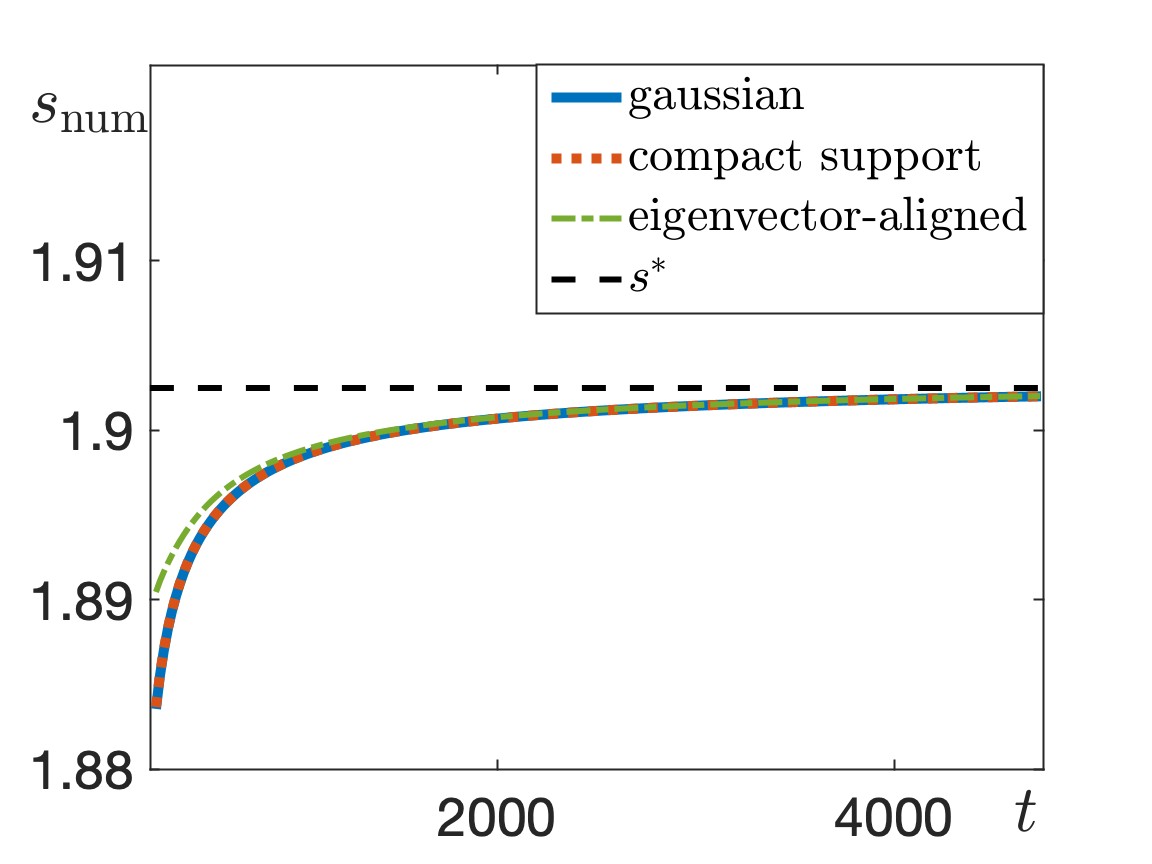}
		\caption{$\mathcal R_D=2.00$}
		\label{fig:D3_speed_RD200}
	\end{subfigure}
	\caption{Time-dependent numerical speed estimates for three values of
		$\mathcal R_D$; the dashed line indicates the corresponding $s^*$. The domain size is $\ell=12000.$ The remaining parameters are chosen as in Figure \ref{fig:threshold}.}
	\label{fig:D3_speed_transients}
\end{figure}
In Figure~\ref{fig:speed_validation}(a) we plot the numerical speeds obtained from the initial
$d$-profile having an exponentially decaying perturbation of the form $e^{-\gamma x}$ alongside the dispersion branch $s_+(\gamma)$ in 	\eqref{eq:spectral_speed_branch}. Denoting with $\gamma^*$ the minimizer of $s_+(\gamma)$, and $s^*$ its minimum value,  the numerical values approximately follow this branch for $\gamma<\gamma^*$ and remain near its minimum $s^*$ for $\gamma\geq\gamma^*$. Since these measurements are performed on finite domains, we also examine the effect of domain length. Figure~\ref{fig:speed_validation}(b) indicates that the agreement with the linear prediction $s^*$ improves as the domain becomes larger. All the above observations support the leading-edge calculation in the tested regimes, without proving nonlinear speed selection.

\begin{figure}[!htbp]
	\centering
	\begin{subfigure}[t]{0.48\linewidth}
		\centering
		\includegraphics[width=\linewidth]{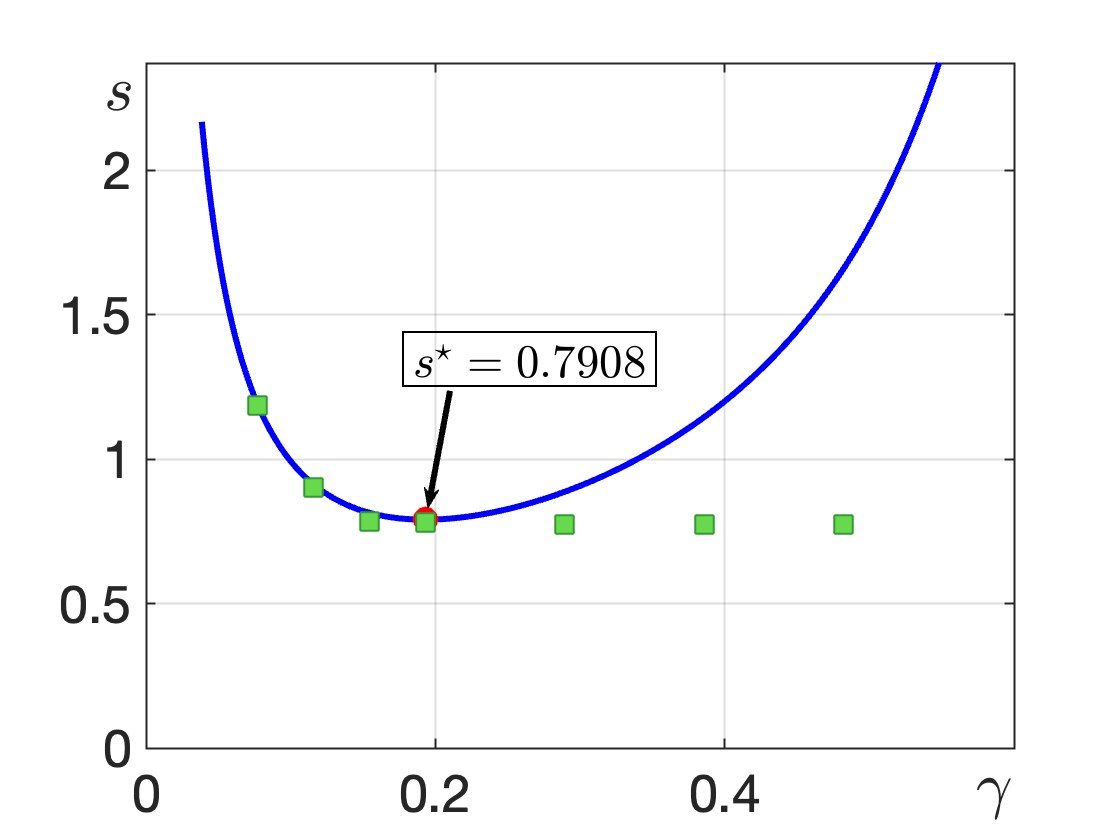}
		\caption{}
		\label{fig:speed_branch_comparison}
	\end{subfigure}
	\hfill
	\begin{subfigure}[t]{0.48\linewidth}
		\centering
		\includegraphics[width=\linewidth]{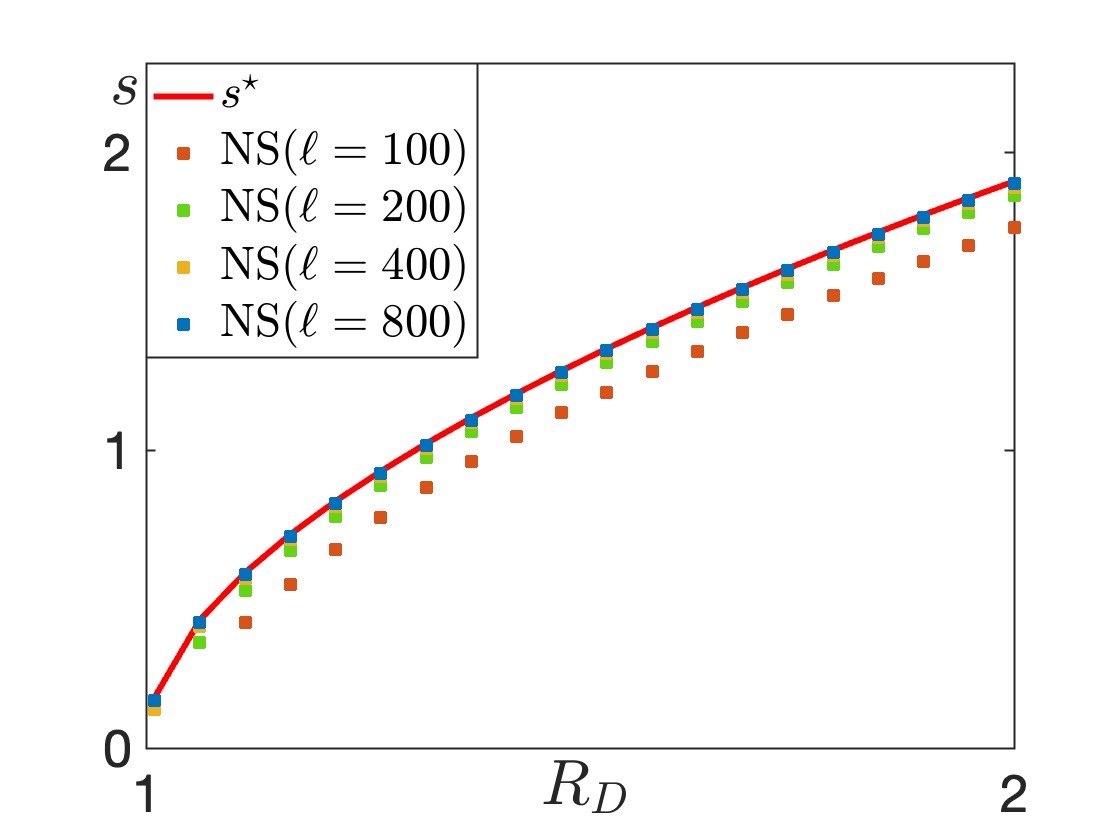}
		\caption{}
		\label{fig:speed_domain_convergence}
	\end{subfigure}
	\caption{(a) Comparison between the analytical dispersion relation $s_+(\gamma)$
		and numerical estimates of the front speed obtained from exponential initial
		conditions of the form $e^{-\gamma x}$. The candidate linear spreading
		speed $s^{\star}$ is defined in~\eqref{frontvelocity}; the values of $\gamma$
		are $0.4,\ 0.6,\ 0.8,\ 1,\ 1.5,\ 2,\ 2.5$ times $\gamma^*\simeq0.193$, the
		minimizer of $s_+(\gamma)$ for $\mathcal R_D=1.2$. The numerical speeds
		approximately follow $s_+(\gamma)$ for $\gamma<\gamma^*$ and remain close
		to $s^*$ for $\gamma\ge\gamma^*$.
		(b) Comparison between analytical and numerical front speeds for Gaussian
		initial conditions and different domain sizes $\ell$. The agreement improves
		as $\ell$ increases, consistently with reduced finite-size effects. The
		numerical comparison does not establish nonlinear speed selection. Parameters are chosen as in Figure \ref{fig:threshold}, in panel (b) $\mathcal R_D$ is varied through $\alpha.$}
	\label{fig:speed_validation}
\end{figure}

We now examine how varying the model parameters influences the evaluation of the numerical and linear speeds. In Figure~\ref{fig:speed_parameter_dependence} we show the outcomes for different values of the diffusivity $D_c$, the damage-clearance rate $\delta$, and the direct-repair rate $\rho$. As predicted by Proposition~\ref{prop:sstar}(ii), the theoretical speed $s^*$ increases with $D_c$ and decreases with $\delta$ and $\rho$. Faster signal diffusion promotes the spread of inflammatory signalling ahead of the damaged region, whereas clearance and repair weaken the damage-amplification mechanism. The measured numerical speeds exhibit the same qualitative trends over the parameter ranges shown in the figure, providing a further validation of the linear prediction.

\begin{figure}[!htbp]
	\centering
	\begin{subfigure}[t]{0.32\linewidth}
		\centering
		\includegraphics[width=\linewidth]{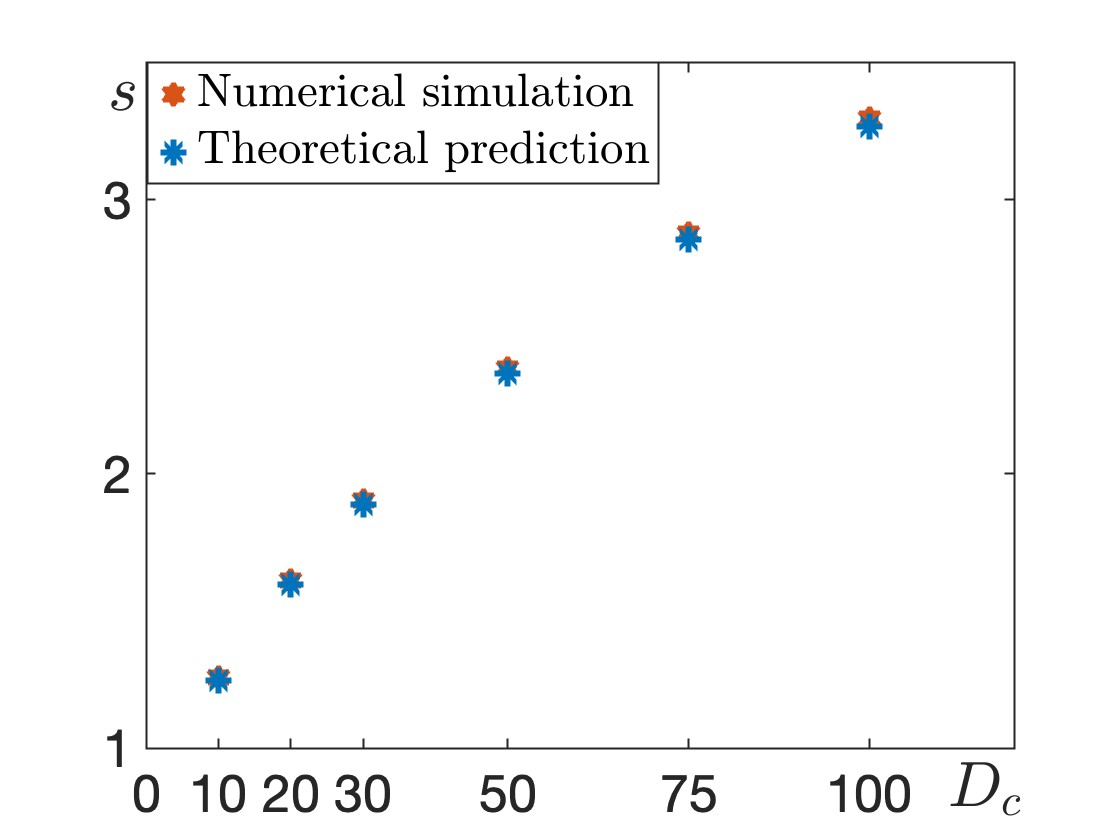}
		\caption{}
		\label{fig:speed_vs_Dc}
	\end{subfigure}
	\hfill
	\begin{subfigure}[t]{0.32\linewidth}
		\centering
		\includegraphics[width=\linewidth]{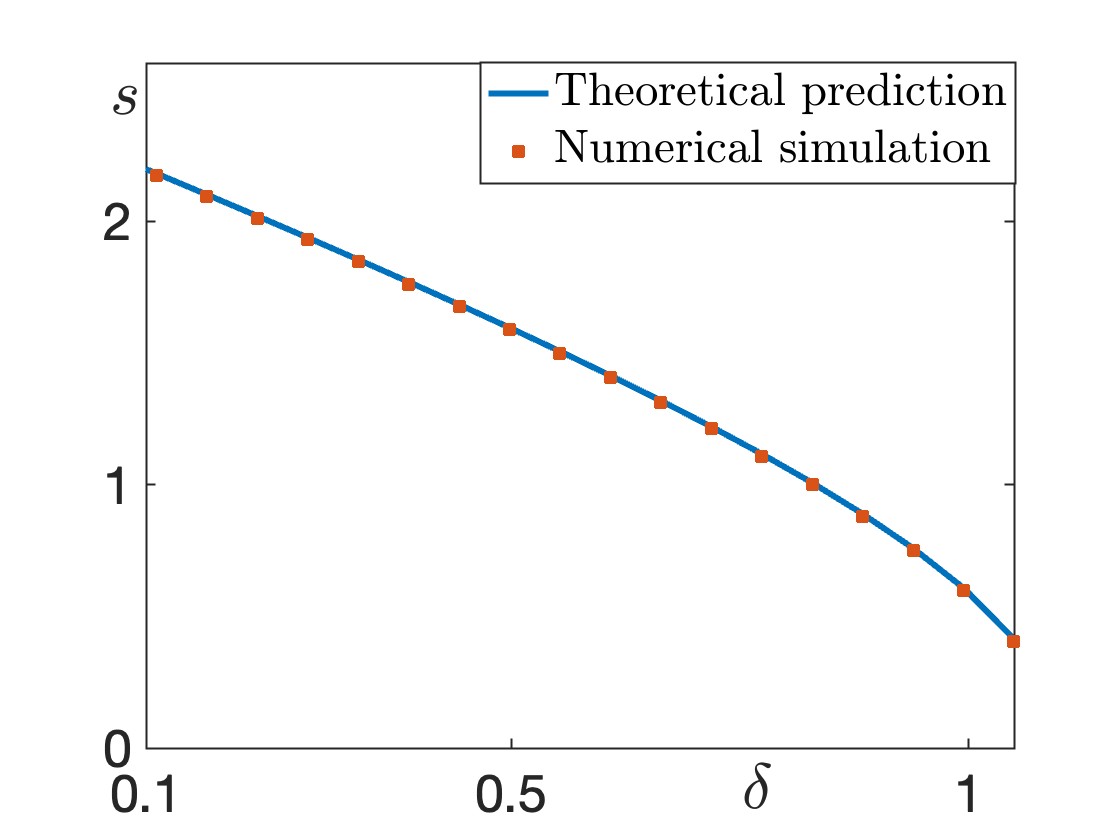}
		\caption{}
		\label{fig:speed_vs_delta}
	\end{subfigure}
	\hfill
	\begin{subfigure}[t]{0.32\linewidth}
		\centering
		\includegraphics[width=\linewidth]{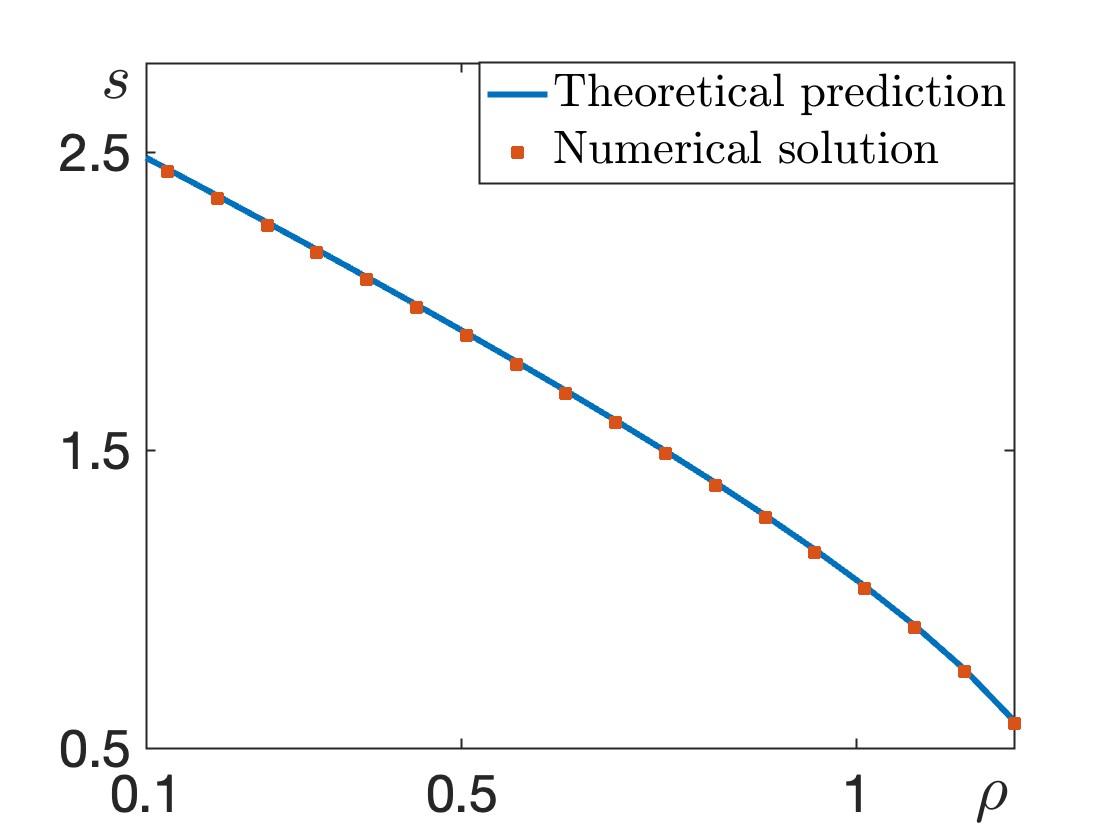}
		\caption{}
		\label{fig:speed_vs_rho}
	\end{subfigure}
	\caption{Comparison between the linear spreading-speed prediction and
		numerical front velocities with model parameters as in Figure~\ref{fig:threshold}.
		In all panels, the theoretical prediction is shown together with the
		corresponding numerical estimates.
		(a) Dependence of the invasion speed on the signal
		diffusivity $D_c$. (b) Dependence on the damage-clearance
		rate $\delta$. (c) Dependence on the direct-repair
		rate $\rho$. }
	\label{fig:speed_parameter_dependence}
\end{figure}

We now show how different parameter combinations can produce the same theoretical and numerical speeds, but different nonlinear states behind the front. In particular, we keep $\alpha\,r=16$ and compare $(\alpha,r)=(16,1)$, $(8,2)$ and $(4,4)$, with the other parameters at their reference values as in the previous simulations. All three combinations have $\mathcal R_D=2$ and, by~\eqref{eq:lambda1} and Proposition~\ref{prop:sstar}, the same linearized spectral branch and theoretical speed as shown in Figure~\ref{fig:alpha_r_invariance}(a). Their numerical speed estimates are already similar over the early-time interval for $t\leq175$, as shown in Figure~\ref{fig:alpha_r_invariance}(b).

However, the agreement in linear and numerical speeds does not imply identical nonlinear front profiles. Figure~\ref{fig:alpha_r_invariance}(c) shows that the damage level behind the front decreases
as $r/\mu$ increases, in agreement with the equilibrium value
\eqref{eq:Estar}, which at fixed $\mathcal R_D$ depends on how the amplification is
split between the
macrophage-mediated damage rate $\alpha$ and the signal production rate $r$:
$d^* \approx 0.19608$, $0.19231$ and $0.18519$ for the three pairs. The reason
is the saturation of macrophage recruitment, $m^* = c^*/(1+c^*)$. A larger
$r/\mu$ produces more signal per unit of damage, drives recruitment further
into saturation and lowers the effective gain of the loop, which is then
balanced at a smaller damage level. Near $E_0$ the signal is small and
saturation plays no role, so the threshold and the linear spreading speed
depend only on the product $\alpha r$. 
\begin{figure}[!htbp]
	\centering
	\begin{subfigure}[t]{0.32\linewidth}
		\centering
		\includegraphics[width=\linewidth]{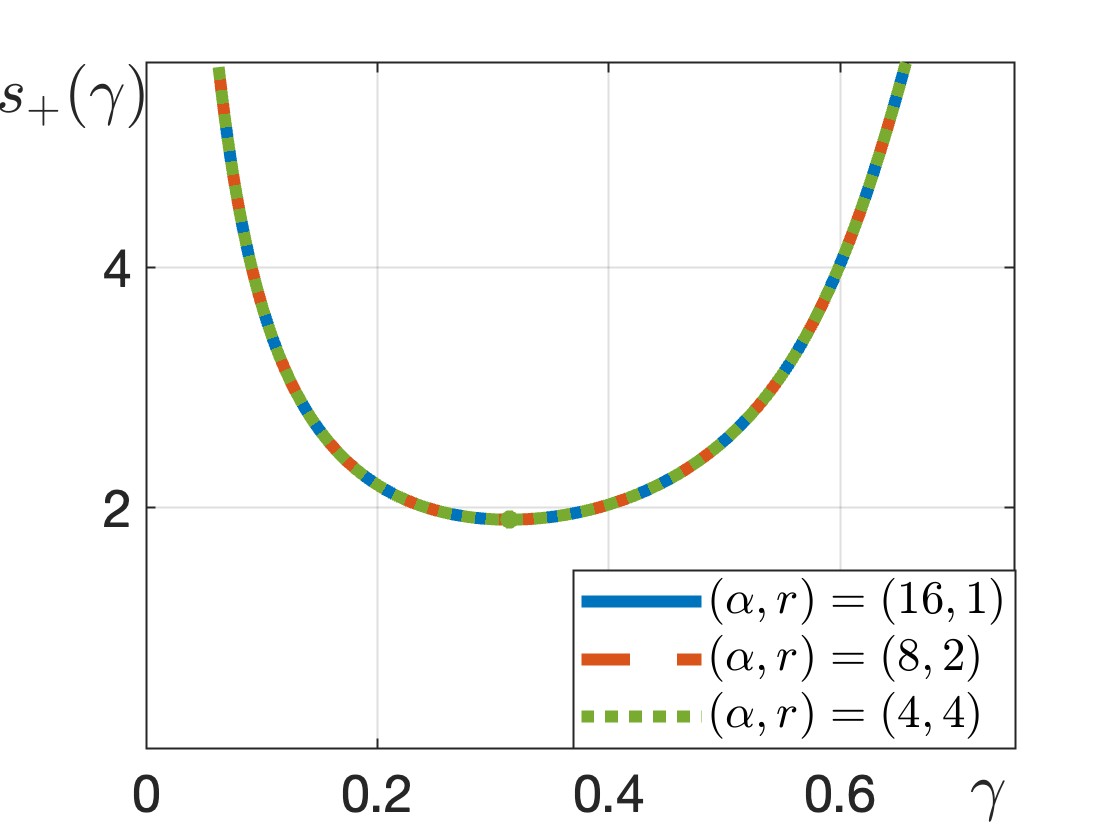}
		\caption{}
		\label{fig:alpha_r_branch}
	\end{subfigure}
	\hfill
	\begin{subfigure}[t]{0.32\linewidth}
		\centering
		\includegraphics[width=\linewidth]{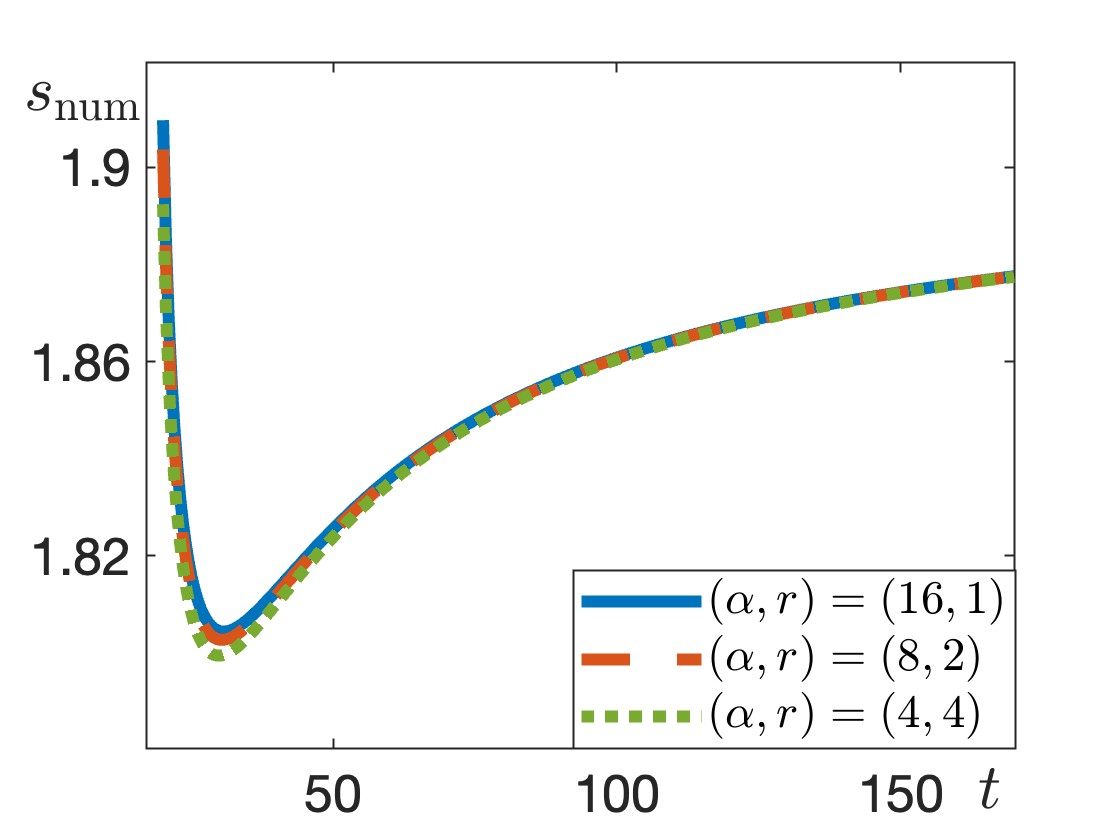}
		\caption{}
		\label{fig:alpha_r_speed}
	\end{subfigure}
	\hfill
	\begin{subfigure}[t]{0.32\linewidth}
		\centering
		\includegraphics[width=\linewidth]{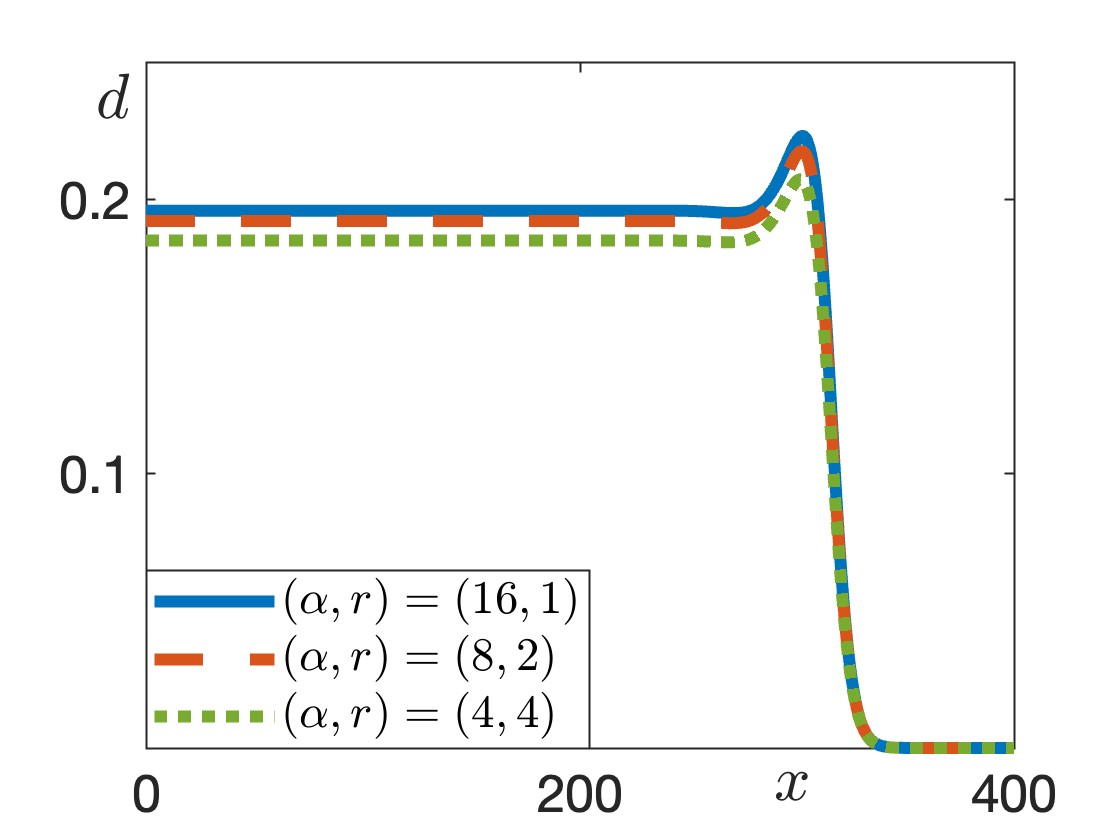}
		\caption{}
		\label{fig:alpha_r_profile}
	\end{subfigure}
	\caption{Effect of different decompositions of the product $\alpha r=16$,
		corresponding to $\mathcal R_D=2$ for the reference values of $\delta+\rho$
		and $\mu$. The parameter pairs are $(\alpha,r)=(16,1),(8,2),(4,4)$; the
		remaining parameters are fixed as in Fig.~\ref{fig:threshold}. (a) Common
		theoretical spectral branch. (b) Numerical front-speed estimates. (c) Damage
		profiles, showing the different levels behind the front. The corresponding equilibrium damage levels are given by
		\eqref{eq:Estar}.}
	\label{fig:alpha_r_invariance}
\end{figure}
\subsection{Effect of nonzero chemotaxis on the invading front}
\label{subsec:chemotaxis_numerics}

The previous numerical tests concern the chemotaxis-free case, for which
Theorem~\ref{thm:upper_speed} provides a nonlinear upper spreading bound.
We now test whether a nonzero chemotactic
coefficient changes the observed nonlinear front. We fix $\mathcal R_D=2$ and
repeat the simulations for $\chi_0=0,\ldots,100$, using the Gaussian
localized initial perturbation on $d$ . 

The measured front speeds decrease monotonically from $s_{\mathrm{num}}=1.88825$ at $\chi_0=0$ to $1.88817$ at $\chi_0=100$.
Their full variation over $\chi_0$ is approximately $8.2\times10^{-5}$, smaller than the
spatial sensitivity of the speed estimate under mesh refinement, which is
of order $3\times10^{-4}$. The simulations therefore reveal no detectable
chemotaxis-induced change in the invasion speed over the tested range.

The damage profiles are visually indistinguishable,
but their peak increases slightly across the displayed values of $\chi_0$.
For the Gaussian initialization at $t=200$, $\max_x d$ increases from
$0.22337$ at $\chi_0=0$ to $0.22481$ at $\chi_0=100$, and the excess over
$d^*$ rises from $13.92\%$ to $14.65\%$. The representative damage profiles are plotted in Figure~\ref{fig:chemotaxis_d_flux}(a).

Nevertheless, chemotaxis can still play a significant role in the spatial dynamics. To quantify its transport
contribution, we compare the infinity norms of the chemotactic and diffusive
macrophage fluxes,
\begin{equation}
	\mathscr Q(t):=
	\frac{\|J_{\mathrm{chem}}(t,\cdot)\|_{\infty}}
	{\|J_{\mathrm{diff}}(t,\cdot)\|_{\infty}+\varepsilon},
	\qquad \varepsilon=10^{-14},
	\label{eq:chemotactic_flux_ratio}
\end{equation}
where $J_{\mathrm{diff}}=-\nabla m$ and
$J_{\mathrm{chem}}=\chi_0\phi(c)m\nabla c$ in the dimensionless variables.
The ratio of the spatial supremum norms is evaluated for $t\ge50$ and it increases with $\chi_0$ and becomes
comparable to or larger than one for the largest tested values, as shown in Figure~\ref{fig:chemotaxis_d_flux}(b). According to
this global indicator, the chemotactic flux can therefore be substantial even
when the measured front position and plotted damage profile change little.

This behaviour is consistent with the linear analysis around $E_0$, where the chemotactic term is of higher order and therefore not entering in the linear spreading-speed prediction. Although this does not establish that the nonlinear front speed is independent of chemotaxis, the numerical results suggest that the effect of the chemotactic term may be more pronounced on the internal front dynamics than on its propagation speed. A similar distinction has been investigated by \cite{SalakoShenXue2019} in a parabolic-elliptic Keller-Segel system with logistic growth, where the spreading speed was shown to be independent of chemotaxis under suitable parameter conditions.

\begin{figure}[!htbp]
	\centering
	\begin{subfigure}[t]{0.48\linewidth}
		\centering
		\includegraphics[width=\linewidth]{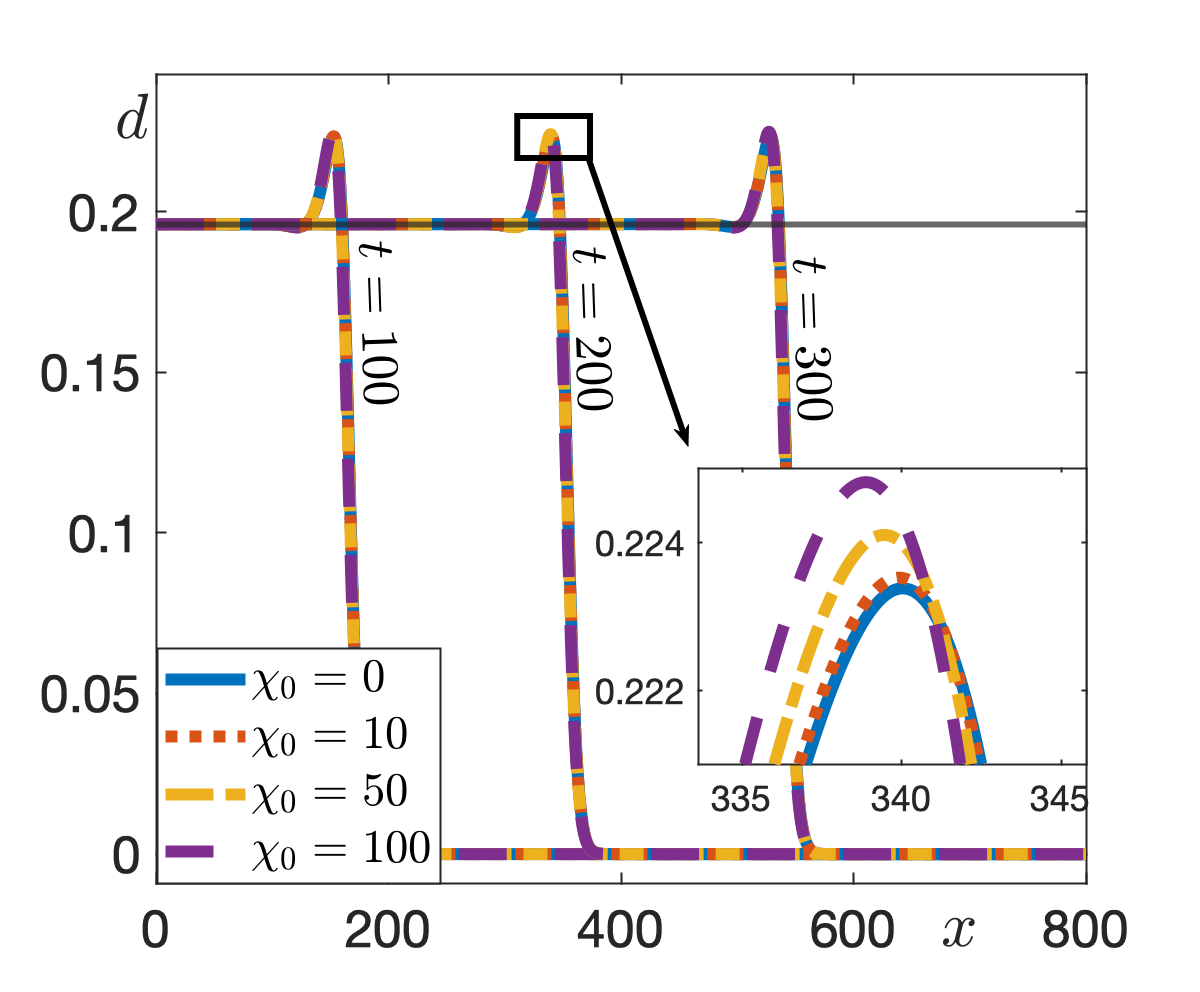}
		\caption{}
	\end{subfigure}
	\hfill
	\begin{subfigure}[t]{0.48\linewidth}
		\centering
		\includegraphics[width=\linewidth]{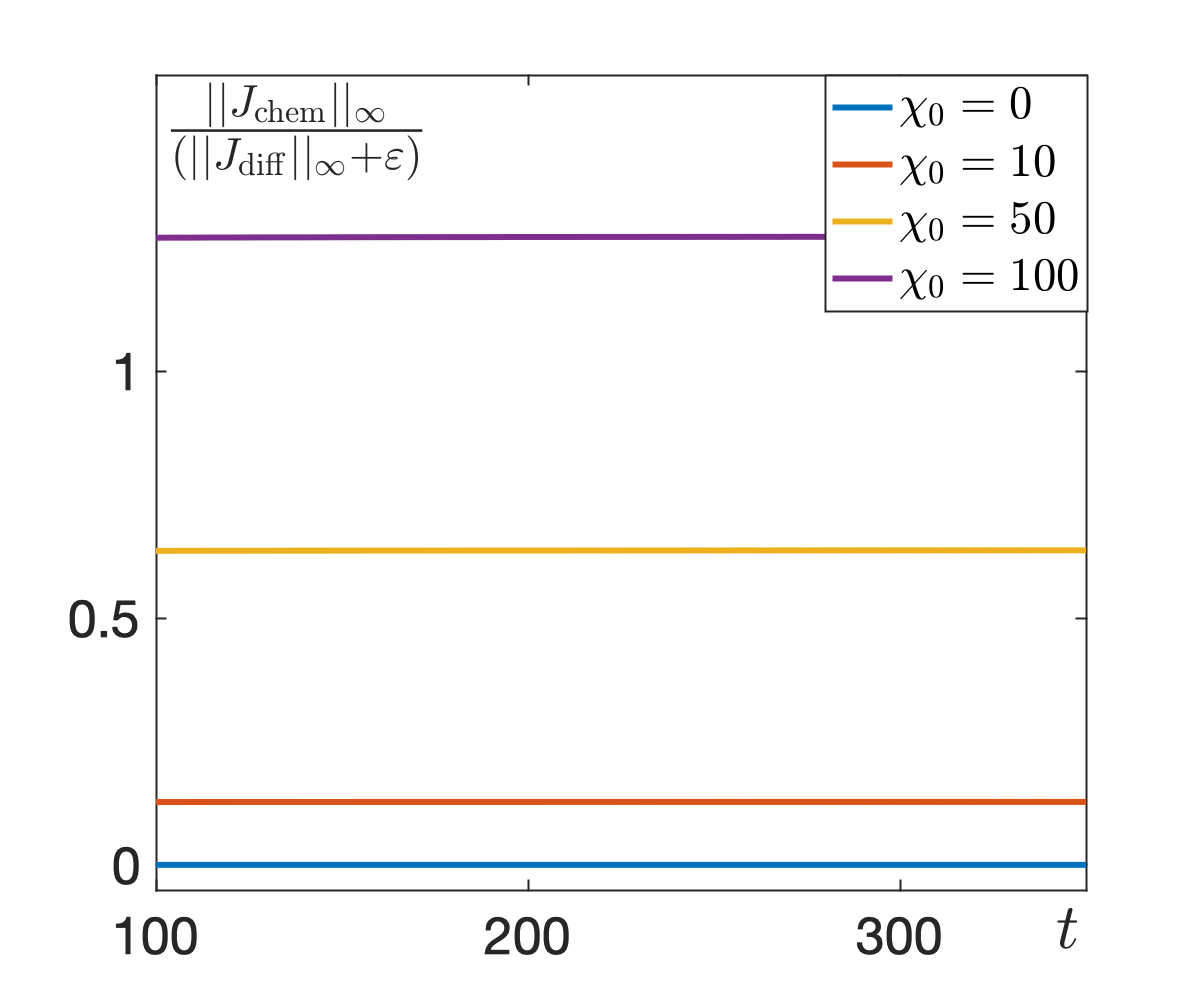}
		\caption{}
	\end{subfigure}
	\caption{Effect of nonzero chemotaxis at $\mathcal R_D=2$. (a) Damage
		profiles for $\chi_0=0,10,50,100$ at $t=100$, $200$, and $300$; in the inset the magnification of the curves at $t=200$ where maximal values are attained. (b) Ratio of the
		spatial supremum norms of the chemotactic and diffusive macrophage fluxes,
		defined in~\eqref{eq:chemotactic_flux_ratio} and evaluated for
		$t\ge50$. The ratio becomes of order one
		or larger for the largest tested values of $\chi_0$, while the measured
		front speeds vary by less than
		$10^{-4}$. Model parameters as in Figure$\,$\ref{fig:threshold}.}
	\label{fig:chemotaxis_d_flux}
\end{figure}
\section{Conclusions}
\label{sec:conclusions}

We have studied a spatial model of damage--inflammation feedback motivated
by Duchenne muscular dystrophy. The local tissue dynamics confines the
tissue fractions to an invariant region. Combined with parabolic estimates,
this property controls the chemotactic drift and allows us to establish
global well-posedness and dissipativity on bounded connected domains, under
the stated regularity assumptions and for arbitrary chemotactic sensitivity.

The threshold identified for the homogeneous feedback retains a decisive role in the spatial system. On bounded domains, active damage, inflammatory signal and effective macrophage activity vanish at or below the threshold.
Above it, the damage--inflammation feedback exhibits uniform componentwise persistence for initial data with at least one nonzero active component.
Decay is exponential below threshold and has an $O((1+t)^{-1})$ bound at criticality.

The threshold, however, does not by itself determine the level of persistent damage. When $\alpha r$ and all other parameters are fixed, the threshold and the candidate linear spreading speed remain unchanged, while the damage level at the positive homogeneous equilibrium depends on the balance between damage production and signal generation. The numerical profiles illustrate this distinction.

The propagation analysis provides a starting point for understanding how localized damage spreads. The leading-edge problem yields a candidate linear speed, and, without chemotaxis, we prove a nonlinear upper spreading bound on the whole line for compactly supported active initial data. The numerical
speeds are close to the linear prediction in the tested cases. With nonzero chemotaxis, they change little across the tested sensitivities, even when the spatial maximum of the chemotactic flux magnitude is comparable to or larger than that of the diffusive flux. These observations do not establish
nonlinear speed selection. A next step is therefore to determine whether a matching lower spreading bound holds and whether travelling waves select the
predicted speed.

The possible role of chemotaxis in spatial heterogeneity also remains open.
At the damage-free equilibrium $E_0$, chemotaxis does not enter the
linearization, and diffusion cannot destabilize this state when it is stable
to homogeneous perturbations. The situation may be different at the positive
homogeneous equilibrium $E_+$: there, chemotaxis contributes to the
linearization when $\chi_0>0$. Determining whether nonconstant modes can
become unstable around $E_+$ would clarify whether the model can support a
spatial instability of the persistent damage--inflammation state.

Finally, the model is intended as a mechanistic framework rather than a
quantitative description of disease progression. Spatial and longitudinal
data could help assess which combinations of its effective parameters are
identifiable, how sensitive its predictions are to parameter uncertainty,
and whether the predicted regimes are consistent with observed patterns of
damage and inflammation.
\appendix
\section{Well-posedness and auxiliary estimates on bounded domains}\label{app:wp}

We collect the auxiliary estimates used in Section~\ref{wp:main_section}. The parabolic estimates are classical; we retain the short-interval uniformity, the $L^p$ potential and the initial trace required by the coupled system.

\subsection{Parabolic estimates and lower-order terms}\label{app:estimates}

The estimates also hold with $p$ replaced by any fixed $\nu>n+2$ and
$E_p$ by the corresponding Neumann trace space $E_\nu$, with constants
depending on $\nu$. Section~\ref{pdep:sec:pde_persistence} uses $\nu=2p$.

For $I=(s,s+\ell)$, $0<\ell\le2$, the solution of
$v_t-D\Delta v+\lambda v=F$ with $v(s)=v_s\in E_p$ and
$\partial_{\mathbf n}v=0$ satisfies, for $D>0$ and $\lambda\ge0$,
\begin{equation}\label{wp:maxreg}
	\|v\|_{W^{1,2}_p(I\times\Omega)}+
	\sup_{t\in\overline I}\|v(t)\|_{E_p}+\|v\|_{C_x^{0,1}}
	\le C(\|F\|_{L^p(I\times\Omega)}+\|v_s\|_{E_p}).
\end{equation}
We use the classical $L^p$ estimates and parabolic embeddings in
\citet[Chs.~II and IV]{LSU88}.  Maximal $L^p$-regularity and the trace
embedding
\[
W^{1,p}(I;L^p(\Omega))\cap L^p(I;W^{2,p}(\Omega))
\hookrightarrow C(\overline I;W^{2-2/p,p}(\Omega))
\]
are standard; see \citet[Ch.~III, Sect.~1.5 and Sect.~4.10,
in particular Theorem~III.4.10.2]{Amann95}.
The constant in \eqref{wp:maxreg} is independent of $s,\ell$: extend $F$
by zero to an interval of length two and restrict its solution.

If $v\in W^{1,2}_p(I\times\Omega)$ has zero initial trace, extend $v$ by zero to $(s-1,s)$. The parabolic H\"older embedding gives, for $0<\beta<1-(n+2)/p$,
\begin{equation}\label{wp:zero_trace}
	\|v\|_{C_x^{0,1}}\le C\ell^\vartheta\|v\|_{W^{1,2}_p(I\times\Omega)},
	\qquad\vartheta=\beta/2.
\end{equation}
The gradient also has zero initial trace because $E_p\hookrightarrow C^1(\overline\Omega)$. We shall also use
\begin{equation}\label{wp:interpolation}
	\|\nabla v\|_{L^p(I\times\Omega)}
	\le\varepsilon\|v\|_{L^p(I;W^{2,p})}+C_\varepsilon\|v\|_{L^p(I\times\Omega)}.
\end{equation}
Let $S_D(t)=e^{tD\Delta}$ be the Neumann heat semigroup.  
%
%
Under our $C^{2+\eta_0}$ boundary assumption, we use the Gaussian
bound for the Neumann heat kernel from \citet[Chs.~4 and 6]{Ouhabaz2005}.
Following \citet[Lemma~1.3]{Winkler10} and restoring the
constant mode, we obtain, for $t>0$ and
$1\le a\le q\le\infty$,
\begin{align}
	\|S_D(t)v\|_q&\le C(1+t^{-\frac n2(1/a-1/q)})\|v\|_a,\label{wp:heat_smooth}\\
	\|\nabla S_D(t)v\|_q&\le C(1+t^{-\frac12-\frac n2(1/a-1/q)})\|v\|_a.\label{wp:heat_gradient}
\end{align}
Self-adjointness and the gradient estimate applied to the conjugate exponents $q'\le a'$ give, for $a<\infty$,
\begin{equation}\label{wp:heat_flux_general}
	\|S_D(t)\nabla\cdot F\|_q
	\le C(1+t^{-\frac12-\frac n2(1/a-1/q)})\|F\|_a.
\end{equation}
Here $\langle\nabla\cdot F,\psi\rangle=-\int_\Omega F\cdot\nabla\psi$; this includes $a=1$, using the $L^{q'}$--$L^\infty$ gradient estimate. For smooth $F$, this is the usual divergence when $F\cdot\mathbf n=0$. Multiplication by $e^{-\lambda t}$ gives the damped estimates.

\begin{lemma}\label{wp:linear_lemma}
	Let $I=(s,s+L)$ with $0<L\le2$, $D>0$, $\lambda\ge0$,
	$b\in L^\infty(I\times\Omega)$,
	$k=\nabla\cdot b\in L^p(I\times\Omega)$, and assume that the normal trace of
	$b$ vanishes on $I\times\partial\Omega$. For $v_s\in E_p$ and
	$g\in L^p(I\times\Omega)$, the problem
	\[
	v_t-D\Delta v=-b\cdot\nabla v-(\lambda+k)v+g,
	\qquad v(s)=v_s,\qquad \partial_{\mathbf n}v=0,
	\]
	has a unique solution satisfying
	\[
	\|v\|_{W^{1,2}_p}+\|v\|_{C_x^{0,1}}+\sup_{s\le t\le s+L}\|v(t)\|_{E_p}
	\le C(\|v_s\|_{E_p}+\|g\|_{L^p(I\times\Omega)}).
	\]
	The constant depends on the coefficient norms and $D,\lambda,p,n,\Omega$, but is uniform for $L\le2$ with these norms bounded. If $g\ge0$ and $v_s\ge0$, then $v\ge0$.
\end{lemma}
\begin{proof}
Put $B=\norm b_\infty$ and $K=\norm k_p$. On an interval
$I=(s,s+\ell)$, let $\mathcal Tz$ solve
\[
 (\mathcal Tz)_t-D\Delta(\mathcal Tz)
 =-b\cdot\nabla z-(\lambda+k)z+g,
 \qquad (\mathcal Tz)(s)=v_s.
\]
This map acts on the affine $W^{1,2}_p$ space with initial trace $v_s$
and Neumann boundary condition. For $V=z_1-z_2$ and
$W=\mathcal Tz_1-\mathcal Tz_2$, both initial traces vanish.
Maximal regularity and \eqref{wp:zero_trace} give
\begin{equation}\label{wp:linear_contraction}
 \norm W_{W^{1,2}_p(I\times\Omega)}
 \le C\ell^\vartheta[(B+\lambda)\ell^{1/p}+K]
       \norm V_{W^{1,2}_p(I\times\Omega)}.
\end{equation}
In particular, $\norm{kV}_p\le K\norm V_\infty
\le CK\ell^\vartheta\norm V_{W^{1,2}_p}$, so the $L^p$ potential is
controlled without a pointwise bound. Choose $\ell_0>0$ so that the
factor in \eqref{wp:linear_contraction} is at most $1/2$ for
$\ell\le\ell_0$. Banach's theorem gives existence and uniqueness on
each such interval. To retain the nonzero trace in the a priori estimate,
write $v=e^{(t-s)D\Delta}v_s+u$, with $u(s)=0$, and absorb the same
lower-order terms in $u$. Estimate \eqref{wp:maxreg} then bounds the
parabolic, trace and $C_x^{0,1}$ norms of $v$ by
$C_{B,K}(\norm{v_s}_{E_p}+\norm g_p)$.
A partition into at most $1+2/\ell_0$ intervals proves the asserted
uniform estimate for $L\le2$.

For positivity, use $k=\nabla\cdot b$ to write
\begin{equation}\label{wp:linear_divergence_form}
 v_t-D\Delta v+\nabla\cdot(bv)+\lambda v=g.
\end{equation}
The negative-part test $-v_-$ is justified by truncation; both boundary
terms vanish. For $g\ge0$, Young's inequality gives
\[
 \frac12\frac d{dt}\norm{v_-}_2^2+\frac D2\norm{\nabla v_-}_2^2
 +\lambda\norm{v_-}_2^2
 \le\frac{B^2}{2D}\norm{v_-}_2^2.
\]
Since $v_-(s)=0$, Gronwall's inequality proves $v\ge0$.
\end{proof}

\begin{lemma}\label{wp:cutoff_lemma}
	Let $v$ be a bounded strong Neumann solution of $v_t-D\Delta v=-b\cdot\nabla v-(\lambda+k)v+g$. On $I_t=(t-1,t)$, $t\ge1$, put $B=\|b\|_{L^\infty(I_t\times\Omega)}$ and $K=\|k\|_{L^p(I_t\times\Omega)}$. Then
	\begin{equation}\label{wp:cutoff_bound}
		\begin{split}
			&\|v\|_{W^{1,2}_p(\Omega_{t-1/2,t})}
			+\sup_{t-1/2\le s\le t}\|v(s)\|_{E_p}+\|v\|_{C_x^{0,1}(\overline{\Omega_{t-1/2,t}})}\\
			&\qquad\le C_B\bigl[(1+K)\|v\|_{L^\infty(I_t\times\Omega)}
			+\|g\|_{L^p(I_t\times\Omega)}\bigr].
		\end{split}
	\end{equation}
	On an initial interval $(0,L)$, $0<L\le1$, with $B=\|b\|_{L^\infty(\Omega_L)}$, the corresponding bound is
	\begin{equation}\label{wp:initial_bound}
		\begin{split}
			&\|v\|_{W^{1,2}_p(\Omega_L)}+\sup_{0\le s\le L}\|v(s)\|_{E_p}
			+\|v\|_{C_x^{0,1}(\overline{\Omega_L})}\\
			&\qquad\le C_B\bigl[\|v(0)\|_{E_p}+(1+\|k\|_{L^p(\Omega_L)})
			\|v\|_{L^\infty(\Omega_L)}+\|g\|_{L^p(\Omega_L)}\bigr].
		\end{split}
	\end{equation}
	The constants also depend on $D,\lambda,p,n,\Omega$, and are independent of the position and length of the stated intervals.
\end{lemma}
\begin{proof}
	Take a fixed smooth cut-off $\zeta$ with $\zeta(t-1)=0$, $\zeta=1$ on $[t-1/2,t]$ and uniformly bounded derivative. The function $w=\zeta v$ satisfies
	\[
	w_t-D\Delta w+\lambda w
	=-b\cdot\nabla w-kw+\zeta g+\zeta'v,
	\qquad w(t-1)=0.
	\]
	Apply \eqref{wp:maxreg}, estimate $\|kw\|_p\le K\|v\|_\infty$, and use \eqref{wp:interpolation} to absorb $b\cdot\nabla w$. This proves \eqref{wp:cutoff_bound}. The same calculation applied to $v$ on $(0,L)$, retaining the initial trace in \eqref{wp:maxreg}, proves \eqref{wp:initial_bound}.
\end{proof}

\subsection{Global bounds in the presence of chemotaxis}\label{app:global_bootstrap}

We record the estimates used in the proof of
Theorem~\ref{wp:theorem}. Since $0\le d\le1$, the source $rd$ is uniformly
bounded, and \eqref{wp:main_signal_bound} controls $c$. To control its
derivatives, apply \eqref{wp:maxreg} on the initial interval, using
$c_{in}\in E_p$, and \eqref{wp:cutoff_bound} on later intervals. In this
application of the cutoff estimate, the coefficients denoted there by $b$
and $k$ are zero, while $\lambda=\mu$ and $g=rd$. These estimates and the
parabolic embedding for $p>n+2$ give
\[
 \|\nabla c\|_{L^\infty(I\times\Omega)}
 +\|\Delta c\|_{L^p(I\times\Omega)}\le C_R
\]
on every interval $I$ of length at most one within the local existence
interval, for initial data satisfying $\|U_{in}\|_{\mathcal X}\le R$.
The constant $C_R$ is independent of the position of $I$. The
boundedness of $\phi$ and $\phi'$ on $[0,\infty)$ therefore yields
\[
 b=\chi_0\phi(c)\nabla c\in L^\infty,
 \qquad
 k=\nabla\cdot b
 =\chi_0\bigl(\phi'(c)|\nabla c|^2+\phi(c)\Delta c\bigr)\in L^p
\]
on such intervals.

Integrating the macrophage equation and using the no-flux condition gives
\eqref{wp:mass_eq}, and therefore
\[
 \|m(t)\|_{L^1(\Omega)}
 \le e^{-t}\|m_{in}\|_{L^1(\Omega)}
 +|\Omega|(1-e^{-t})
 \le M_1,
 \qquad M_1:=\max\{\|m_{in}\|_{L^1(\Omega)},|\Omega|\}.
\]
For $s<t$, the variation-of-constants formula is
\begin{align*}
 m(t)={}&e^{(t-s)(\Delta-1)}m(s)
 +\int_s^t e^{(t-\tau)(\Delta-1)}f(c(\tau))\,d\tau\\
 &-\int_s^t e^{(t-\tau)(\Delta-1)}\nabla\cdot(bm)(\tau)\,d\tau.
\end{align*}
Choose $s=0$ for $t\le1$ and $s=t-1$ otherwise. Fix $q>\max\{n,1\}$ and
set $\vartheta=\frac12+\frac n{2q}<1$. The heat-semigroup estimates give
\[
 \|bm\|_{L^q}
 \le \|b\|_{L^\infty}M_1^{1/q}\|m\|_{L^\infty}^{1-1/q}
\]
and hence, with $Y_T=\sup_{0\le t\le T}\|m(t)\|_\infty$,
\[
 Y_T\le C+C Y_T^{1-1/q}
 \int_0^1(1+s^{-\vartheta})\,ds.
\]
Because $1-1/q<1$, Young's inequality gives
$CY_T^{1-1/q}\le\varepsilon Y_T+C_\varepsilon$. Choosing
$\varepsilon<1/2$ yields a bound for $Y_T$ independent of $T$ on bounded
sets of initial data.

The bound on $m$ controls the potential term $km$. The signal estimates
above and the cutoff estimate for $m$ give uniform parabolic, trace and
$C_x^{0,1}$ bounds on unit time slabs. On the initial interval, the estimate
with initial trace is used; on later intervals, the cutoff estimate is applied
to the last half-unit interval. These bounds control the phase-space norm and
therefore exclude finite-time blow-up by the continuation criterion.

Finally,
\[
\nabla\cdot(\phi(c)m\nabla c)
=\phi(c)\nabla m\cdot\nabla c
 +m\phi'(c)|\nabla c|^2
 +m\phi(c)\Delta c
\]
belongs to $L^p(\Omega_T)$ because $m,c\in W^{1,2}_p(\Omega_T)$ and
$p>n+2$. This completes the global bootstrap used in
Theorem~\ref{wp:theorem}.

\subsection{Estimates for the absorbing set}\label{app:absorbing_bootstrap}

We give the details behind Proposition~\ref{wp:absorbing}. By
\eqref{wp:main_signal_bound}, $c$ is eventually bounded independently of
the initial datum. Since $0\le d\le1$, the cutoff estimate for the signal
equation, as applied above, gives an eventual data-independent bound for
$\nabla c$ and hence for $b$. Together with \eqref{wp:mass_eq}, this yields
after a transient time
\[
 \|b(t)\|_\infty\le B_*,\qquad \|m(t)\|_{L^1(\Omega)}\le M_*,
\]
where $B_*$ and $M_*$ are independent of the initial datum. The entrance
time is uniform on bounded admissible sets. Apply~\eqref{wp:m_duhamel} on
$[t-1,t]$ and use the heat-flux estimate as in the proof of
Theorem~\ref{wp:theorem}. With $\omega_q=1-1/q\in(0,1)$, this gives
\[
 \|m(t)\|_\infty\le A_*+C_*
 \left(\sup_{t-1\le s\le t}\|m(s)\|_\infty\right)^{\omega_q},
\]
with $A_*$ and $C_*$ independent of the initial datum. If $y_j$ denotes
the supremum of $\|m(t)\|_\infty$ on the $j$-th unit interval, then
\[
 y_j\le A_*+C_*\max\{y_{j-1},y_j\}^{\omega_q}.
\]
Set $F(s)=A_*+C_*s^{\omega_q}$ and choose
\[
 K>\max\{2A_*,(2C_*)^{1/(1-\omega_q)}\}.
\]
Then $F(s)<s$ and $F'(s)\le1/2$ for $s\ge K$. Consequently,
\[
 y_j\le K+\tfrac12(y_{j-1}-K)_+,
 \qquad
 y_j\le K+2^{-j}(y_0-K)_+.
\]
The initial value $y_0$ is bounded uniformly on bounded sets of initial
data by Theorem~\ref{wp:theorem}; hence the entrance time below $K+1$ is
uniform on such sets. Applying~\eqref{wp:cutoff_bound} first to $c$ and
then to $m$ gives the eventual trace and gradient bounds in
\eqref{wp:eventual_strong}. This proves the absorbing property of
$B_0$ and completes the proof of Proposition~\ref{wp:absorbing}.

\subsection{Iterated bounds for the homogeneous system}\label{app:ode_iteration}

For completeness, we record the iteration used in the homogeneous case in
Remark~\ref{wp:remark_iterated_bounds}. Proposition~\ref{cor:global} gives
\[
 M_0:=\frac{r}{r+\mu},
 \qquad \limsup_{t\to\infty}m(t)\le M_0.
\]
Define, for $n\ge0$,
\begin{equation}\label{wp:ode_bound_iteration}
 D_{n+1}:=\frac{\alpha M_n}{\alpha M_n+\eta},
 \qquad
 M_{n+1}:=\frac{rD_{n+1}}{\mu+rD_{n+1}}.
\end{equation}
Repeated scalar comparison yields
\[
 \limsup_{t\to\infty}d(t)\le D_{n+1},\qquad
 \limsup_{t\to\infty}c(t)\le\frac r\mu D_{n+1},\qquad
 \limsup_{t\to\infty}m(t)\le M_{n+1}
\]
for every $n\ge0$. The iteration is generated by the increasing map
\[
 G(M):=\frac{r\alpha M}{\mu\eta+\alpha(\mu+r)M}.
\]
Since $D_1<1$, one has $M_1<M_0$, and therefore $(M_n)$ is decreasing. If
$\alpha r>\eta\mu$, the positive fixed point of $G$ lies strictly below
$M_0$ and monotonicity gives a lower bound for the sequence. If
$\alpha r\le\eta\mu$, zero is the only non-negative fixed point. Hence
\[
 M_\infty=
 \frac{(\alpha r-\eta\mu)_+}{\alpha(\mu+r)}
 =\left(M_0-\frac{\eta\mu}{\alpha(\mu+r)}\right)_+.
\]
Thus the limiting explicit upper bound vanishes for $0 \leq \mathcal R_D \leq 1$ and is positive for $\mathcal R_D > 1$ and the iteration provides a sequence of improving
asymptotic bounds.

\subsection{Local construction and contraction estimates}\label{app:local_construction}

We give the details behind Theorem~\ref{wp:local}. For $T\le1$, consider
the closed set
\[
 X_{T,R}=\left\{\widetilde m\in C([0,T];C(\overline\Omega)):
 0\le\widetilde m\le R,\quad \widetilde m(0)=m_{in}\right\}.
\]
Given $\widetilde m\in X_{T,R}$, solve successively the tissue ODEs, the
linear Neumann problem for $c$, and the linear drift--diffusion problem for
$m$. The invariant-region lemma gives $0\le d\le1$ and comparison gives
$c\ge0$. 
Estimate \eqref{wp:maxreg} and Lemma~\ref{wp:linear_lemma} then give a bound on the output $m$ that is independent of the input $\widetilde m$.
Choosing $R$ larger than this bound makes the map a
self-map of $X_{T,R}$.

For two inputs, write
\[
 \widetilde M=\widetilde m_1-\widetilde m_2,\quad
 H=h_1-h_2,\quad D=d_1-d_2,\quad
 \Gamma=c_1-c_2,\quad M=m_1-m_2.
\]
All differences have zero initial trace. The tissue equations and maximal
regularity for the signal equation give
\begin{equation}\label{wp:local_differences}
 \|H\|_\infty+\|D\|_\infty\le CT\|\widetilde M\|_\infty,
 \qquad
 \|\Gamma\|_{W^{1,2}_p}\le CT^{1+1/p}\|\widetilde M\|_\infty.
\end{equation}
If $b_i=\chi_0\phi(c_i)\nabla c_i$ and $k_i=\nabla\cdot b_i$, the parabolic
embedding and the bounds on $\phi$ and its derivatives imply
\[
 \|b_1-b_2\|_\infty+\|k_1-k_2\|_p
 \le C\|\Gamma\|_{W^{1,2}_p}.
\]
The difference of the chemotactic fluxes is decomposed as
\begin{align*}
&\nabla\cdot\bigl(\phi(c_1)m_1\nabla c_1
 -\phi(c_2)m_2\nabla c_2\bigr)\\
&=\nabla\cdot\bigl(\phi(c_1)(m_1-m_2)\nabla c_1\bigr)
 +\nabla\cdot\bigl(\phi(c_1)m_2\nabla(c_1-c_2)\bigr)\\
&\quad+\nabla\cdot\bigl((\phi(c_1)-\phi(c_2))m_2\nabla c_2\bigr).
\end{align*}
In particular, the term involving
$[\phi(c_1)-\phi(c_2)]\Delta c_2$ is bounded in $L^p$ by
$C\|\Gamma\|_\infty\|\Delta c_2\|_p$. The macrophage difference satisfies
\[
 M_t-\Delta M=-b_1\cdot\nabla M-(1+k_1)M+F,
\]
where
\[
 F=-(b_1-b_2)\cdot\nabla m_2-(k_1-k_2)m_2+f(c_1)-f(c_2),
 \qquad \|F\|_p\le C\|\Gamma\|_{W^{1,2}_p}.
\]
The zero-trace estimate~\eqref{wp:zero_trace} controls the factor
$\|M\|_\infty$ in the $L^p$ potential term and yields
\begin{equation}\label{wp:map_contraction}
 \|\mathcal S(\widetilde m_1)-\mathcal S(\widetilde m_2)\|_\infty
 \le CT^{1+1/p+\vartheta}
 \|\widetilde m_1-\widetilde m_2\|_\infty,
\end{equation}
for some $\vartheta>0$. Choosing $T$ so that the factor on the right is
less than one proves the contraction claim. The same estimates yield
uniqueness on overlapping intervals. Uniformity on bounded admissible sets
gives the uniform local existence time, and the standard restart argument
gives the continuation criterion in~\eqref{wp:continuation_criterion}.

\subsection{Continuous dependence and classical regularity}\label{app:dependence}

Fix $T>0$ and a bounded admissible set of initial data.  All constants
below are uniform for pairs of solutions issued from this set. For two
such solutions, write $H=h_1-h_2$, $D=d_1-d_2$,
$\Gamma=c_1-c_2$ and $M=m_1-m_2$. On
an interval $I_j=(t_j,t_j+\ell)$, $0<\ell\le1$, set
\[
\varepsilon_j:=\|U_1(t_j)-U_2(t_j)\|_{\mathcal X}.
\]
After translating the time variable, subtract the homogeneous solutions
\[
M_{h,j}(t)=e^{(t-t_j)(\Delta-1)}M(t_j),\qquad
\Gamma_{h,j}(t)=e^{(t-t_j)(D_c\Delta-\mu)}\Gamma(t_j).
\]
Their parabolic, trace and $C_x^{0,1}$ norms on $I_j$ are bounded by
$C\varepsilon_j$ by \eqref{wp:maxreg}.  The functions
$M-M_{h,j}$ and $\Gamma-\Gamma_{h,j}$ have zero trace at $t=t_j$, so
\eqref{wp:zero_trace} applies to them.  Put
\[
\begin{aligned}
A_j&=\|H\|_{L^\infty(I_j\times\Omega)}
    +\|D\|_{L^\infty(I_j\times\Omega)},\\
G_j&=\|\Gamma\|_{W^{1,2}_p(I_j\times\Omega)},\qquad
P_j=\|M\|_{W^{1,2}_p(I_j\times\Omega)}.
\end{aligned}
\]
The coefficient estimates in the local proof, with the nonzero initial
trace retained in the homogeneous terms, give
\[
\begin{aligned}
	A_j&\le C(\varepsilon_j+\ell\|M\|_{L^\infty(I_j\times\Omega)}),&
	G_j&\le C(\varepsilon_j+\ell^{1/p}A_j),\\
	P_j&\le C(\varepsilon_j+G_j),&
	\|M\|_{L^\infty(I_j\times\Omega)}
	&\le C\varepsilon_j+C\ell^\vartheta P_j.
\end{aligned}
\]
In estimating the coefficient differences one uses, more precisely,
\[
\|\Gamma\|_{C_x^{0,1}(\overline{I_j\times\Omega})}
\le C\varepsilon_j+C\ell^\vartheta G_j.
\]
Substitution of the fourth inequality into the first, then into the
second and the third, yields
\[
P_j\le C\varepsilon_j+C\ell^{1+1/p+\vartheta}P_j.
\]
Choose $\ell_0\in(0,1]$, uniformly for the prescribed bounded set of
data, so that $C\ell_0^{1+1/p+\vartheta}\le1/2$.  For
$0<\ell\le\ell_0$, absorption and the preceding four inequalities give
\begin{equation}\label{wp:one_slab_dependence}
P_j+G_j+A_j+\|M\|_{L^\infty(I_j\times\Omega)}
\le C_*\varepsilon_j.
\end{equation}
The trace part of \eqref{wp:maxreg}, applied to $M$ and $\Gamma$, and
the tissue estimate contained in $A_j$ consequently imply
\begin{equation}\label{wp:one_slab_trace}
\sup_{t\in\overline I_j}\|U_1(t)-U_2(t)\|_{\mathcal X}
\le C_*\varepsilon_j.
\end{equation}

Iterate \eqref{wp:one_slab_trace} over a partition of $[0,T]$ into at
most $1+T/\ell_0$ intervals. Summing the $p$th powers of the parabolic
norms in \eqref{wp:one_slab_dependence} gives
\begin{equation}\label{wp:data_dependence}
	\begin{split}
		&\sup_{0\le t\le T}\|U_1(t)-U_2(t)\|_{\mathcal X}
		+\|m_1-m_2\|_{W^{1,2}_p(\Omega_T)}+\|c_1-c_2\|_{W^{1,2}_p(\Omega_T)}\\
		&\hspace{2cm}\le C_T\|U_{1,in}-U_{2,in}\|_{\mathcal X}.
	\end{split}
\end{equation}
The constants are uniform for the chosen bounded initial set by
Theorem~\ref{wp:theorem}; $C_T$ may also depend on $T$.

For continuous tissue data, the ODE integral equations give $h,d\in C^1([0,T];C(\overline\Omega))$. Under the H\"older hypotheses in Theorem~\ref{wp:theorem}, take $0<\vartheta_{\rm reg}<\min\{\eta_0,1-(n+2)/p\}$. The tissue equations give $h,d\in C^{\vartheta_{\rm reg},\vartheta_{\rm reg}/2}$. The Schauder estimates in \citet[Ch.~IV]{LSU88} first apply to $c$, then to $m$ with coefficients $b,k$; the initial Neumann compatibility gives the asserted regularity up to $t=0$. For $C^2$ tissue data, the ODE system in the Banach algebra $C^2(\overline\Omega)$ gives $h,d\in C^1([0,T];C^2(\overline\Omega))$.

\subsection{Tissue evolution: contraction and regularity}\label{app:tissue}

Set $u=(h,d,1-h-d)^T$ and $a=\alpha m$. For almost every $x\in\Omega$,
the tissue evolution is given by~\eqref{wp:Q}. Its evolution matrix
$U_x(t,s)$ is nonnegative with column sums one. If $0\le a\le A$, variation of constants gives, for $P=U_x(s+1,s)$,
\[
P_{11}\ge e^{-A},\qquad P_{12}\ge\rho e^{-(A+\eta)},\qquad
P_{13}\ge\sigma e^{-(A+\sigma)}.
\]
Every entry of the first row is thus at least $\varepsilon=e^{-(A+\eta+\sigma)}\min\{1,\rho,\sigma\}\in(0,1)$. Write $P=\varepsilon e_1\mathbf1^T+(1-\varepsilon)R$, with $R$ nonnegative and column-stochastic. On vectors with zero component sum the first term vanishes and $R$ has $\ell^1$ norm at most one. Iteration over unit intervals, and norm at most one on the remaining interval, give
\begin{equation}\label{wp:contraction}
	|U_x(t,s)v|_1\le C_0e^{-\omega(t-s)}|v|_1,
	\qquad\mathbf1^Tv=0,
\end{equation}
with $\omega=-\log(1-\varepsilon)$ and $C_0=(1-\varepsilon)^{-1}$, uniformly in $x$.

For tissue data in $W^{1,\infty}$, the ODE system in this Banach algebra justifies weak differentiation:
\[
(\partial_j u)_t=Q(a)\partial_j u
+\alpha h\,\partial_jm\,(-1,1,0)^T.
\]
The derivative and forcing have zero component sum. Therefore
\begin{equation}\label{wp:gradient_decay}
	\|\partial_j u(t)\|_{L^\infty(\ell^1)}
	\le C_0e^{-\omega t}\|\partial_j u(0)\|_{L^\infty(\ell^1)}
	+\frac{2\alpha C_0}{\omega}\sup_{s\ge0}\|\partial_jm(s)\|_\infty.
\end{equation}
This proves \eqref{wp:all_gradient_bound}. Continuity of $m$ with values in $W^{1,\infty}$ gives the claimed $C^1$ time regularity of the tissue variables in that space.

\begin{proof}[Proof of Corollary~\ref{wp:higher_tissue_main}]
	Write $\mathbb W=W^{2,p}(\Omega)$.  Since $p>n+2$, in particular
	$p>n$, the space $\mathbb W$ is a Banach algebra and
	$\mathbb W\hookrightarrow W^{1,\infty}(\Omega)$.  Thus there is a
	constant $C_A$ such that
	\begin{equation}\label{wp:W2p_algebra}
		\|fg\|_{\mathbb W}\le C_A\|f\|_{\mathbb W}\|g\|_{\mathbb W}
		\qquad(f,g\in\mathbb W).
	\end{equation}
	The key point is that the tissue equations are ordinary differential
	equations in the Banach algebra $W^{2,p}(\Omega)$, with an operator
	coefficient belonging to $L^p(0,T)$. More precisely, set
	$a(t)=\alpha m(t)$. Since
	$m\in L^p((0,T);W^{2,p}(\Omega))$, we have
	$a\in L^p((0,T);W^{2,p}(\Omega))$. For each $t$, define
	$\mathcal Q(t)$ on $\mathbb W^3$ by
	\[
	\mathcal Q(t)
	\begin{pmatrix}u_1\\u_2\\u_3\end{pmatrix}
	=
	\begin{pmatrix}
		-a(t)u_1+\rho u_2+\sigma u_3\\
		a(t)u_1-\eta u_2\\
		\delta u_2-\sigma u_3
	\end{pmatrix}.
	\]
	The algebra estimate \eqref{wp:W2p_algebra} implies
	\[
	\|\mathcal Q(t)\|_{\mathcal L(\mathbb W^3)}
	\le C\bigl(1+\|a(t)\|_{\mathbb W}\bigr),
	\]
	and hence $t\mapsto\mathcal Q(t)$ belongs to
	$L^p(0,T;\mathcal L(\mathbb W^3))\subset
	L^1(0,T;\mathcal L(\mathbb W^3))$.
	If the initial data satisfy the bound in
	Corollary~\ref{wp:higher_tissue_main}, the same embedding gives
	\[
	\|h_{in}\|_{W^{1,\infty}}+\|d_{in}\|_{W^{1,\infty}}
	\le C_{\mathrm{emb}}R_2.
	\]
	Hence \eqref{wp:all_gradient_bound} applies with
	$R_1=C_{\mathrm{emb}}R_2$; in particular, all first-derivative bounds used
	below depend only on $R$, $R_2$ and the structural data.
	Fix $T>0$. Theorem~\ref{wp:theorem} gives
	$m\in L^p((0,T);\mathbb W)\subset L^1((0,T);\mathbb W)$.
	The vector $u_{\rm in}=(h_{\rm in},d_{\rm in},
	1-h_{\rm in}-d_{\rm in})^T$ belongs to $\mathbb W^3$. The Carathéodory
	theory for linear ODEs in a Banach space with an $L^1$ operator
	coefficient therefore gives a unique solution
	\[
	u\in W^{1,p}(0,T;\mathbb W^3)\cap C([0,T];\mathbb W^3)
	\]
	of $u_t=\mathcal Q(t)u$. Its first two components agree with $(h,d)$
	by uniqueness in the phase space $\mathcal X$. In particular,
	$h,d\in C([0,T];W^{2,p}(\Omega))$ and
	$h_t,d_t\in L^p((0,T);W^{2,p}(\Omega))$, which proves
	\eqref{wp:tissue_second_class}.
	
	We next establish the bound uniform in time.  Taking two weak spatial
	derivatives in \eqref{wp:Q}, with
	$v_{ij}=\partial_i\partial_j u$, gives
	\[
	(v_{ij})_t=Q(a)v_{ij}
	+\alpha\bigl(h\,\partial_{ij}m+\partial_i h\,\partial_jm
	+\partial_jh\,\partial_i m\bigr)(-1,1,0)^T.
	\]
	This identity first holds in $L^p(\Omega)$ for almost every time and may
	then be integrated by variation of constants.  The forcing has zero
	component sum.  Moreover, $0\le h\le1$, while
	Theorem~\ref{wp:theorem}, the embedding
	$E_p\hookrightarrow W^{1,\infty}(\Omega)$, and
	\eqref{wp:all_gradient_bound} give
	\begin{align*}
		&\|h\,\partial_{ij}m
		+\partial_i h\,\partial_jm
		+\partial_jh\,\partial_i m\|_{L^p}\\
		&\hspace{2cm}\le \|\partial_{ij}m\|_{L^p}
		+C_{R,R_2},
	\end{align*}
	where the constant is independent of time.  Applying the pointwise
	contraction \eqref{wp:contraction}, followed by the $L^p(\Omega)$ norm,
	therefore yields
	\[
	\|v_{ij}(t)\|_p\le C_0e^{-\omega t}\|v_{ij}(0)\|_p
	+C\int_0^t e^{-\omega(t-s)}(\|\partial_{ij}m(s)\|_p+1)\,ds.
	\]
	Estimate \eqref{wp:uniform_slabs} applies uniformly on all unit intervals $(s,s+1)$ with $s\ge0$, including $(0,1)$. Thus
	\[
	M_2:=\sup_{s\ge0}
	\|m\|_{L^p((s,s+1);W^{2,p})}\le C_R<\infty.
	\]
	Splitting the convolution into backward intervals of length at most one
	and using H\"older's inequality on each such interval gives
	\[
	\int_0^t e^{-\omega(t-s)}\|\partial_{ij}m(s)\|_p\,ds
	\le\sum_{\ell\ge0}e^{-\omega\ell}M_2
	=\frac{M_2}{1-e^{-\omega}}.
	\]
	The convolution of the constant term is bounded by $1/\omega$.
	Together with the already established zeroth- and first-order bounds,
	this proves \eqref{wp:tissue_second_bound}.  Notice that no
	pointwise-in-time estimate of the second derivatives of $m$ is used.
\end{proof}

\section{From mass decay to strong norms}\label{app:decay}

The following lemma transfers the mass estimates in Section~\ref{wp:extinction_section} to the phase-space norms with the same decay rate.

\begin{lemma}\label{wp:decay_lifting}
	Let $U$ be a solution of Theorem~\ref{wp:theorem}. Suppose
	\[
	\|d(t)\|_1+\|m(t)\|_1+\|c(t)\|_1\le Cw(t),\qquad t\ge0,
	\]
where $w(t)=e^{-\nu_{\rm w}t}$ with $0<\nu_{\rm w}<\eta$, or $w(t)=(1+t)^{-1}$. Then
	\begin{equation}\label{wp:decay_lifted}
		\|d(t)\|_\infty+\|m(t)\|_{E_p}+\|c(t)\|_{E_p}
		+\|m(t)\|_{W^{1,\infty}}+\|c(t)\|_{W^{1,\infty}}
		\le C'w(t).
	\end{equation}
	The constants are uniform on bounded admissible sets whenever the mass-decay constant is uniform on those sets.
\end{lemma}
\begin{proof}
	Suppose the bound holds in $L^a$, $1\le a<\infty$. Choose $a<q\le\infty$ with $\vartheta=1/2+\frac n2(1/a-1/q)<1$. Since $b$ is bounded and $f(c)\le c$, \eqref{wp:m_duhamel} and the heat estimates give, for $t\ge1$,
	\[
	\|m(t)\|_q+\|c(t)\|_q
	\le Cw(t-1)+C\int_{t-1}^t(1+(t-s)^{-\vartheta})w(s)\,ds\le Cw(t).
	\]
	Both weights are comparable to $w(t)$ on $[t-1,t]$. The tissue equation gives
	\begin{equation}\label{wp:d_weighted}
		\|d(t)\|_q\le e^{-\eta t}\|d_{in}\|_q+
		\alpha\int_0^t e^{-\eta(t-s)}\|m(s)\|_q\,ds\le Cw(t).
	\end{equation}
	Here $\nu_{\rm w}<\eta$ handles the exponential weight; splitting at $t/2$ handles the algebraic one. Initial intervals are controlled by Theorem~\ref{wp:theorem}.
	
	Set $1/a_j=1-j/(2n)$, $j=0,\ldots,2n-1$. Iterate with $a=a_j$, $q=a_{j+1}$ for $j=0,\ldots,2n-2$, and then with $a=a_{2n-1}=2n$, $q=\infty$. At every step $\vartheta=3/4$, including $n=1$. This gives the bound in $L^\infty$ for all three components.
	
	For $t\ge2$, use \eqref{wp:cutoff_bound}. The signal source $rd$ and the supremum of $c$ on the last unit interval are $O(w(t))$. In the macrophage equation, $f(c)\le c$ and
	\[
	\|km\|_{L^p(\Omega_{t-1,t})}
	\le\|k\|_{L^p(\Omega_{t-1,t})}\|m\|_{L^\infty(\Omega_{t-1,t})}
	\le Cw(t).
	\]
	The quantitative, homogeneous bound in Lemma~\ref{wp:cutoff_lemma} therefore preserves the same weight in the trace and gradient norms. Enlarging the constant covers $[0,2]$. The global estimates make all constants uniform on the stated sets of data.
\end{proof}


\section{Whole-line well-posedness without chemotaxis}
\label{app:whole_line}

We record the construction used in Theorem~\ref{thm:upper_speed}.  It is
separate from the bounded-domain theory because no chemotactic drift is
present and the spatial domain is $\mathbb R$.

\begin{lemma}[Global whole-line solution]\label{app:whole_line_wp}
	Let $\chi_0=0$ and let
	$U_{in}=(h_{in},d_{in},m_{in},c_{in})\in BUC(\mathbb R)^4$ satisfy
	\[
	 h_{in},d_{in},m_{in},c_{in}\ge0,\qquad h_{in}+d_{in}\le1.
	\]
	Then \eqref{eq:dimensionless} on $\mathbb R$ has a unique global mild
	solution $U=(h,d,m,c)\in C([0,\infty);BUC(\mathbb R)^4)$.  For $t>0$ it
	is classical, with
	\[
	 h,d\in C^1\big((0,\infty);BUC(\mathbb R)\big)
	\]
	and
	\[
	 m,c\in C^1\big((0,\infty);BUC(\mathbb R)\big)
	 \cap C\big((0,\infty);C_b^2(\mathbb R)\big).
	\]
	Here $BUC(\mathbb R)$ denotes the Banach space of bounded uniformly
	continuous functions, while $C_b^2(\mathbb R)$ denotes the functions with
	bounded continuous derivatives up to order two.
	Moreover, for every $t\ge0$, the components $m(t,\cdot)$ and $c(t,\cdot)$ are nonnegative, and the tissue variables satisfy
	\begin{equation}\label{app:whole_line_bounds}
	 0\le h(t),d(t),\qquad h(t)+d(t)\le1,
	\end{equation}
	and
	\begin{equation}\label{app:whole_line_mc_bounds}
	 \|c(t)\|_\infty\le
	 \max\!\left\{\|c_{in}\|_\infty,\frac r\mu\right\},
	 \qquad
	 \|m(t)\|_\infty\le
	 \max\!\left\{\|m_{in}\|_\infty,1\right\}.
	\end{equation}
\end{lemma}

\begin{proof}
	Let $S_D(t)$ denote the heat semigroup with diffusivity $D>0$ on
	$BUC(\mathbb R)$.  It is a positive contraction $C_0$-semigroup.  On
	$\mathbb X:=BUC(\mathbb R)^4$, replace $f$ temporarily by the
	non-negative, locally Lipschitz extension
	\[
	 \widehat f(z)=0\quad(z\le0),\qquad
	 \widehat f(z)=\frac{z}{1+z}\quad(z\ge0).
	\]
	Write the two tissue equations in integral form and the two parabolic
	equations as
	\begin{align*}
	 m(t)&=e^{-t}S_1(t)m_{in}
	       +\int_0^t e^{-(t-s)}S_1(t-s)\widehat f(c(s))\,ds,\\
	 c(t)&=e^{-\mu t}S_{D_c}(t)c_{in}
	       +r\int_0^t e^{-\mu(t-s)}S_{D_c}(t-s)d(s)\,ds,
	\end{align*}
	The resulting Nemytskii map, together
	with the polynomial tissue reactions, is locally Lipschitz on bounded
	subsets of $\mathbb X$.  The four integral equations therefore define a
	contraction on a sufficiently small closed ball in
	$C([0,T];\mathbb X)$ when $T>0$ is small.  This gives a unique maximal
	mild solution and the usual blow-up alternative.  This is the standard
	semilinear-semigroup construction; see, for example,
	\citet[Chapter~6, Section~1]{Pazy1983}.

The positive cone and the tissue simplex
\[
\Sigma=\{(h,d):h\ge0,\ d\ge0,\ h+d\le1\}
\]
are invariant. Indeed, positivity of the heat semigroup and
$\widehat f\ge0$ give $m\ge0$. For each fixed $x\in\mathbb R$, the tissue
equations form an ODE system with continuous nonnegative coefficient
$m(t,x)$. The pointwise invariance argument used in the proof of
Lemma~\ref{wp:lemma_region} therefore applies and gives
\eqref{app:whole_line_bounds}. In particular, $d\ge0$; the mild formula
for $c$ then gives $c\ge0$. Consequently, $\widehat f(c)=f(c)$ along the
solution, so the constructed solution solves the original system.

Since $0\le d\le1$, comparison in the signal equation with the constant supersolution
	$C_*:=\max\{\|c_{in}\|_\infty,r/\mu\}$ gives
	$0\le c(t,x)\le C_*$.  Similarly, $0\le f(c)\le1$, and comparison in the
	macrophage equation with
	$M_*:=\max\{\|m_{in}\|_\infty,1\}$ gives
	$0\le m(t,x)\le M_*$.  Thus every component remains bounded on each
	finite time interval, uniformly up to the maximal existence time.  The
	blow-up alternative implies that this time is infinite and proves
	\eqref{app:whole_line_mc_bounds}.

	The tissue integral equations and the continuity of $h,d,m$ in $BUC$
	first give $h,d$ $\in $ \\ $C^1([0,\infty);BUC(\mathbb R))$. In particular,
	$rd$ is locally Lipschitz in time with values in $BUC(\mathbb R)$.
	To make the parabolic regularity step explicit, let $A$ be the generator
	of either damped heat semigroup $T(t)$ and let $g$ be locally Lipschitz
	in time with values in $BUC(\mathbb R)$. On an interval $a<t$, the
	inhomogeneous term satisfies
	\[
	 A\!\int_a^t T(t-s)g(s)\,ds
	 =(T(t-a)-I)g(t)
	 +\int_a^t AT(t-s)\bigl(g(s)-g(t)\bigr)\,ds.
	\]
	Indeed, $\|AT(\tau)\|\le C/\tau$, while
	$\|g(s)-g(t)\|=O(t-s)$ locally; the singularity at $s=t$ is therefore
	integrable, and the identity follows by closedness of $A$. Together with
	analytic-semigroup regularization, this yields continuous dependence in
	the graph norm and continuous time differentiability for positive times;
	see \citet[Chapter~4, Section~3, and Chapter~6, Section~3]{Pazy1983}.
	Applying this argument first to $c$ with $g=rd$ shows that $c$ is locally
	Lipschitz in time on $(0,\infty)$ with values in $BUC$. Since $f$ is
	Lipschitz on $[0,\infty)$, the same is true of $f(c)$ on every compact
	interval in $(0,\infty)$. The argument then applies to $m$ starting at
	any $a>0$. Thus $m,c$ are continuously differentiable in time
	with values in $BUC$ and continuous in the graph norms of their respective
	generators for $t>0$.

	For these generators, membership in the domain implies that the second
	spatial derivative belongs to $BUC$; the one-dimensional Landau inequality
	controls the first derivative by the function and its second derivative.
	Graph-norm continuity therefore gives
	$m,c\in C\bigl((0,\infty);C_b^2(\mathbb R)\bigr)$, as claimed.
\end{proof}
\bmhead{Acknowledgements}
This work has been partially supported 
by the European Union – Next Generation EU: PRIN project
MUR PrinPNRR 2022, project code 
P202254HT8, Grant No. CUP B53D23027760001; by the PRIN project MUR PrinPNRR 2022, project code P2022Z7ZAJ, Grant No. CUP B53D23027930001 and by  
the PRIN project MIUR Prin 2022, project code
1074 2022M9BKBC, Grant No. CUP B53D23009350006.
The authors also gratefully acknowledge the financial support of GNFM-INdAM.
\section*{Declarations}
\begin{itemize}
	\item Funding: Open access funding provided by University of Palermo.
	\item Conflict of interest: The authors declare that they have no known competing financial
	interests or personal relationships that could have appeared to
	influence the work reported in this paper.
	\item Data availability: The manuscript has no associated data.
\end{itemize}
\bibliographystyle{unsrtnat}
\bibliography{distrofia_bib}
\end{document}